\documentclass[reprint,onecolumn,aip,jcp]{revtex4-1}%
\usepackage[utf8]{inputenc}
\usepackage[T1]{fontenc}
\usepackage{algorithm}
\usepackage{xr}
\usepackage{algpseudocode}
\usepackage{stmaryrd}
\usepackage{amsmath}
\usepackage{amssymb}
\usepackage{amsthm}
\usepackage{bbm}
\usepackage{graphicx}
\usepackage{ifpdf}
\ifpdf
  \DeclareGraphicsExtensions{.pdf,.png,.jpg}
\else
  \DeclareGraphicsExtensions{.eps}
\fi
\usepackage{color}

\DeclareMathAlphabet{\mathbbo}{U}{bbold}{m}{n}
\newcommand{\dspi}{\mathbbo{n}}
\newcommand{\dsg}{\mathbbo{g}}
\newcommand{\g}{g}
\newcommand{\dsh}{\mathbbo{h}}
\newcommand{\dsw}{\mathbbo{r}}

\newcommand{\w}{r}

\newcommand{\varlambda}{\zeta}
\newcommand{\estimator}[2]{\Upsilon^{#1}_{\mathrm{#2}}}

\newcommand{\estimate}[2]{\widehat{\mathpzc{#1}}^M_\mathrm{#2}}

\DeclareFontFamily{OT1}{pzc}{}
\DeclareFontShape{OT1}{pzc}{m}{it}{<-> s * [1.05] pzcmi7t}{}
\DeclareMathAlphabet{\mathpzc}{OT1}{pzc}{m}{it}

\DeclareFontFamily{OT1}{pzcm}{}
\DeclareFontShape{OT1}{pzcm}{m}{it}{<-> s * [1.35] pzcmi7t}{}
\DeclareMathAlphabet{\mathpzcm}{OT1}{pzcm}{m}{it}
\DeclareMathAlphabet{\mathantt}{OT1}{antt}{m}{it}

\newcommand{\pot}{\mathpzc{U}}

\newcommand{\n}{\Psi^\varLambda}
\newcommand{\na}{\Psi^\varLambda_\aux}

\newcommand{\aux}{\mathantt{a}}

\newcommand{\lpaux}{\mathrm{p}^{\varLambda}_\aux}
\newcommand{\lpo}{\mathrm{p}^\varLambda_0}
\newcommand{\1}{\mathbf{1}}
\newcommand{\2}{\mathbf{1}}

\newcommand{\pmf}{\mathpzc{A}}
\newcommand{\wpmf}{\widehat{\mathpzc{A}}}
\newcommand{\pmff}{\mathpzc{F}}
\newcommand{\obs}{\mathcal{O}}

\newcommand{\p}{\mathrm{p}}
\newcommand{\paux}{{\p}_{\aux}}
\newcommand{\bpaux}{\p^{\mathcal{Q}}_{\aux}}

\newcommand{\bpi}{\bar{\pi}}

\newcommand{\sigmaux}{\sigma_\aux}

\newcommand{\e}{\mathbbm{E}}
\newcommand{\eaux}{\e_{\aux}}
\newcommand{\beaux}{\e^{\mathcal{Q}}_{\aux}}
\newcommand{\leaux}{\e^{\varLambda}_{\aux}}
\newcommand{\raux}{\mathpzc{B}_{\aux}}
\newcommand{\ro}{\mathpzc{B}_{0}}
\newcommand{\restrain}{\mathpzc{R}}
\newcommand{\rc}{q}
\newcommand{\brc}{{\bar{q}}}
\newcommand{\rcm}{\mathrm{\rc}}
\newcommand{\xis}{{\xi^\star}}
\newcommand{\rf}{\mathbbo{\w}^\lambda_\aux}

\newcommand{\lambdag}{{\lambda^\star}}
\newcommand{\convlaw}{\overset{\mathcal{L}}{\underset{M \rightarrow + \infty}{\longrightarrow}}}

\newcommand{\cov}{\mathrm{cov}_\aux}
\newcommand{\vara}{\mathbbm{V}_\aux}
\newcommand{\iam}{\mathtt{I}^M}

\newcommand{\cardTheta}{\| \varLambda \|}

\newtheorem{theorem}{Theorem}[section]

\newtheorem{lemma}[theorem]{Lemma}

\newcommand{\ct}[1]{\textcolor{black}{#1}}
\newcommand{\trc}{{\rc}}
\begin{document}

\title{Estimating thermodynamic expectations and free energies in expanded ensemble simulations: systematic variance reduction through conditioning}

\author{Manuel Ath\`enes}
\email{manuel.athenes@cea.fr}
\affiliation{CEA, DEN, Service de Recherches de M\'etallurgie Physique, Universit\'e Paris-Saclay, F-91191 Gif-sur-Yvette, France}
\author{Pierre Terrier}
\email{pierre.terrier@cea.fr}
\affiliation{CEA, DEN, Service de Recherches de M\'etallurgie Physique, Universit\'e Paris-Saclay, F-91191 Gif-sur-Yvette, France}
\affiliation{Universit\'e Paris-Est, CERMICS (ENPC), INRIA, F-77455 Marne-la-Vall\'ee, France}
\begin{abstract}
Markov chain Monte Carlo methods are primarily used for sampling from a given probability distribution and estimating multi-dimensional integrals based on the information contained in the generated samples. Whenever it is possible, more accurate estimates are obtained by combining Monte Carlo integration and integration by numerical quadrature along particular coordinates.  
We show that this variance reduction technique, referred to as conditioning in probability theory, can be advantageously implemented in \emph{expanded ensemble} simulations. These simulations aim at estimating thermodynamic expectations as a function of an external parameter that is sampled like an additional coordinate. Conditioning therein entails integrating along the external coordinate by numerical quadrature. We prove variance reduction with respect to alternative standard estimators and demonstrate the practical efficiency of the technique by estimating free energies and characterizing a structural phase transition between two solid phases. 
\end{abstract}

\maketitle

\section{Introduction}
One important goal of molecular simulation is the estimation of equilibrium thermodynamic quantities. Any such quantity corresponds to the ensemble average of an observable 
over the space of accessible states $\mathcal{Q}$ of a multi-dimensional system. The contribution of any given state to the ensemble average is weighted by its occurrence probability in the considered thermodynamic ensemble (the Gibbs-Boltzmann probability). If one knows how to sample from this probability distribution, for instance  by using a Markov chain Monte Carlo method or a molecular dynamics algorithm, then any thermodynamic quantity can be accurately estimated in practice by taking advantage of the ergodic theorem.~\cite{chandler:1987,frenkel:2001,wales:2003,lelievre:2010} 

A major bottleneck facing molecular simulation is broken numerical ergodicity, 
i.e., insufficient sampling of important configurations occurring during rare excursions or inability to explore important regions because of kinetic traps~\cite{athenes:2014} on the simulation timescale. 

A versatile and widespread solution to alleviate this problem consists of sampling an expanded ensemble~\cite{lyubartsev:1992,burov:2006} that is a mixture of canonical ensembles equipped with an auxiliary biasing potential $\aux(\lambda)$. The subensemble indexes are referred to by an external parameter $\lambda$ which corresponds to extensive quantities (such as the volume or the numbers of particles) or intensive quantities (thermodynamic forces such as pressure, chemical potentials or inverse temperature). The  method of expanded ensembles~\cite{lyubartsev:1992,burov:2006} also includes simulated tempering,~\cite{marinari:1992,geyer:1995,park:2007} simulated scaling~\cite{li:2007}, unified free energy dynamics~\cite{cuendet:2014} and extended Lagrangian metadynamics~\cite{iannuzzi:2003,laio:2008} methods. 
In all these implementations, the external parameter behaves like an additional coordinate of the system. 
Improved numerical ergodicity is usually achieved compared to direct Monte Carlo or molecular dynamics simulations by allowing the system to circumvent barriers between metastable basins or even to cross them when, for instance, the external parameter is associated with inverse temperature and takes small values. 
Because of its popularity and versatility, the expanded ensemble method has been the subject of numerous studies~\cite{iba:2001,mitsutake:2001,chodera:2007,park:2008,chodera:2011,cao:2014,terrier:2015} over recent years. To evaluate the thermodynamic expectations, several estimators have been proposed, namely, the (standard) binning estimator,~\cite{lyubartsev:1992} a standard reweighting estimator~\cite{terrier:2015}, a self-consistent reweighting estimator~\cite{chodera:2007} called \textit{simulated tempering weighted histogram analysis method} (STWHAM), a histogram reweighting estimator~\cite{lesage:2016} called \textit{corrected z-averaged restraints} (CZAR) and an \textit{adiabatic reweighting} (AR) estimator.~\cite{cao:2014,terrier:2015}
Self-consistent reweighting estimators are known under various names such as the \textit{Bennett acceptance ratio} (BAR) method~\cite{bennett:1976}, the \textit{wheighted histogram analysis method}~\cite{swendsen:1986,ferrenberg:1989} (WHAM), the reverse logistic regression~\cite{geyer:1994}, bridge sampling,~\cite{gelman:1998} the multi-state BAR method~\cite{shirts:2008}, binless WHAM~\cite{tan:2012} and the global likelihood method.~\cite{kong:2003,tan:2004} \ct{These various methods are equivalent in the sense that they involve first harvesting data from a number of independent simulations carried out for a predetermined set of external parameter values, possibly using the expanded ensemble method, and second solving a set of nonlinear equations.}

In contrast, the binning, standard reweighting and adiabatic reweighting estimator do not require any postprocessing and can be implemented online during a single simulation, provided the biasing potential ensures ergodic sampling in the expanded ensemble. This one is usually constructed adaptively during the course of the simulation in such a way that uniform sampling is achieved along the external parameter. A similar proceeding is done in multicanonical sampling methods,~\cite{berg:1992,wang:2001,laio:2002,junghans:2014} for which adequate biasing factors are associated with a generalized coordinate rather than an external parameter. The biasing factors are then set equal to the inverse of the density of states with respect to the associated coordinate that can be a reaction coordinate~\cite{darve:2001} or an order parameter,~\cite{berg:1992} depending on whether one is interested in monitoring a reaction along its transition pathway or probing the various structures of a system.  A typical generalized coordinate is also the internal energy of the system.~\cite{wang:2001}

The biasing factors may be constructed using various techniques that can be classified into two categories: \textit{adaptive biasing potential} (ABP) methods \---- along internal~\cite{wang:2001,marsili:2006,fort:2014} or external~\cite{brukhno:1996,brukhno:1996b} coordinates\---- and \textit{adaptive biasing force} (ABF) methods \---- along internal~\cite{darve:2001,darve:2008,henin:2010,lelievre:2008,comer:2014} or external~\cite{park:2007,cao:2014,lesage:2016} coordinates\----,  depending on whether it is the auxiliary biasing potential that is updated or its gradient. 
A review of ABP and ABF techniques can be found in Ref.~\citenum{lelievre:2010}. 
A detailed proof of convergence is given in Ref.~\citenum{fort:2014} for the the ABP algorithm~\cite{wang:2001} proposed by Wang and Landau and in Ref.~\citenum{lelievre:2008} for the ABF algorithm of Darve and Pohorille. 
\ct{The  biasing potential may be adapted on the free energy or a part of it. The latter variant technique,~\cite{fort:2016} referred to as \textit{partial biasing}, may accelerate the adaptive sampling process in some circumstances.} 
Besides, Tan~\cite{tan:2015} showed how to reduce the asymptotic statistical variance of an ABP scheme~\cite{footnote1} dedicated to expanded ensemble simulations through a procedure known as \emph{conditioning} in probability theory.~\cite{jourdain:2013,mcbook}  In expanded ensembles, the procedure consists of combining Monte Carlo integration over the internal coordinate space and integration by numerical quadrature along the external parameter. 

\ct{Considerable effort has been devoted to the development of efficient ABP and ABF procedures for expanded ensemble simulations, yet relatively few studies~\cite{tan:2015,fort:2016} have focused on optimizing the biasing potential or the estimator. 
Optimizing the biasing potential is a difficult problem involving functional minimization~\cite{pham:2012} of the variances associated with an estimator that is also to be optimized. We herein address the easier problem of optimizing the estimator given the biasing potential by exploiting the conditioning approach proposed by Tan in Ref.~\citenum{tan:2015}. A \textit{partial biasing} technique for constructing a biasing potential that is optimal for the optimized estimator will be presented in a separate paper. We therefore consider that the adaptive sampling process has been successfully completed: the biasing potential is frozen and does not depend on the simulation time anymore.  We then show how to condition the binning,\cite{lyubartsev:1992} standard reweighting,\cite{terrier:2015} self-consistent reweighting~\cite{chodera:2007} and histogram reweighting~\cite{lesage:2016} estimators of thermodynamic expectations, and demonstrate that the conditioning procedure systematically leads to the formulation of the adiabatic reweighting estimator. By comparing the asymptotic variances of the involved estimators, we eventually deduce that the conditioning procedure ensures systematic variance reduction, entailing that the adiabatic reweighting estimator is optimal among the large class of considered estimators.}

The paper is organized as follows: in section~\ref{section:expectations}, we define the expanded ensemble and show that thermodynamic ensemble averages can be cast as conditional expectations of the observable given the value of the external coordinate and that free energies can be expressed as the co-logarithm of a total expectation. In Section~\ref{section:conditioning}, the adiabatic reweighting estimator is constructed through the conditioning of a generic estimator that covers the binning and standard reweighting approaches as particular cases and also of the self-consistent reweighting estimator STWHAM. 
In both situations, the systematic reduction of the asymptotic variances is proved using the law of total variance.
In section~\ref{sec:assessment}, we assess the variance reduction on free energy and rare event simulations of a simple model. 
In Sections~\ref{application:lj38} and~\ref{sec:characterization}, we implement the conditioning approach to estimate the free energy along a bond-orientational order parameter, making a comparison with the CZAR approach, and to determine a structural transition temperature between two solid phases in an atomic cluster, respectively. The specific ABF algorithms that are used to construct the biasing potentials in Sections~\ref{application:lj38} and~\ref{sec:characterization} are described in the \textcolor{blue}{Supplementary Material}. 

\section{Thermodynamic expectations in expanded ensemble}
\label{section:expectations}
\subsection{Preliminary definitions}
Let first define the thermodynamic expectations and free energies along $\lambda$ in an expanded ensemble. The external parameter $\lambda$ is often associated with a thermodynamic force (e.g. temperature, pressure or averaged mechanical restraint) that, for finite systems, varies continuously with respect to its thermodynamic conjugate. However, we consider here that the external coordinate takes discretized values inside a finite set denoted by $\varLambda$, over which appropriate summations will correctly approximate integrations along $\lambda$ (the aforementioned numerical quadrature). Systematic discretization errors are negligible. We assume that the auxiliary biasing potential $\aux(\lambda)$ has already been adapted using one of the aforementioned ABP or ABF techniques and does not evolve anymore during the simulations (time homogeneity). We denote the particle position coordinates by $\rc \in \mathcal{Q}$ where $\mathcal{Q}$ is the position space. 
We define the dimensionless potential energy in the expanded space $\varLambda\times \mathcal{Q}$ as $\pot (\lambda,q)$  
and  the probability density within the expanded ensemble as
\begin{equation}
\mathrm{p}_\aux(\lambda,\rc) = \exp \left[\aux(\lambda) - \pot(\lambda,\rc) - \na \right], 
\end{equation}
where the log-normalizing constant $\na$ corresponds to the logarithm of the partition function in the expanded ensemble: 
\begin{equation}
 \na = \ln \sum_{\varlambda\in\varLambda}\int_{\mathcal{Q}} \exp\left[\aux(\varlambda) - \pot(\varlambda,\brc) \right] d\brc. 
\end{equation}
The constant $\na$ is sometimes referred to as the entropic potential. By convention, we denote the extended coordinates by $\left( \varlambda,\brc\right)$  when they are summed, integrated or sampled. Herein, the observable $(\lambda,\rc) \mapsto \obs(\lambda,\rc)$ will depend both on the external parameter $\lambda\in\varLambda$ and the internal coordinates $\rc\in \mathcal{Q}$. 
The total expectation associated with the observable is given by: 
\begin{equation}
 \eaux\left[\obs \right] = \sum_{\varlambda\in\varLambda}\int_\mathcal{Q} \obs(\varlambda,\brc) \mathrm{p}_\aux(\varlambda,\brc)d\brc. 
\end{equation}
For simplicity, we omit the dependence of $\obs$ on the stochastic variables inside the expectations. 
Adopting Landau's definition, the free energy $\pmf(\lambda)$ along the external parameter is defined as the co-logarithm of the expected value of the indicator function $\1_\lambda(\varlambda)$: 
\begin{equation}\label{def:fe}
 \pmf(\lambda) = - \ln \e \left[\1_\lambda \right], 
\end{equation}
where the notation $\e$ corresponds to expectation $\eaux$ with biasing potential $\aux(\lambda)$ set to 0 (unbiased ensemble):
\begin{equation}\label{expectation:total}
  \e \left[\obs \right] = \sum_{\varlambda\in\varLambda}\int_\mathcal{Q} \obs(\varlambda,\brc) \exp \left[ -\pot(\varlambda,\brc) - \n_0 \right] d\brc. 
\end{equation}
The free energy allows us to formally define the conditional probability of $\rc$ given that $\varlambda$ is equal to $\lambda$ 
\begin{equation}\nonumber
\pi(\rc |\lambda) = \frac{\p_0(\lambda,\rc)}{\e\left[\1_\lambda \right]} = \exp \left[ \pmf(\lambda)-\pot(\lambda,\rc) - \n_0  \right], 
\end{equation}
as well as the conditional probability 
\begin{eqnarray}
   \e\big[ \obs |\lambda \big] =\int_\mathcal{Q} \obs(\lambda,\brc) \pi(\brc |\lambda)d\brc. 
\end{eqnarray}
The fact that the relation $\pi(\rc |\lambda) = \paux(\lambda,\rc) \slash \eaux\left[\1_\lambda \right]$ is valid whatever the value of the biasing potential entails the useful relation between total and conditional expectations
\begin{equation}
 \e\big[ \obs|\lambda \big] = \frac{\eaux \left[ \1_\lambda(\cdot) \obs(\lambda,\cdot) \right]}{\eaux\left[\1_\lambda\right]}. 
\end{equation}
For clarity, the constant variable $\lambda$ is written explicitly within the observable expectations $\eaux$, while the stochastic variable associated with $\brc$ is denoted by a dot.  
The denominator of the fraction corresponds to the marginal probability of $\lambda$, 
since we have
\begin{eqnarray}\label{def:histogram_lambda}
 \lpaux(\lambda) &=& \int_\mathcal{Q} \paux(\lambda,\brc)d\brc = \eaux\left[\1_\lambda \right]. 
\end{eqnarray}
Similarly, the marginal probability of $\rc$ is equal to the expected values of the delta distribution $\delta_{\rc}(\brc)$ 
\begin{eqnarray}
  \bpaux(\rc) &=& \sum_{\varlambda\in\varLambda} \paux(\varlambda,\rc)  = \eaux\left[\delta_{\rc} \right]. 
\end{eqnarray}
Let denote the expectations defined with respect to the marginal probabilities $\lpaux$ and $\bpaux$ by $\leaux$ and $\beaux$, respectively. These expectations are useful for expressing the laws of total expectations, with conditioning done with respect to  $\lambda$: 
\begin{eqnarray} \label{law:total_expectation_lambda}
 \begin{aligned}
 \eaux\left[\obs \right] &= \leaux \left[ \eaux \left[ \obs(\lambda,\cdot) | \lambda \right] \right]  \\
 &=\sum_{\lambda\in\varLambda} \eaux\left[  \obs(\lambda,\cdot)| \lambda \right]  \lpaux(\lambda), 
 \end{aligned}
\end{eqnarray}
or with respect to $q$: 
\begin{eqnarray}\label{def:marg_expectation}
 \begin{aligned}
 \eaux\left[\obs \right] & = \beaux \left[ \eaux \left[ \obs(\cdot,\rc) | \rc \right] \big. \right] \\
 & = \int_{\mathcal{Q}} \eaux\left[  \obs(\cdot,\rc)| \rc \right]  \bpaux(\rc)d\rc. 
 \end{aligned}
\end{eqnarray}
By convention, the variables $\lambda$ and $\rc$ are stochastic within expectations $\beaux$ and $\leaux$, respectively, but both refers to constant values within total expectations of type $\eaux$. 
As an illustration, when the observable is set to $\2_\xis [\xi(q)]$, the indicator function associated with reaction coordinate $\brc \mapsto \xi(\brc)$ and bin $\xis$, then the conditional expectation of the indicator function given $q$ writes 
\begin{eqnarray}
\eaux\left[ \2_{\xis}\circ\xi(q) | q \right] &= \2_{\xis}\circ\xi(q)
\end{eqnarray}
where the empty symbol $\circ$ denotes the functional composition. The expected value of $\2_\xis\circ\xi$ with respect to the marginal probability distribution $\bpaux$ writes: 
\begin{eqnarray}\label{def:conditional_expectation_given_q}
\beaux\left[ \2_{\xis}\circ\xi(q) \right] = \beaux\left[ \2_{\xis}\circ\xi \right] = \eaux\left[ \2_{\xis}\circ\xi\right] . 
\end{eqnarray}
We eventually formulate the law of total variance, an elementary identity of probability theory, which states that the (total) variance of $\obs$ with respect to probability $\paux$ is equal to the sum of two terms, the expectation of the conditional variances of $\obs$ given $q$ and the variance of the conditional expectation of $\obs$ given $q$: 
\begin{equation}\label{law:total_variance_elementary}
 \vara [\obs] = \beaux \Big[ \vara[ \obs(\cdot,q)|q]\Big] + \vara^\mathcal{Q} \Big[ \eaux  \big( \obs(\cdot,q)| q \big) \Big]. 
\end{equation}  
This identity can be verified by expressing the involved variances as functions of their expectation:
\begin{eqnarray}\label{def:variance}
\vara\left[ \obs \right] &=& \eaux \left[\obs^2\right] - \eaux \left[\obs\right]^2, \\
\vara\left[ \obs | \rc \right] &=& \eaux \left[\left.\obs^2 \right| \rc \right] - \eaux \left[\obs| \rc \right]^2. \label{def:conditional_variance}
\end{eqnarray} 
We similarly define $\vara^\mathcal{Q}$, the variance with respect to probability $\bpaux$,  by replacing expectations $\eaux$ by expectations $\beaux$ in Eq.~\eqref{def:variance} and~\eqref{def:conditional_variance}. 
The law of total variance~\eqref{law:total_variance_elementary} will be useful in Sec.~\ref{section:conditioning} for proving variance reduction through conditioning of various estimators. However, estimators are used in practice to estimate conditional expectations of type $\e\left[ \obs|\lambda\right]$ which are defined with respect to the unbiased probability (reference distribution without biasing potential) whereas states are sampled from the biased probability distribution (with biasing potential switched on). This entails that the contribution of the sampled states must be weighted by including an adequate weighing function $g_\aux^\lambda$, so as to correct for the biasing potential.

\subsection{Estimating conditional expectations}

The generic estimator of the conditional expectation of observable $\obs(\lambda,\brc)$ given $\lambda$ will be based on the following generic weighing function 
\begin{equation}
\g_\aux^\lambda (\varlambda,\brc) =\frac{\paux(\lambda,\brc)}{\paux(\varlambda,\brc)}  K_{\lambda \varlambda} , 
\end{equation}
where the matrix $K_{\lambda \varlambda}$ satisfies the condition
\begin{equation}\label{def:kernel}
 \sum\nolimits_{\varlambda\in\varLambda} K_{\lambda \varlambda} =1, \quad \forall \lambda \in \varLambda. 
\end{equation}
The generic weighing function thus defines a large class of estimators which the conditioning approach will apply to. Owing to normalization~\eqref{def:kernel}, the expectation of $\g^\lambda_\aux$ is always equal to the marginal probability of the external parameter:  
\begin{equation}\label{scaling}
 \e_{\aux} \big[  \g_\aux^\lambda\big] = \sum_{\varlambda\in\varLambda}\int_ \mathcal{Q} \g_\aux^\lambda(\varlambda,\brc) \mathrm{p}_{\aux}(\varlambda,\brc)d\brc = \lpaux(\lambda). 
\end{equation}
Similarly, the following relation holds whatever the generic function: 
\begin{eqnarray*}
 \e_{\aux} \big[\g_\aux^\lambda (\cdot,\cdot) \obs(\lambda,\cdot) \big] & = & \int_\mathcal{Q} \obs(\lambda,\brc) \mathrm{p}_{\aux}(\lambda,\brc)d\brc \\ 
 & = & \e\big[\obs(\lambda,\cdot)\big| \lambda \big] \times \lpaux(\lambda). 
\end{eqnarray*}
This relation together with Eq.~\eqref{scaling} allows to express the conditional expectation of $\obs$ given $\lambda$ as a function of $\g_\aux^\lambda$
\begin{equation}\label{expectation:generic_form}
 \e\big[\obs\big| \lambda \big] = \frac{\eaux \left[ \g^\lambda_{\aux}(\cdot,\cdot) \obs(\lambda,\cdot) \right]}{\eaux \left[ \g^\lambda_{\aux} \right]}, 
\end{equation}
and to obtain the generic estimators by application of the ergodic theorem 
\begin{equation}\label{estimator:generic}
 \estimator{M}{G} (\obs |\lambda) = \frac{\tfrac{1}{M} \sum_{m=1}^M \g_\aux^\lambda(\varlambda_m,\brc_m)\obs(\lambda,\brc_m) }{\tfrac{1}{M} \sum_{m=1}^M \g_\aux^\lambda(\varlambda_m,\brc_m)},   
\end{equation}
where we considered a Markov chain $\left\{\varlambda_m,\brc_m \right\}_{1\leq m \leq M}$ ergodic with respect to $\paux(\varlambda,\brc)$. 

The generic weighing estimator generalizes the binning estimator~\cite{lyubartsev:1992}  and the standard reweighting estimator.~\cite{terrier:2015}
The binning estimator is obtained by setting the kernel matrix $K_{\lambda \varlambda}$ to the identity matrix. The function $g^\lambda_\aux(\varlambda,\brc) $ then reads 
\begin{equation}
h^\lambda_\aux(\varlambda,\brc) \triangleq \1_{\lambda} (\varlambda)
\end{equation} 
and we denote the corresponding estimator by $\estimator{M}{H}(\obs|\lambda)$. 
The standard reweighting estimator of $\e \left[ \obs| \lambda\right]$, denoted by $\estimator{M}{R}(\obs|\lambda)$, is obtained by setting the entries of the kernel matrix to the inverse of $\cardTheta = \sum_{\varlambda \in \varLambda}1$, the cardinal of the discrete set $\varLambda$. Unlike the binning estimator $\estimator{M}{H}(\obs|\lambda)$, the standard reweighting estimator $\estimator{M}{R}(\obs|\lambda)$ includes information from all the sampled subensembles $(\varlambda\in\varLambda)$ owing to the standard reweighting function, 
\begin{equation}
 \w_\aux^\lambda(\varlambda,\brc) = \frac{1}{\cardTheta} \dfrac{\exp\left[ \aux(\lambda)-\pot(\lambda,\brc)\right] } {\exp \left[\aux(\varlambda)-\pot(\varlambda,\brc) \right]}. 
\end{equation}
Thus, the factors $\w_\aux^\lambda(\varlambda_m,\brc_m)$ must be used in place of the generic weighing factors  $g_\aux^\lambda(\varlambda_m,\brc_m)$ in~\eqref{estimator:generic}. 
This standard reweighting function with $\cardTheta=1$ is commonly used is in non-Boltzmann sampling or free energy perturbation techniques~\cite{zwanzig:1954,torrie:1977,frenkel:2001,lelievre:2010} to rescale the contribution of states sampled with probability $\brc \mapsto \pi(\brc | \varlambda)$
with respect to target probability $\brc \mapsto \pi(\brc | \lambda)$ inside the estimator.
The forthcoming derivations involving the generic estimator will cover binning and standard reweighting as two particular remarkable cases. 

\subsection{Coupling protocols}

Within the expanded ensemble, we consider here that the external parameter couples either linearly to a potential energy of the internal coordinates or harmonically to a reaction coordinate $\xi(\rc)$. Note that nonlinear couplings between the external parameter and the potential energy can possibly be employed, for instance to achieve improved efficiency through functional minimization.~\cite{pham:2012}

With linear coupling, the extended system evolves between a reference system $S_0$ and a target system $S_1$. In practice, it is convenient to write the extended potential as 
\begin{equation} \label{def:alchemical}
 \pot (\lambda,\rc) =  (1-\lambda)\pot_0(\rc) + \lambda \pot_1(\rc), 
\end{equation}
where $\lambda$ ranges in the interval $[0,1]$. In~Eq.~\eqref{def:alchemical}, $\pot_0(\rc)$ and $\pot_1(\rc)$ are the potentials of the reference and target systems, respectively. 
This parameterization covers situations where the external parameter is associated with inverse temperature or pressure. For instance, when the reference and target systems correspond to a system of same potential energy $\mathpzc{V}(\rc)$ held at two distinct temperatures $\beta_\mathrm{min}$ and  $\beta_\mathrm{max}$, we simply set  $\pot_0 = \beta_\mathrm{min} \mathpzc{V}$ and $ \pot_1 = \beta_\mathrm{max} \mathpzc{V}$. 
An important task is to estimate conditional expectations given some values of the external parameter $\lambda$ [see~Eq.~\eqref{expectation:generic_form}] using an estimator described in subsection~\ref{estimators:conditional_expectation}. Estimating the average internal energy at a given inverse temperature is a typical example that will be illustrated in Section~\ref{sec:characterization} for icosahedral and octahedral atomic clusters. 
The present linear setup is also used in Section~\ref{sec:assessment} to illustrate the  estimation of total expectations, free energies and rare-event probabilities in a simple model using various estimators.

With harmonic coupling of the external parameter to a reaction coordinate $\xi: \rc \in \mathcal{Q} \mapsto \xi(\rc)$, the extended potential exhibits the following form 
\begin{equation}\label{eq:quadratic}
\pot(\lambda,\rc) = \pot_0(\rc) + \restrain\big(\lambda,\xi(\rc)\big),
\end{equation}
with $\pot_0(\rc) = \beta_{\mathrm{ref}}\mathcal{V}(\rc)$ wherein $\mathcal{V}(\rc)$ and $\beta_{\mathrm{ref}}$ are the potential and  the inverse temperature of the reference system, respectively. 
Here, the restraining potential $ \restrain\big(\lambda,\xi(\rc)\big)$ is harmonic and centered on the value of the reaction coordinate $\xi(\rc)$, denoted by  $\xis$ below
\begin{equation}
 \restrain\big(\lambda,\xis\big) = 
  \frac{1}{2}\beta_{\mathrm{ref}}\kappa \| \lambda - \xis\|^2 + \varepsilon(\xis), 
\end{equation}
where $\kappa$ is the spring stiffness and $\varepsilon(\xis)$ is a small correcting potential that is included to ensure that the restraint does not affect the reference distribution. We have 
\[\varepsilon(\xis) = \sum_{\varlambda\in\varLambda} \exp\left[- 
  \frac{1}{2}\beta_{\mathrm{ref}}\kappa \| \varlambda - \xis\|^2 \right].\]
The unbiased marginal probability reads 
\begin{equation}
\p^{\mathcal{Q}}_0(\rc)=\exp\big[-\pot_0(\rc) - \n_0 \big]. 
\end{equation}  

Consider now the observable $\2_\xis \circ \xi)$. Since it depends on the internal coordinates exclusively, its expectation in the unbiased expanded ensemble is equal to the value expected in the reference thermodynamic ensemble:  
\begin{equation}\label{expectation:total0}
 \e\big[\2_\xis \big] = \sum_{\varlambda\in\varLambda}\int_{\mathcal{Q}} \2_\xis\circ\xi(\brc) \mathrm{p}_0(\varlambda,\brc)d\brc = \int_\mathcal{Q} \2_\xis\circ\xi(\rc) \p^{\mathcal{Q}}_0(\rc)d\rc \triangleq \e^\mathcal{Q} \left[\2_\xis\circ\xi\right]. 
\end{equation}
The restraining potential itself, without the action of the auxiliary biasing potential, does not affect the thermodynamic expectations of observable $\2_\xis \circ \xi$. Subtracting an adequate biasing potential $\aux(\varlambda)$ to the extended system enables one to drive the external parameter that in turn will pull the particle system towards regions of interests. 
This setup is used when one wishes to mechanically restrain the average value of the reaction coordinate $\xi$ owing to $\aux(\varlambda)$. 
As a result, the biased probability density  $\mathrm{p}_\aux(\varlambda,\brc)$ is sampled, but expectations should be evaluated with respect to the unbiased probability density $\mathrm{p}_0(\varlambda,\brc)$. We show how this task can be efficiently done through conditioning in Sec.~\ref{application:lj38}. In the considered example, the $\xi(q)$ is a bond orientational order parameter~\cite{steinhardt:1983} able to describe structural transitions in a small atomic cluster. 

The present harmonic setting is used in extended ABF,~\cite{cao:2014} unified free energy dynamics~\cite{cuendet:2014} and extended Lagrangian metadynamics.~\cite{iannuzzi:2003,laio:2008} The last formulation of metadynamics differs from the (currently) most widely implemented formulation~\cite{laio:2002} in which a biasing potential similar to $\aux$ acts directly on the reaction coordinate, $\xi(\rc)$. 
A difficulty with metadynamics, like with any ABP algorithm, is that the updating rate of the biasing potential must be chosen adequately. If the rate rapidly converges to zero, then the biasing potential will evolve extremely slowly and likely not converge during the simulation. In contrast, if the rate slowly converges to zero, then the biasing potential will fluctuate for a long period prior to stabilizing. Finding a good trade-off between these two adverse situations requires judiciously tuning the updating parameters, a difficult task in general. ABF methods are (almost) free of such updating parameters. The advantage of extended ABF is that it is still applicable when the reaction coordinate is discrete or discontinuous. In contrast, standard ABF requires differentiating the reaction coordinate twice and can not be implemented to compute the free energy along the widely used bond orientational order parameters of Steinhardt et al.~\cite{steinhardt:1983} for instance. 

\section{Conditioning and variance reduction}
\label{section:conditioning}

Prior to showing how \emph{conditioning} is done within the generic estimator, let first describe a simple conditioning scheme and prove variance reduction in the estimation of the marginal probability of $\lambda$ .  

\subsection{Estimating the marginal probability of $\lambda$}
\label{section:estimate_marg_prob}

Let consider a sample $\left\{ \varlambda_m,\brc_m \right\}_{1 \leq m \leq M}$ drawn from the distribution of probability mass $\mathrm{p}_\aux(\varlambda,\brc)$ using a Metropolis  Monte Carlo algorithm or a Langevin process. The generic estimator for estimating  $\lpaux(\lambda)$ consists of evaluating the arithmetic mean of $\g_\aux$, denoted by 
\begin{equation}\label{estimator:standard_generic}
\iam \big(\g_\aux^\lambda\big) = \tfrac{1}{M} \sum_{m=1}^M \g_\aux^\lambda (\varlambda_m,\brc_m ),  
\end{equation}
where we applied the ergodic theorem to relation~\eqref{scaling}. 
Let also denote the conditional probability of $\lambda$ given $\rc$ by 
\begin{equation}
\pi^\lambda_\aux(\rc) = \frac{\exp\left[\aux(\lambda)-\pot(\lambda,\rc)\right]}{\sum_{\varlambda\in\varLambda} \exp\left[\aux(\varlambda)-\pot(\varlambda,\rc)\right]}. 
\end{equation}
The conditionally expected value of $\g_\aux^\lambda$ given $\rc$ is the conditional probability of $\lambda$ given $\rc$
\begin{equation} \label{relation:phy_pmarg}
  \eaux \big[ \g^\lambda_\aux (\cdot, \rc) \big| \rc \big] = \pi^\lambda_\aux(\rc). 
\end{equation}
Besides, the expected value of $\g_\aux^\lambda$ is the expected value of $\pi^\lambda_\aux$ (law of total expectation)
\begin{multline}\nonumber
 \eaux \big[ \g_\aux^\lambda  \big] = \beaux \big[ \eaux \big[ \g_\aux^\lambda (\cdot,\rc) \big| \rc \big] \big] \\= \beaux \big[  \pi_\aux^\lambda \big] = \eaux \big[ \pi_\aux^\lambda \big], 
\end{multline}
where we first resorted to Eq.~\eqref{def:marg_expectation} with $\obs(\cdot,\rc)$ set to $\eaux \big[ \g_\aux^\lambda (\cdot,\rc) \big| \rc \big]$. 
Interestingly, the last term in the sequence of equalities above means that for estimating $\eaux \big[ \g_\aux^\lambda \big]$, it is possible to replace $\g_\aux^\lambda (\varlambda_m,\brc_m )$ in Eq.~\eqref{estimator:standard_generic} by $\pi^\lambda_\aux(\brc_m )$. 
It is precisely this replacement scheme that is referred to as \emph{conditioning}. Estimating the marginal probability of $\lambda$ with the conditioning scheme therefore consists of evaluating the following quantity
\begin{equation}\label{estimator:marginal}
 \iam\big(\pi_\aux^\lambda\big) = \tfrac{1}{M} \sum_{m=1}^M \pi^\lambda_\aux(\brc_m). 
\end{equation}
The equality between expectations $\beaux\left[\pi_\aux^\lambda \right]$ and $\eaux\left[\pi_\aux^\lambda \right]$ indicates that the arithmetic estimator can still be employed using a configuration chain $Q^T=\left\{\brc_m \right \}_{1\leq m \leq M}$ distributed according to the marginal probability $\bpaux(\rc)$. Superscript $T$ stands for transposition so that $Q$ is a column stochastic vector. 

In general, the sampled configurations are \emph{identically} but \emph{not independently} distributed. This implies that the covariance matrix of $Q$ is non diagonal in general (see Lemma~\ref{lemma} for the definition of the covariance matrix of a random vector). 
Denoting the vector encompassing the sampled factors by  $G = \big\lbrace g_\aux^\lambda(\varlambda_m,\brc_m)\big\rbrace^T_{1\leq m \leq M}$ and the $M$-dimensional vector whose components are equal to $M^{-1}$ by $e_M$, the quantities $\iam(g_\aux^\lambda)$ and $\iam(\pi_\aux^\lambda)$ can be re-written as $e^T_MG$ and $\eaux \left(e^T_MG|Q\right)$, respectively. The statistical variances of estimators $\iam (g_\aux^\lambda)$ and $\iam (\pi_\aux^\lambda)$  are then 
\begin{eqnarray}
 \vara \left[\iam(g_\aux^\lambda)\right]   & = & \vara \left[ e_M^T G \right] \label{estimator:img}, \\
 \vara \left[\iam(\pi_\aux^\lambda)\right] & = & \vara^\mathcal{Q} \left[ \eaux \left(e_M^T G|Q\right) \right]\label{estimator:imp}, \end{eqnarray}
where the definition of $\vara$ is given in Eqs.~\ref{def:variance} and~\ref{def:conditional_variance}.  
The reduction of the variance through conditioning stems from the law of total variance~\eqref{law:total_variance_elementary}. Setting $\obs$ to $e^T_MG$, the law writes
\begin{equation}\label{law:total_variance_2}
 \vara [e_M^TG] = \beaux \Big[ \vara[ e_M^T G |Q]\Big] + \vara^\mathcal{Q} \Big[ \eaux  \big( e_M^TG \big| Q \big) \Big]. 
\end{equation}
We assume that the expected conditional variance is strictly positive 
\begin{align}\label{inequality:1}
 \beaux \Big[ \vara[ e_M^T G |Q]\Big] > 0.
\end{align}
The equality is reached when, for any given sample $Q$, $e_M^TG$ is constant, which is assumed to be never the case. Equality~\eqref{law:total_variance_2} and inequality~\eqref{inequality:1} in turn imply the following strict inequality between the variance of the two estimators in~\eqref{estimator:img} and~\eqref{estimator:imp}
\[
\vara \left[ \iam (g_\aux^\lambda) \right] \;  > \;  \vara \left[ \iam (\pi_\aux^\lambda) \right]. 
\]
Hence, the statistical variance associated with the arithmetic mean of the generic function is always larger than that associated with arithmetic mean of the conditional probabilities of $\lambda$, whatever the value of the biasing potential. It is therefore always preferable to use an estimator obtained through conditioning, provided that the overhead associated with the evaluation of the conditional expectation given the sampled states is small enough. 
The cost of conditioning becomes substantial in practice when the dimension of $\varLambda$ exceeds two or three. In this situation, the reduction of the variance may not be important enough to justify implementing a conditioning scheme. In the following, we always assume that performing the numerical quadrature integration within the conditioning scheme has a negligible cost compared to the one of evaluating the potential energy and its gradient. Assuming that the sampled configurations are \textit{identically} and \textit{independently distributed} (i.i.d) entails that the covariance matrix of $G$ is diagonal. The statistical variances of the considered estimators therefore take the following simple forms
\begin{eqnarray}
 \vara \left[\iam(g_\aux^\lambda)\right]   & = & \tfrac{1}{M}\vara \left[ g_\aux^\lambda \right], \\
 \vara \left[\iam(\pi_\aux^\lambda)\right] & = & \tfrac{1}{M}\vara \left[ \pi_\aux^\lambda \right]. \end{eqnarray}
The i.i.d. assumption is made from now so as to facilitate the comparison of the estimator variances in the asymptotic limit of large sample sizes. It will not modify the various inequalities which will be derived to compare the asymptotic variances of the generic and conditioned estimators. 

\subsection{Estimation of conditional expectations}
\label{estimators:conditional_expectation}

\subsubsection*{standard reweighting, binning and generic weighing}

Conditioning for estimating the conditional expectations of $\obs$ given $\lambda$ consists in substituting the conditional expectation given the sampled states for the sampled values of the generic function in the generic estimator of Eq.~\ref{estimator:generic}. Thus, substituting $\pi_\aux^\lambda(\brc_m)$ for $\g_\aux^\lambda ( \varlambda_m,\brc_m)$ yields the adiabatic reweighting estimator of $\e\left[\obs | \lambda \right]$: 
\begin{equation}\label{estimator:conditioning_conditional}
\estimator{M}{\Pi}(\obs|\lambda) = \frac{ \tfrac{1}{M} \sum_{m=1}^M  \pi^\lambda_{\aux}(\brc_m) \obs(\lambda,\brc_m)}{ \tfrac{1}{M} \sum_{m=1}^M \pi^\lambda_{\aux}(\brc_m)}, 
\end{equation}
where we used a Markov chain $\left\{\varlambda_m,\brc_m \right\}_{1\leq m \leq M}$ distributed according to $\paux(\varlambda,\brc)$, or a configuration chain $\left\{\brc_m \right \}_{1\leq m \leq M}$ distributed according to the marginal $\bpaux(\brc)$. 
The substitution that is done amounts to plugging the law of total expectation both in the numerator and the denominator of Eq.~\eqref{expectation:generic_form},
\begin{multline} \label{total_expectation_ratio_conditioning}
 \e\left[ \obs | \lambda \right] = \frac{\beaux \left[ \eaux  \left[\g_{\aux}^\lambda (\cdot,\rc) \obs(\lambda,\rc) \big| \rc \right] \right]}{ \beaux \left[\eaux  \left[\g_{\aux}^\lambda (\cdot,\rc)\big| \rc \right] \right] } \\= \frac{\eaux \left[ {\pi}_{\aux}^\lambda(\cdot) \obs(\lambda,\cdot) \right]}{\eaux \left[ {\pi}_{\aux}^\lambda \right]}. 
\end{multline}
The AR estimator can be directly deduced from Bayes formula as shown in Ref.~\citenum{cao:2014}~and~\citenum{terrier:2015}. Bayes formula is recovered here from Eq.~\eqref{total_expectation_ratio_conditioning} by substituting $\delta(\rc^\star-\rc)$ for the observable and replacing $q^\star$ by $q$
\begin{equation}
 \pi(\rc|\lambda)= \frac{\pi^\lambda_\aux(\rc)\bpaux(\rc)}{\lpaux(\lambda)}. 
\end{equation}
The AR estimator can alternatively be viewed as an instance of waste-recycling Monte Carlo~\cite{frenkel:2004,delmas:2009} when the $\varlambda_m$ are sampled directly from the conditional probabilities $\varlambda \mapsto \pi_\aux^\varlambda(\brc_m)$ at each given $\brc_m$ as suggested in Ref.~\citenum{lidmar:2012}: the wasted information relative to rejected trial moves for $\varlambda\in\varLambda$ is recycled in the estimator. However, resorting to such a Gibbs sampler~\cite{chodera:2011} is not a necessary prescription and any sampler satisfying the detailed balance condition can possibly  be used. In fact, adiabatic reweighting amounts to performing virtual Monte Carlo moves and can be viewed as a particular instance of the virtual-move Monte Carlo method propopsed in Ref.~\citenum{frenkel:2006}.

The present conditioning scheme entails variance reduction in the asymptotic limit of large sample sizes. To prove this property, we will compare the asymptotic variance of the adiabatic reweighting estimator (Eq.~\eqref{estimator:conditioning_conditional}) to that of the generic estimator (Eq.~\eqref{estimator:generic}).  
The present situation differs from that of subsection~\ref{section:estimate_marg_prob} wherein reduction was guaranteed for any sample sizes. The difficulty is due to the presence of a denominator in~Eqs.~\eqref{estimator:generic} and~\eqref{estimator:conditioning_conditional}. 
Let us assume that the function $\rc \mapsto \obs(\lambda,\rc)$ is non-constant (otherwise sampling would not be necessary) and introduce the centered observable $\obs^\lambda(\rc) = \obs(\lambda,\rc) - \e[\obs |\lambda]$, a quantity centered on the value of the conditional expectation given $\lambda$. Then, the quantity $\obs^\lambda \g_\aux^\lambda$ is centered with respect to the total expectation. We have
 \begin{multline}\nonumber
 \e_{\aux} \big[  \g_\aux^\lambda\obs^\lambda\big] = \beaux \Big[ \eaux \big[ \g_\aux^\lambda \left(\cdot,\rc\right) \obs^\lambda(\rc) | \rc \big] \Big]  \\= \beaux \big[ \pi^\lambda_\aux(\cdot)\obs (\lambda,\cdot)  \big] -  \e(\obs|\lambda)\lpaux(\lambda) = 0 . 
\end{multline}
Let us now assume that the generated Markov chains $\left\{\varlambda_m,\brc_m \right\}_{1 \leq m \leq M}$  consist of a sequences of random variables that are i.i.d. according to $\mathrm{p}_{\aux}(\varlambda,\brc)$. Then, the variance of the arithmetic mean of $\g_\aux^\lambda(\varlambda_m,\brc_m)\obs^\lambda(\brc_m)$ multiplied by $M$ decomposes into the variance of $\g_\aux^\lambda(\varlambda,\brc)\obs^\lambda(\brc)$: 
\begin{equation} \label{eq:variance}
M \vara \left[\tfrac{1}{M} \sum_{m=1}^M  \g_\aux^\lambda(\varlambda_m,\brc_m)\obs^\lambda(\brc_m)\right] = \vara \big[\g_\aux^\lambda\obs^\lambda\big]. 
\end{equation} 
In the limit of large sample sizes, the variance of the $\sqrt{M} \estimator{M}{G}(\obs|\lambda)$ quantity becomes equivalent to the following variance 
\begin{equation}\label{asymptotic_equivalence}
M \vara \left[ \estimator{M}{G}\left(\obs \big|\lambda\right) \right]
\underset{M \rightarrow + \infty}{\sim}
\vara \left[ \frac{\g_\aux^\lambda\obs^\lambda}{ \eaux \big[\g_\aux^\lambda \big] } 
 \right].  
\end{equation}
The limit of the left-hand side term of Eq.~\eqref{asymptotic_equivalence} as $M$ tends to infinity is called the asymptotic variance of the $\estimator{M}{G}( \obs|\lambda)$ estimator.
The square-root of the asymptotic variance corresponds to the asymptotic standard error and writes
\begin{equation}\label{asymptotic_error_phi}
 \sigmaux\left[ \estimator{\infty}{G} (\obs|\lambda) \right] = \frac{1}{\lpaux(\lambda)} \sqrt{ \vara \big[\g_\aux^\lambda\obs^\lambda\big]},  
\end{equation}
where we have substituted $\lpaux(\lambda)$ for $\eaux (\g_\aux^\lambda)$. 
This mathematical result is a consequence of the delta method (see Appendix~\ref{appendix:proof}). More precisely, the following convergence in law holds
\begin{equation}\nonumber
  \frac{\tfrac{1}{M} \sum_{m=1}^M  \g_\aux^\lambda(\varlambda_m,\brc_m)\obs (\brc_m)}{\tfrac{1}{M} \sum_{m=1}^M \g_\aux^\lambda(\varlambda_m,\brc_m)} \convlaw \mathcal{N}\left(\e \left[\obs|\lambda\right]  , \sigmaux^2\left[ \estimator{\infty}{G} (\obs|\lambda) \right]  \right), 
\end{equation}
where $\mathcal{N}(\mu,\varsigma)$ denotes the normal law of mean $\mu$ and variance $\varsigma$. 
Similarly, the asymptotic error of the adiabatic reweighting estimator is, with i.i.d. assumption again,
\begin{equation}\label{asymptotic_error_pi}
 \sigmaux \left[\estimator{\infty}{\Pi}(\obs|\lambda) \right] = \frac{1}{\lpaux(\lambda)} \sqrt{\vara \big[ {\pi}_{\aux}^\lambda\obs^\lambda\big]}. 
\end{equation}

To compare the two asymptotic errors, we resort to the law of total variance as in subsection~\ref{section:estimate_marg_prob}, but with respect to $\obs^\lambda \g^\lambda_\aux$ quantity in place of $\g^\lambda_\aux$. The law states that the total variance is equal to the sum of the expectation of the conditional variances given $\rc$ and the variance of the conditional expected values given $\rc$: 
\begin{multline}\label{law:total_variance}\nonumber
 \vara \big(\g_\aux^\lambda\obs^\lambda\big) = \beaux \left[\vara\big( \g_\aux^\lambda(\cdot,q) \obs^\lambda (\rc) \big|\rc\big)\right] + \\ \vara^\mathcal{Q} \left[ \eaux  \big(  \g_\aux^\lambda(\cdot,\rc) \obs^\lambda(\rc) | \rc \big) \right]. 
\end{multline}
Plugging~Eq.~\eqref{relation:phy_pmarg} into the law of total variance leads to 
\begin{equation}\label{inequality:var}
\vara \left[ \pi_{\aux}^\lambda  \obs^\lambda \right] = \vara \left[\g_\aux^\lambda\obs^\lambda\right] - \beaux \left[  \vara\big[ \g_\aux^\lambda \obs^\lambda |\rc \big] \right]. 
\end{equation}
The function $\rc \mapsto \obs^\lambda(\rc)$ being non-constant and the conditional variance of $\g_\aux^\lambda$ being strictly positive for all $\rc$, the last expectation above is strictly positive. Thus, $\vara \big( \pi_{\aux}^\lambda \obs^\lambda\big)$ is strictly lower than  $\vara \big(\g_\aux^\lambda\obs^\lambda\big)$. 
From identities~\eqref{asymptotic_error_phi}, ~\eqref{asymptotic_error_pi} and~\eqref{inequality:var}, we deduce the following strict inequality between the asymptotic errors of the estimators  
\begin{equation}\label{inequality:estimator}
 \sigmaux\left[ \estimator{\infty}{\Pi}(\obs|\lambda) \right]  < \sigmaux \left[ \estimator{\infty}{G}(\obs|\lambda) \right]. 
\end{equation}
It is therefore always preferable to use the adiabatic reweighting estimator rather than the binning, standard reweighting or generic estimators. We now go on by discussing the relevance of implementing self-consistent reweighting estimators in combination with expanded ensemble simulations.

\subsubsection*{Self-consistent reweighting}
\label{section:well-conditioning}

To estimate conditional expectations in expanded ensemble simulations with high accuracy, it has been suggested by Chodera \emph{et al.}~\cite{chodera:2007} to implement the self-consistent reweighting estimator WHAM.  
The original WHAM approach consists of performing a number of independent simulations for each $\lambda\in\varLambda$ and pooling the generated samples into a single sample whose total size is $M=\sum_{\lambda\in\varLambda}M_\lambda$. The $M_\lambda$ quantities are the sizes of the $\lambda$-samples, the original samples generated at constant $\lambda$ using the independent simulations. These sizes are crucial input parameters in self-consistent reweighting estimators together with the pooled sample. Since a single simulation is performed in the expanded ensemble, it has been proposed to estimate the $\lambda$-sample sizes through standard histogram binning (see Eq.~\ref{estimator:standard_generic})
\[M_\lambda=\sum_{m=1}^M \1_{\lambda}(\varlambda_m) = M\iam\left(\1_{\lambda}\right),\quad \lambda\in\varLambda .\] 
The collected data are then postprocessed using WHAM. The overall procedure~\cite{chodera:2007} is referred to as \textit{simulated tempering WHAM} (STWHAM). Resorting to the bin-less formulation~\cite{shirts:2008,tan:2012} of WHAM, the self-consistent reweighting estimator of $\e\left[\obs|\lambda\right]$ can be cast into the following form 
\begin{equation}\label{estimator:mbar}
\estimator{M}{SC}(\obs | \lambda) = \sum\limits_{m=1}^M  \frac{\obs(\lambda,\brc_m)\exp[ \widehat{\pmf}(\lambda) -\pot(\lambda,\brc_m) ]}{\sum_{\varlambda\in\varLambda} M_\varlambda \exp [\widehat{\pmf}(\varlambda) - \pot(\varlambda,\brc_m) ]}, 
\end{equation}
where the quantities $\widehat{\pmf}(\varlambda)$ for $\varlambda\in\varLambda$ are the estimated free energies. They are given by the solutions (up to a common constant) to the following set of nonlinear equations: 
\begin{equation} \label{system:SC}
\estimator{M}{SC}(1 | \lambda) = 1,\qquad \lambda \in \varLambda.
\end{equation} 
System~\eqref{system:SC} is equivalent to 
\begin{equation} \label{system:SC2}
\iam\left(\pi^\lambda_{\widehat{\aux}}\right) = \iam\left( \1_\lambda \right) ,\qquad  \lambda \in \varLambda.
\end{equation}
where the unknown function $\widehat{\aux}(\lambda)$ stands for the quantity $\widehat{\pmf}(\lambda)+\ln M_\lambda$ in Eq~\eqref{estimator:mbar}. 
Since the $\lambda$-sample sizes $M_\lambda$ are stochastic quantities in the STWHAM estimator~\eqref{estimator:mbar}, it is legitimate to do some conditioning on them and to use as input data
\begin{equation} \label{eep}
\widehat{M}_\lambda = \sum\limits_{m=1}^M \frac{\exp \left[ {\aux(\lambda)- \pot (\lambda ,\brc_m)} \right]} {\sum_{\varlambda \in \varLambda} \exp\left[{\aux(\varlambda)-\pot(\varlambda, \brc_m)} \right]}= M\iam\left(\pi_\aux^\lambda\right), 
\end{equation}
instead of $M_\lambda$  for $\lambda  \in \varLambda$. We show in Appendix~\ref{appendix:STWHAM} that conditioning the self-consistent reweighting estimator further reduces its asymptotic statistical variance. 
Besides, conditioning amounts to substituting $\iam\left(\pi^\lambda_{\aux}\right)$ for $\iam\left( \1_\lambda \right)$ in system~\eqref{system:SC2}. 
This entails that the self-consistent solutions $\widehat{\aux}(\lambda)$ become equal to $\aux(\lambda)$ (up to a common additive constant). 
The conditioned self-consistent reweighting estimator is therefore obtained by substituting
$\aux(\lambda) - \ln \widehat{M}_\lambda$ for $\widehat{\pmf}(\lambda)$ and $\widehat{M}_\lambda$ for $M_\lambda$ in the unconditioned self-consistent reweighting estimator~\eqref{estimator:mbar}. These two substitutions enable one to recover the AR estimator~\eqref{estimator:conditioning_conditional} as shown below: 
\begin{equation}\label{estimator:adiabatic_mbar}\nonumber
\begin{aligned}
\estimator{\widehat{M}}{SC}(\obs|\lambda) 
& =  \sum\limits_{m=1}^M  \frac{\obs(\lambda,\brc_m) \exp \left[ {\aux(\lambda)-\ln\widehat{M}_\lambda- \pot (\lambda ,\brc_m)} \right]} {\sum_{\varlambda\in\varLambda} \widehat{M}_\varlambda \exp\left[{\aux(\varlambda)-\ln\widehat{M}_\varlambda-\pot(\varlambda , \brc_m)} \right]}\\[0.25cm] &=
\dfrac{ \tfrac{1}{M} \sum\limits_{m=1}^M  \obs(\lambda,\brc_m) \dfrac{\exp \left[ {\aux(\lambda)- \pot (\lambda ,\brc_m)} \right]} {\sum_{\varlambda\in\varLambda} \exp\left[{\aux(\varlambda)-\pot(\varlambda , \brc_m)} \right]}}{\tfrac{1}{M} \sum\limits_{m=1}^M \dfrac{\exp \left[ {\aux(\lambda)- \pot (\lambda ,\brc_m)} \right]} {\sum_{\varlambda\in\varLambda} \exp\left[{\aux(\varlambda)-\pot(\varlambda , \brc_m)} \right]} } \\[0.25cm] 
&= \dfrac{ \tfrac{1}{M} \sum_{m=1}^M   \pi_\aux^\lambda ( \brc_m) \obs(\lambda,\brc_m) }{\tfrac{1}{M} \sum_{m=1}^M \pi_\aux^\lambda ( \brc_m)} = \estimator{M}{\Pi}(\obs|\lambda).& 
\end{aligned}
\end{equation} 

To summarize, the use of the conditioned $\lambda$-sample sizes in STWHAM provides a reduction of the statistical variance (see Appendix~\ref{appendix:STWHAM}), and, concomitantly, the task of solving a large set of nonlinear equations is avoided because the   self-consistent reweighting with conditioning amount to adiabatic reweighting. More generally, \emph{conditioning} the expectations associated with the binning, standard reweighting and self-consistent reweighting estimators reduces the statistical variances in the analysis of \emph{expanded ensemble simulations} and leads to the formulation of the identical adiabatic reweighting estimator. Adiabatic reweighting is thus asymptotically optimal among the large class of considered estimators. 

\subsection{Estimation of total expectations}
\label{estimators:total_expectation}

In addition to conditional expectations, expanded ensemble simulations also aim at estimating the total expectations whose general form is given in Eq.~\ref{expectation:total}, for instance in order to construct histograms of occupation probabilities along an external parameter as in Eq.~\ref{def:histogram_lambda} or along an internal reaction coordinate as in Eq.~\ref{expectation:total0}. In both situations, the co-logarithm of the occupation probabilities gives access to the free energy along the external parameter (see Eq.~\ref{def:fe} of Sec.~\ref{section:expectations}) or along the internal coordinate (see Eq.~\ref{def:fexi} of Sec.~\ref{application:lj38}). The questions then arise as to (i) how to transpose the generic estimator, (ii) how to condition the transposed estimator and (iii) whether conditioning achieves variance reduction. 

To answer question (i), we decompose the total expectation of $\obs$ resorting to the law of total expectation~\eqref{law:total_expectation_lambda} with respect to $\lambda$ and plug the expectation ratio~\eqref{expectation:generic_form} that involves the weighing function $\g_\aux^\lambda$. This yields 
\begin{equation}\label{expectation:decomposition}
 \e\left[\obs \right] = \sum_{\lambda\in\varLambda} \frac{\eaux \left[ \g^\lambda_{\aux}(\cdot,\cdot) \obs(\lambda,\cdot) \right]}{\eaux \left[ \g^\lambda_{\aux} \right]} \lpo(\lambda),  
\end{equation} 
The unbiased marginal probability of $\lambda$ that appears in the right-hand side (rhs) term is now expressed as a function of the marginal probability of $\lambda$ with the biasing potential switched on
\begin{equation}\label{margprobtheta}\nonumber
 \lpo (\lambda) = \frac{e^{-\aux(\lambda)} \lpaux (\lambda) }{\sum_{\lambdag\in\varLambda} e^{-\aux(\lambdag)}\lpaux(\lambdag) } \\= \dfrac{ e^{-\aux(\lambda)}\eaux \left[\g_\aux^\lambda \right]}{ \eaux \left[ \sum_{\lambdag\in\varLambda} e^{-\aux(\lambdag)  }  \g^{\lambdag}_{\aux} \right]}. 
\end{equation}
Inserting the last term above into the rhs term of Eq.~\ref{expectation:decomposition} and eventually permuting the expectation $\eaux$ and the sum $\sum_{\lambda\in\varLambda}$ yields 
\begin{equation}\label{expectation:unconditional}
 \e\left[\obs \right] =  \frac{\eaux \left[ \sum_{\lambda\in\varLambda}\obs(\lambda,\cdot) e^{-\aux(\lambda)} \g_\aux^{\lambda} (\cdot,\cdot) \right]}{ \eaux \left[ \sum_{\lambda\in\varLambda} e^{-\aux(\lambda)}  \g^{\lambda}_{\aux}\right]}. 
\end{equation}
To manipulate such total expectations, it is more convenient to multiply the previously employed weighing functions ($\1_\lambda$, $\w_\aux^\lambda$, $\g_\aux^\lambda$ and $\pi_\aux^\lambda$) by $\exp[-\aux(\lambda)]$. The modified weighing functions are denoted by $\dsh_\aux^\lambda$, $\dsw_\aux^\lambda$, $\dsg_\aux^\lambda$ and $\dspi_\aux^\lambda$, respectively. The employed notations with their definitions are compiled in Table~\ref{table:weighing}. 
Inserting the functions $\dsg_\aux^{\lambda}(\varlambda,\brc)=e^{-\aux(\lambda)}\g^{\lambda}_\aux (\varlambda,\brc)$ and 
$ \dsg_\aux(\varlambda,\brc) = \sum_{\lambda \in\varLambda } \dsg_\aux^{\lambda}(\varlambda,\brc)$ into Eq.~\ref{expectation:unconditional} yields
\begin{equation}\label{expectation:total2}
 \e\left[\obs\right] =  \frac{\eaux \left[ \sum_{\lambda\in\varLambda}\obs(\lambda,\cdot) \dsg_\aux^{\lambda} (\cdot,\cdot) \right]}{ \eaux \left[ \dsg_{\aux} \right]}. 
\end{equation}
We are now in a position to formulate the estimator of the total expectation based on relation~\eqref{expectation:total2}: 
\begin{equation}~\label{estimator:generic_total}
 \estimator{M}{G} (\obs) =\frac{\tfrac{1}{M} \sum_{m=1}^M \sum_{\lambda\in\varLambda}\obs(\lambda,\brc_m) \dsg_\aux^{\lambda}(\varlambda_m,\brc_m)}{\tfrac{1}{M} \sum_{m=1}^M \dsg_\aux(\varlambda_m,\brc_m)},   
\end{equation}
where $\left\{\varlambda_m,\brc_m\right\}_{1 \le m \le M} $ is a Markov chain generated according to the probability distribution $\mathrm{p}_\aux (\varlambda,\brc)$. 

We have in particular for the binning estimator of $\e\left[ \obs \right]$
\[
\estimator{M}{H} (\obs) =  \frac{\tfrac{1}{M} \sum_{m=1}^M \obs(\varlambda_m,\brc_m) \exp[-\aux(\varlambda_m)] }{\tfrac{1}{M} \sum_{m=1}^M \exp[-\aux(\varlambda_m)]},
\]
where we replaced $\dsg^\lambda_\aux $ and  $\dsg_\aux $ in Eq.~\eqref{expectation:total2} by 
$\dsh^\lambda_\aux$ and $\dsh_\aux$. With \textit{harmonic coupling}, the standard reweighting estimator of $\e\left[ \2_\xis \right]$ writes
\[
\estimator{M}{R}(\2_\xis) = \dfrac{\tfrac{1}{M} \sum\nolimits_{m=1}^M \2_\xis(\xi_m) \exp\left[ -\aux(\varlambda_m) + \restrain (\varlambda_m, \xi_m)\right] }{\tfrac{1}{M} \sum\nolimits_{m=1}^M \exp\left[ -\aux(\varlambda_m) + \restrain (\varlambda_m, \xi_m)\right]}, 
\]
where we replaced $\dsg^\lambda_\aux $ and  $\dsg_\aux $ in Eq.~\eqref{expectation:total2} by
$\dsw^{\lambda}_\aux$ and $\dsw_\aux$ and $\xi_m$ stands for $\xi(\brc_m)$. 
This relation shows that it is in principle possible to remove the biasing and harmonic potentials simultaneously. However, this way of proceeding is not efficient. Instead, integration of the averaged restraint $\eaux\left[ \partial_\xis \restrain(\cdot,\xis) | \xis \right]$ along $\xis$ is preferred in practice. We refer to the discussion of CZAR method in Sec.~\ref{application:lj38}. 

With regards to question (ii), conditioning consists in replacing the $\dsg^{\lambda}_\aux(\varlambda_m,\brc_m)$ and $\dsg_\aux(\varlambda_m,\brc_m)$ terms by their expected values given $\brc_m$, which are respectively ($\rc \equiv \brc_m$)
\begin{eqnarray}\label{exp_raux} \nonumber
\eaux \left[\dsg^{\lambda}_\aux(\cdot,\rc) | \rc \right]  &=&\exp\left[- \aux(\lambda)\right]\eaux \left[\g^{\lambda}_\aux(\cdot,\rc) \big| \rc \right]= \dspi_\aux^{\lambda}(\rc),\nonumber \end{eqnarray}
and 
\begin{eqnarray}\nonumber
 \eaux \left[ \dsg_\aux(\cdot,\rc) | \rc \right] &=& \sum_{\lambda\in\varLambda} \dspi_\aux^{\lambda}(\rc) = \dspi_\aux(\rc).  \nonumber
\end{eqnarray} 
Next, we write the law of total expectation in the rhs ratio of Eq.~\ref{expectation:unconditional} and plug the expected value of $\dsg_\aux(\varlambda,\rc)$ given $\rc$ 
\begin{multline}\label{expectation:unconditional_bf}
  \e\left[\obs\right] = \frac{\beaux \Big[\eaux \big[ \sum_{\lambda\in\varLambda} \dsg^\lambda_\aux(\cdot,\rc)\obs(\lambda,\rc) \big| \rc \big] \Big]}{ \beaux \Big[ \eaux \big[ \dsg_\aux(\cdot,\rc) \big| \rc \big] \Big]} \\= \frac{\eaux \Big[\sum_{\lambda\in\varLambda} \dspi_\aux^{\lambda}(\cdot)\obs(\lambda,\cdot) \Big]}{ \eaux \Big[ \dspi_\aux\Big]}.
\end{multline}
From the rhs expectation ratio of Eq.~\ref{expectation:unconditional_bf} and by application of the ergodic theorem, the adiabatic reweighting estimator of $\e[\obs]$ below is deduced: 
\begin{equation}\label{estimator:total_general}
 \estimator{M}{\Pi}(\obs) =   \frac{\tfrac{1}{M} \sum_{m=1}^M\sum_{\lambda\in\varLambda} \dspi_\aux^{\lambda}(\brc_m)\obs(\lambda,\brc_m) }{\tfrac{1}{M} \sum_{m=1}^M \dspi_\aux(\brc_m)}, 
\end{equation}
where $\left\{\brc_m \right\}_{1 \leq m \leq M}$ is a Markov chain of states distributed according to probability distribution $\bpaux(\rc)$. This one can possibly be extracted from an expanded Markov chain $\left\{\varlambda_m,\brc_m \right\}_{1 \leq m \leq M}$ generated according to $\mathrm{p}_\aux(\varlambda,\brc)$ probability density. 
We have compiled in Table~\ref{table:weighing} the relations between total expectations and weighing functions which are useful for the construction of the estimators. 

We answer question (iii) in the affirmative: conditioning for estimating total expectations achieves variance reduction. As for estimations of conditional expectations in section~\ref{section:conditioning}, this property is a consequence of the law of total variance. The proof follows the same reasoning, but requires replacing the conditionally centered variable $\obs^\lambda(\rc) = \obs(\lambda,\rc)-\e\left[\obs|\lambda\right]$ by the totally centered variable 
\begin{equation}
\obs^\varLambda(\varlambda,\brc) = \obs(\varlambda,\brc)-\e\left[\obs\right], 
\end{equation} 
and the weighing factor $\g_\aux^\lambda$ by a sum over $\lambda\in\varLambda$ involving the $\dsg_\aux^{\lambda}$ factors. 
The asymptotic variance of the $\estimator{M}{R}(\obs)$ estimator writes (see Appendix~\ref{appendix:proof})
\begin{align} \label{variance_total_expectation}
\sigmaux^2 \left[\estimator{\infty}{G} (\obs) \right] & = \vara \left[ \sum_{\lambda\in\varLambda} \obs^\varLambda(\lambda,\cdot) \dfrac{\dsg^{\lambda}_\aux(\cdot,\cdot) }{\eaux(\dsg_\aux)}\right], 
\end{align} 
where the quantity inside the variance is also centered. 
The asymptotic variance of the $\estimator{M}{\Pi}(\obs)$ estimator is obtained from the one of the $\estimator{M}{\Pi}(\obs|\lambda)$ estimator by replacing the conditional probabilities $\pi_\aux^{\lambda}(\rc)$ by sums involving $\dspi_\aux^{\lambda}(\rc)$ over $\lambda \in \varLambda$. One obtains
\begin{multline}\nonumber
\sigmaux^2 \left[ \estimator{\infty}{\Pi} (\obs)\right]  = \vara \left[ \sum_{\lambda\in\varLambda} \obs^\varLambda(\lambda,\cdot)\dfrac{ \dspi_\aux^{\lambda}(\cdot)}{\eaux(\dspi_\aux)} \right] \\ = \vara^\mathcal{Q} \left[ \eaux\left[\sum_{\lambda\in\varLambda} \obs^\varLambda(\lambda,\rc) \dfrac{\dsg_\aux^{\lambda}(\cdot,\rc)}{\eaux(\dsg_\aux)} \bigg|\rc\right] \right], 
\end{multline}
where $\dspi_\aux=\sum_{\lambda\in\varLambda}\dspi_\aux^\lambda$. 
Plugging the law of total variance into the right-hand side variance enables one to conclude that the asymptotic variance of the $\estimator{M}{\Pi}(\obs)$ estimator is smaller than that of the $\estimator{M}{G}(\obs)$ estimator  
\begin{multline}
\sigmaux^2 \left[\estimator{\infty}{\Pi}(\obs)\right] = \vara \left[ \sum_{\lambda\in\varLambda} \obs^\varLambda(\lambda,\cdot) \dfrac{\dsg_\aux^{\lambda}(\cdot,\cdot) }{\eaux(\dsg_\aux)}\right] \\- \beaux \left[ \vara \left[ \sum_{\lambda\in\varLambda} \obs^\varLambda(\lambda,\rc) \dfrac{\dsg_\aux^{\lambda}(\cdot,\rc)}{\eaux(\dsg_\aux)} \bigg|\rc \right] \right] \nonumber  < \sigmaux^2 \left[\estimator {\infty}{G}(\obs)\right]. 
\end{multline}

We will illustrate the estimation of total expectations by setting the observable to the indicator functions $\1_\lambda(\varlambda)$ or $\2_\xis\left(\xi(\brc)\right)$ where $\xi(\brc)$ is an internal reaction coordinate. 
The co-logarithms of the expected values of the two indicator functions yield the free energies along the external and internal coordinate, respectively. 
Various ways of estimating free energies, the primary goal of expanded ensemble simulations, are discussed next in subsection~\ref{estimation_free_energy} and also later in  Sec.~\ref{application:lj38}, for the external and internal coordinate cases, respectively. 

\begin{table}[h]
\caption{Notations and definitions of weighing functions and relations to conditional and total expectations.}
\label{table:weighing}
\scalebox{0.95}{$
\begin{array}{|c|}
\hline\\[-0.4cm]
\mbox{Definitions}\\ \hline\\[-0.1cm]
\begin{array}{lclcll}
 \dsh^\lambda_\aux(\varlambda,\brc) &=& \exp\left[-\aux(\lambda) \right] h^\lambda_\aux(\varlambda,\brc) &=& \exp\left[ -\aux(\lambda) \right] \1_\lambda(\varlambda) & \qquad\dsh_\aux =\sum_{\lambda \in \varLambda}\mathbbm{h}_\aux^\lambda,  \\[0.5cm]
 \dsw^\lambda_\aux(\varlambda,\brc) &=& \exp\left[-\aux(\lambda) \right]\w^\lambda_\aux(\varlambda,\brc)  &=& \dfrac{1}{ \cardTheta }\dfrac{\exp\left[- \pot (\lambda,\brc)\right]}{\exp\left[\aux(\varlambda,\brc)-\pot(\varlambda,\brc) \right]} &\qquad  \dsw_\aux =\sum_{\lambda \in \varLambda}\dsw_\aux^\lambda,  \\[0.75cm]
\dsg_\aux ^\lambda (\varlambda,\brc) &=& \exp[-\aux(\lambda)] \g_\aux(\varlambda,\brc)  & = & \begin{cases} \dsh^\lambda_\aux(\varlambda,\brc) & \mbox{if binning,}\\
\dsw^\lambda_\aux(\varlambda,\brc) & \mbox{if standard reweighting,}
\end{cases}&\qquad \dsg_\aux =\sum_{\lambda \in \varLambda}\dsg_\aux^\lambda,	\\[0.75cm]
\dspi^\lambda_\aux(\brc) &=& \exp\left[ -\aux(\lambda) \right] \pi_\aux^\lambda(\brc) &=&  \dfrac{\exp\left[- \pot (\lambda,\brc)\right]}{\sum_{\varlambda \in \varLambda}\exp\left[\aux(\varlambda,\brc)-\pot(\varlambda,\brc) \right]}, &\qquad\dspi_\aux = \sum_{\lambda \in \varLambda}\dspi_\aux^\lambda . \\[0.75cm]
\end{array}\\
\hline\\[-0.6cm]
\mbox{Expressions for total expectations}\\ \hline\\[-0.1cm]
\begin{array}{ccccccc}
\e \left[ \obs\right] &=& \dfrac{\eaux\left[\sum_{\lambda \in \varLambda}  \dsg^{\lambda}_\aux(\cdot,\cdot) \obs(\lambda,\cdot) \right]}{\eaux\left[ \dsg_\aux \right]} 
&\Longrightarrow& 
\e\left[\1_\lambda \right] &=& \dfrac{\eaux\left[ \dsg^{\lambda}_\aux \right]}{\eaux\left[ \dsg_\aux \right]}\\[0.5cm]
\e \left[ \obs\right] &=& \dfrac{\eaux\left[\sum_{\lambda \in \varLambda}  \dspi^{\lambda}_\aux(\cdot) \obs(\lambda,\cdot) \right]}{\eaux\left[ \dspi_\aux \right]} 
&\quad\Longrightarrow\quad& 
\e\left[\1_\lambda\right] &=& \dfrac{\eaux\left[ \dspi^{\lambda}_\aux \right]}{\eaux\left[ \dspi_\aux \right]}\\[0.5cm]
\end{array}\\[1.0cm]\hline\\[-0.4cm]
\begin{array}{c}
\mbox{Expressions for conditional expectations}
\end{array}\\ \hline \\[-0.2cm]
\begin{array}{rclcrcl}
 \eaux \left[ \dsg^\lambda_\aux(\cdot,\rc)\big|\rc \right] &=& \dspi^\lambda_\aux(\rc)
&\quad\Longleftarrow\quad&
\eaux \left[ \g^\lambda_\aux(\cdot,\rc)\big|\rc \right] &=& \pi^\lambda_\aux(\rc) 
, \\[0.4cm]
  
\e \left[ \partial_\lambda \pot (\lambda,\cdot) |\lambda \right] &=& - \partial^\lambda \ln \eaux \left[ \dsw^\lambda_\aux\right] &\quad\Longleftarrow\quad& \partial^\lambda\dsw^\lambda_\aux(\cdot,\cdot) &=& - \partial_\lambda\pot(\lambda,\cdot) \dsw^\lambda_\aux(\cdot,\cdot), \\[0.4cm]

\e \left[ \partial_\lambda \pot (\lambda,\cdot) |\lambda \right] &=& - \partial^\lambda \ln \eaux \left[ \dspi^\lambda_\aux \right] 
&\quad\Longleftarrow\quad&
\partial^\lambda\dspi^\lambda_\aux(\cdot) & = & - \partial_\lambda\pot(\lambda,\cdot) \dspi^\lambda_\aux(\cdot)
. \\[0.4cm]
\end{array}\\[0.8cm]
\hline
\end{array}
$}
\end{table}

\subsection{Estimation of the free energy along an external parameter}
\label{estimation_free_energy}

As defined in Eq.~\ref{def:fe}, the free energy $\pmf(\lambda)$ is the co-logarithm of the total expectation of the indicator function $\1_\lambda(\varlambda)$ for $\lambda \in \varLambda$. Its derivative is a conditional expectation of the $\lambda$-derivative of the extended potential given $\lambda$: 
\begin{equation}
\pmf^\prime (\lambda) = \e \left[ \partial_\lambda \pot (\lambda,\cdot )|\lambda\right]. 
\end{equation} 
This quantity can be estimated and then integrated to obtain the free energy. Table~\ref{table:fe_expectations} illustrates the various ways of estimating the corresponding total and conditional expectations using the generic weighing functions and a time-independent auxiliary biasing potential. 

Considering the binning, standard reweighting and adiabatic reweighting estimators, denoted by $\estimator{M}{H}$, $\estimator{M}{R}$ and $\estimator{M}{\Pi}$ respectively, potentially makes 6 direct methods of computing the free energy, while direct free energy methods are usually classified into three \emph{overlapping} categories in the literature: thermodynamic occupation~\cite{chandler:1987,frenkel:2001,wales:2003} (TO), thermodynamic integration~\cite{kirkwood:1935} (TI), free energy perturbation~\cite{zwanzig:1954} (FEP). 
We next analyze the correspondence between estimators and free energy methods. 

\begin{table}[h]
\caption{Expectation ratios based on which the various free energy estimators of Table~\ref{table:fe_estimators}  are constructed.}
\label{table:fe_expectations}
\[
\boxed{
\begin{array}{ccll}
\\[0cm]
 \pmf(\lambda) & = & \left\{  \begin{array}{lll} 
-\ln \dfrac{\eaux [ \dsg^\lambda_\aux ]}{\eaux [ \dsg_\aux  ]}& \multicolumn{2}{c}{\mbox{with generic weighing,}} \\[0.75cm]
  -\ln \dfrac{\beaux\left[ \eaux \left[ \dsg^\lambda_\aux(\cdot,q) \big|\rc\right] \right]}{\beaux \left[ \eaux\left[ \dsg_\aux(\cdot,q) \big|\rc \right] \right]} & = -\ln \dfrac{\eaux\left[ \dspi^\lambda_\aux \right] }{\eaux \left[  \dspi_\aux  \right]}  & \mbox{with conditioning, } 
 \end{array} \right.
\\[2cm]
 \pmf^\prime (\lambda) & =& \left\{ 
 \begin{array}{lll} 
\dfrac{\eaux [ \dsg^\lambda_\aux \partial_\lambda \pot(\lambda,\cdot) ]}{\eaux [ \dsg^\lambda_\aux  ]}& 
  \multicolumn{2}{c}{\mbox{with generic weighing,}} \\[0.75cm]
  \dfrac{\beaux\left[ \eaux \left[\partial_\lambda \pot(\lambda,\rc)\dsg^\lambda_\aux(\cdot,\rc)  \big|\rc \right]\right] }{\beaux \left[ \eaux\left[ \dsg^\lambda_\aux(\cdot,q) \big| \rc \right] \right]} & = - \dfrac{\eaux\left[ \partial^\lambda \dspi^\lambda_\aux \right] }{\eaux \left[  \dspi^\lambda_\aux  \right]}  & \mbox{with conditioning. }
 \end{array} \right.
   \\[1.5cm]
\end{array}
}
\]
\end{table}
\vspace{2cm}

As for methods belonging to the first category, estimating the free energy from the log-probability of $\lambda$ is done resorting to the binning estimator as follows
\begin{multline}\nonumber
 \widehat{\pmf(\lambda)}^M_\mathrm{TO} = - \ln \left[ \dfrac{\tfrac{1}{M}  \sum_{m=1}^M \1_\lambda(\varlambda_m)\exp[-\aux(\varlambda_m)] }{\tfrac{1}{M}\sum_{m=1}^M \exp[-\aux(\varlambda_m)]} \right] =\\ - \ln \dfrac{\iam \left( \dsh_\aux^\lambda \right)}{\iam \left( \dsh_\aux\right)} = - \ln \estimator{M}{H}\left(\1_\lambda\right). 
\end{multline}

Methods of the second category consists in estimating the free energy derivative and evaluating the free energy through numerical integration. This is what is actually done in the extended ABF technique.~\cite{cao:2014,lesage:2016} From an expanded ensemble simulation, a simple way of obtaining an estimate $\widehat{\pmf}^\prime(\lambda)$ of the mean force  involves the binning estimator $\estimator{M}{H}$
\begin{multline}\nonumber
\widehat{\pmf^\prime(\lambda)}^M_\mathrm{TI} = \dfrac{\tfrac{1}{M} \sum_{m=1}^M  \partial_\lambda\pot(\lambda,\brc_m) \1_\lambda(\varlambda_m)}{ \tfrac{1}{M} \sum_{m=1}^M \1_\lambda (\varlambda_m)} \\= \dfrac{\iam(\partial_\lambda \pot~\dsh_\aux^\lambda)}{ \iam(\dsh_\aux^\lambda)} = \estimator{M}{H}\big[\partial_\lambda \pot \big|\lambda \big]. 
\end{multline}
It may be suggested to estimate the free energy derivative resorting instead to the standard reweighting estimator as it is done in umbrella sampling.~\cite{torrie:1977}
Using the standard reweighting function introduced in Table~\ref{table:weighing} together with the property $\partial_\lambda \rf (\varlambda,\brc) = - \partial_\lambda \pot(\lambda,\brc)\rf (\varlambda,\brc)$, we have 
\begin{eqnarray}
 \label{estimator:SR_EE}\nonumber
\widehat{\pmf^\prime(\lambda)}_{\mathrm{FEP}}^M = \estimator{M}{R} \big(\partial_\lambda \pot\big| \lambda\big) = - \partial^\lambda \ln \left[ \tfrac{1}{M}{\sum_{m=1}^M}  \rf(\varlambda_m,\brc_m)\right]. 
\end{eqnarray}
The fact that the estimator can be written has a logarithmic derivative of another standard reweighting estimator indicates that it is not necessary to integrate the mean force to obtain the free energy. 
The standard reweighting approach pertains to the second category of free energy methods (FEP), which aim at directly evaluating the free energy by estimating a partition function ratio and then taking its co-logarithm
\begin{equation}\nonumber
\widehat{\pmf(\lambda)}^M_\mathrm{FEP} = - \ln \left[ \dfrac{\tfrac{1}{M}{\sum_{m=1}^M}  \rf(\varlambda_m,\brc_m)}{\tfrac{1}{M}{\sum_{m=1}^M}  \dsw_\aux(\varlambda_m,\brc_m)}  \right] = - \ln \estimator{M}{R} (\1_\lambda). 
\end{equation} 
To perform a conditioning with respect to the FEP and TO method above, one must replace the weighing factors $\dsh_\aux^\lambda$ and $\rf$ by their conditional expected values given $\rc$, which happens to be given by $\dspi^\lambda_\aux(\rc)$. Similarly, $\dsh_\aux$ and $\dsw_\aux$ must also be replaced by $ \dspi_\aux(\rc)$. 
One thus obtains the following estimator 
\begin{equation}\nonumber
 \widehat{\pmf(\lambda)}^M_\mathrm{AR} = - \ln \left[ \dfrac{\tfrac{1}{M}  \sum_{m=1}^M \dspi_\aux^\lambda(\brc_m)}{\tfrac{1}{M} \sum_{m=1}^M \dspi_\aux (\brc_m)} \right]. 
\end{equation}
Differentiating the free energy estimate with respect to $\lambda$ yields 
\begin{eqnarray}\nonumber 
\dfrac{d\widehat{\pmf(\lambda)}^M_{\mathrm{AR}}}{d\lambda} &=& - \partial^\lambda \ln \left[ \dfrac{\sum_{m=1}^M \dspi_\aux^\lambda(\brc_m)}{\sum_{m=1}^M \dspi_\aux(\brc_m)} \right]  \\ \nonumber
&=& - \dfrac{\sum_{m=1}^M \partial^\lambda \dspi^\lambda_\aux(\brc_m)}{\sum_{m=1}^M \dspi^\lambda_\aux(\brc_m)}  \\ \nonumber
&=&\estimator{M}{\Pi} (\partial_\lambda \pot |\lambda) = \widehat{\pmf^\prime(\lambda)}^M_{\mathrm{AR}}, 
\end{eqnarray}
where we substituted $-\partial_\lambda \pot(\lambda,\brc_m) \dspi_\aux^\lambda(\brc_m)$ for $\partial^\lambda \dspi_\aux^\lambda(\brc_m)$ in the second line and eventually identify with the AR estimate of $\e\left[ \partial_\lambda \pot |\lambda\right]$. The consistency between the estimated mean force and the derivative of the estimated free energy in the AR method is a property inherited from the FEP method. However, unlike FEP method, the adiabatic reweighting approach is directly related to thermodynamic integration, since the estimated mean force can also be constructed from the following conditioning scheme 
\begin{eqnarray}\nonumber
\widehat{\pmf^\prime(\lambda)}^M_{\mathrm{AR}}
& = & \dfrac{\tfrac{1}{M}\sum_{m=1}^M \partial_\lambda \pot (\lambda,\brc_m) \eaux\left[\dsh^\lambda_\aux(\varlambda) \big|\brc_m\right]}{\tfrac{1}{M} \sum_{m=1}^M \eaux\left[\dsh^\lambda_\aux(\varlambda) \big|\brc_m\right] }. 
\end{eqnarray}
This means in particular that each sampled point $\brc_m$ with ${1\leq m \leq M}$ contributes to the estimated mean force with an integrated weight over $\varLambda$ that is equal to one. 
This property is inherited from the TI method and does not hold for the FEP method. This explains why the latter method may yield completely inaccurate results in some circumstances. This point is illustrated in Section~\ref{sec:assessment}. 
The four ways of estimating the free energy and its derivative are summarized in Table~\ref{table:fe_estimators}. 

\begin{table}[h]
\caption{Standard (TI, FEP, TO) and conditioning (AR) estimators for computing mean forces and free energies. Note that, unlike $\dsw_\aux^\lambda$ and $\dspi_\aux^\lambda$ weighing functions, the function $\dsh_\aux^\lambda$ can not be differentiated with respect to $\lambda$. It results that two distinct methods (TO and TI) are based on binning.}
\label{table:fe_estimators}
\[
\begin{array}{|l|c|rl|}
\multicolumn{3}{c}{} \\
\cline{1-1}\cline{3-4}
&&&\\[-0.3cm]
\multicolumn{1}{|c|}{ \mbox{Estimation of free energy } \widehat{\pmf (\lambda)}_\mathrm{X}^M } &  & \multicolumn{2}{c|}{\mbox{Corresponding mean force } \widehat{\pmf^\prime (\lambda)}_\mathrm{X}^M} \\[0.25cm]
\hline 
&&&  \\[-0.2cm]
\widehat{\pmf (\lambda)}_\mathrm{TO}^M = 

- \ln \estimator{M}{H}\left(\1_\lambda\right)=-\ln \dfrac{\iam \left(\dsh_\aux^\lambda \right)}{\iam\left(\dsh_\aux \right)} 

& \rightarrow& 
 
 \multicolumn{2}{|l|}{\mbox{~~derivative by finite difference~~~} \rightarrow \quad \widehat{\pmf^\prime (\lambda)}_\mathrm{TO}^M } \\[0.75cm]

\widehat{\pmf (\lambda)}_\mathrm{TI}^M \leftarrow 

\mbox{ numerical quadrature } 

& \leftarrow&    &
\hspace*{-0.05 cm}
\dfrac{\iam\left(\dsh_\aux^\lambda \partial_\lambda \pot\right)}{\iam\left(\dsh_\aux^\lambda\right)} \; = \;\; \estimator{M}{H}(\partial_\lambda \pot |\lambda) \;\;

\\[0.75cm]

\widehat{\pmf (\lambda)}_\mathrm{FEP}^M= 
- \ln \estimator{M}{R}\left(\1_\lambda\right)= -\ln \dfrac{\iam \left(  \dsw_\aux^\lambda \right)}{\iam \left(  \dsw_\aux \right)}  

& \rightleftarrows&  

-\partial^\lambda \ln \dfrac{\iam \left(  \dsw_\aux^\lambda \right)}{\iam \left(  \dsw_\aux \right)} \;=& \! \! \dfrac{\iam \left(-\partial^\lambda \dsw_\aux^\lambda \right)}{\iam \left(  \dsw_\aux^\lambda \right)}\;\; =\;\;\estimator{M}{R}\left(\partial_\lambda \pot |\lambda\right) 
\\[0.75cm]

\widehat{\pmf (\lambda)}_\mathrm{AR}^M = 

- \ln \estimator{M}{\Pi}\left(\1_\lambda\right)=-\ln \dfrac{\iam \left(  \dspi_\aux^\lambda \right)}{\iam \left(  \dspi_\aux \right)} & 

\rightleftarrows& - \partial^\lambda \ln \dfrac{\iam \left(  \dspi_\aux^\lambda \right)}{\iam \left(\dspi_\aux \right)} \;= & \!\! \dfrac{\iam \left( -\partial^\lambda \dspi_\aux^\lambda \right)}{\iam \left(  \dspi_\aux^\lambda \right)} \; = \;\ \estimator{M}{\Pi}\left(\partial_\lambda \pot |\lambda\right)
\\[0.5cm]

\hline
\end{array}
\]
\end{table}

Next, we compare the variances of the TI, FEP and TO methods to the one of the AR method. The asymptotic variance for the TO/FEP method can be cast in the following form using the generic functions $\dsg^\lambda_\aux$ and $\dsg_\aux$ (see Appendix~\ref{appendix:proof}, Eq.~\eqref{variance_fepto})
\begin{equation} \label{variance_fepto_chap2}
\sigmaux^2\left[\widehat{\pmf(\lambda)}^\infty_{\mathrm{FEP/TO}} \right] =  \vara \left[ \frac{\dsg^\lambda_\aux(\varlambda,\rc)}{\eaux\left[ \dsg^\lambda_\aux \right]} - \frac{\dsg_\aux(\varlambda,\rc)}{\eaux\left[\dsg_\aux \right]}\right]. 
\end{equation}
The asymptotic variance of the AR method can be cast in the similar form (see Appendix~\ref{appendix:proof}, Eq.~\eqref{variance_fear})
\begin{multline}\nonumber
 \sigmaux^2\left[\widehat{\pmf(\lambda)}^\infty_{\mathrm{AR}} \right]  = \vara^\mathcal{Q} \left[ \eaux \left[\frac{\dsg^\lambda_\aux(\cdot,\rc)}{\eaux \left[ \dsg^\lambda_\aux \right]} - \frac{\dsg_\aux(\cdot,\rc)}{\eaux \left[\dsg_\aux \right] }\bigg| \rc \right] \right] \\ = \vara \left[ \frac{\dspi_\aux^\lambda}{\eaux\left[ \dspi^\lambda_\aux \right]} - \frac{{\dspi}_\aux}{\eaux\left[{\dspi}_\aux \right]}\right]. 
\end{multline}
The law of total variance then entails the following strict inequality
\begin{multline}\nonumber
 \vara \left[ \dfrac{\dsg^\lambda_\aux}{\eaux\left[ \dsg^\lambda_\aux \right]} - \dfrac{\dsg_\aux}{\eaux\left[\dsg_\aux \right]}\right]-\vara \left[ \dfrac{\dspi_\aux^\lambda}{\eaux\left[ \dspi^\lambda_\aux \right]} -  \dfrac{{\dspi}_\aux}{\eaux\left[{\dspi}_\aux \right]}\right] = \\
\eaux\Bigg\{\vara \left[ \dfrac{\dsg^\lambda_\aux(\cdot,q) }{\eaux\left[ \dsg^\lambda_\aux \right]} - \dfrac{\dsg_\aux(\cdot,q) }{\eaux\left[\dsg_\aux\right]} \bigg| \rc \right] \Bigg\}>0. 
\end{multline}

It results the following strict inequality for the asymptotic variances
\[\sigmaux^2\left[\widehat{\pmf(\lambda)}^\infty_{\mathrm{AR}} \right] 
< \min \left\{\sigmaux^2\left[\widehat{\pmf(\lambda)}^\infty_{\mathrm{TO}} \right],  \sigmaux^2\left[\widehat{\pmf(\lambda)}^\infty_{\mathrm{FEP}} \right] \right\}.\]

With regards to the TI method, the efficiencies of the adiabatic reweighting and histogram binning estimators are more easily compared considering the derivative of the free energy. The asymptotic variances associated with the $\estimator{M}{\Pi}(\partial_\lambda \pot |\lambda)$ and $\estimator{M}{H}(\partial_\lambda \pot |\lambda)$ estimators satisfy the relation $\sigmaux^2\left[\estimator{M}{\Pi}(\partial_\lambda \pot |\lambda) \right]<\sigmaux^2\left[\estimator{M}{H}(\partial_\lambda \pot |\lambda)\right]$. 
It is therefore always preferable to estimate free energies in combination with a conditioning scheme when the auxiliary biasing potential is time-independent. Remarkably, whatever the standard free energy method (TO, FEP, TI) that is chosen, conditioning with respect to the external parameter provides the same AR estimator, as illustrated in Table~\ref{table:fe_estimators}. 

When estimating conditional expectations of an observable $\rc \mapsto \obs(\lambda,\rc)$ from expanded ensemble simulations with the external parameter taking values in $\varLambda$, the biasing potential should ideally be chosen such that it minimizes a statistical variance averaged over $\varLambda$. It appears that  some optimal biasing potentials for the adiabatic reweighting estimator $\estimator{M}{\Pi}(\obs | \lambda)$ can be constructed adaptively in expanded ensemble simulations using partial biasing techniques.~\cite{fort:2016} This avenue of research will be discussed in a companion paper. Here, we consider that the biasing potential has converged to the free energy, which is the optimal biasing potential for the thermodynamic occupation method when the number of bins is large. 

\section{Assessment of variance reduction}
\label{sec:assessment}
We assess the numerical performance of the aforementioned methods of estimating the free energy $\pmf(\lambda)$ by computing the reduction of the statistical variances in subsection~\ref{subsec:assessment}. We also justify the conditioning approach for estimating the probability of rare events in subsection~\ref{sec:rare_event_probability}. 

\subsection{Free energy estimations}\label{subsec:assessment}
 
We consider a one-dimensional system with extended potential 
\begin{equation} \label{toymodel}
\pot(\varlambda,\brc) = \omega \left( \brc^2 - 2 \brc\varlambda \right), 
\end{equation}
and set $\varLambda=\left\{ \lambda^j \right\}_{0\leq j \leq J}$ with $\lambda^j= j/J$. 
The extended probability $\p_\aux(\varlambda,\brc)$ is proportional to $\exp\left[\aux(\varlambda) - \omega \brc^2 +2 \omega \varlambda \brc  \right]$
and the conditional probability of $\rc$ given $\lambda$ is
\begin{equation}\label{eq:distrib}
\pi(\rc | \lambda) = \sqrt{\omega/\pi} \exp\left[-\omega(\rc-\lambda)^2 \right]. 
\end{equation}  
For small values of $\omega$, the various conditional probabilities substantially overlap with each other. The degree of overlap decreases with increasing $\omega$.  

We generate a series of free energy estimates $\widehat{\pmf}^M_\mathrm{X}(\lambda;k)$ wherein $k \in\llbracket 1,K \rrbracket $ is the simulation index, using $K=2\cdot 10^3$ simulations, employing method $\text{X} =$ TI, FEP, TO or AR and setting the biasing potential equal to the true free energy $\pmf(\lambda)$. 
Each simulation consists of up to $M_{\mathrm{max}}=10^5$ sampled states. We compute the variance of the generated estimates using the following variance estimator  
\begin{multline}\label{variances}
\bar{\mathrm{var}}^K\left(\widehat{\pmf}^M_\mathrm{X} \right) = \tfrac{1}{\cardTheta}\sum_{\lambda\in\varLambda}\tfrac{1}{K}\sum_{k=1}^K \left( \dfrac{ \estimate{P}{X}(\lambda;k)}{\tfrac{1}{K}\sum_{h=1}^K \estimate{P}{X}(\lambda;h)}-\right. \\
\left.\dfrac{ \sum_{\varlambda \in \varLambda} \estimate{P}{X}(\varlambda;k)}{\sum_{\varlambda \in \varLambda} \tfrac{1}{K} \sum_{h=1}^K \estimate{P}{X}(\varlambda;h)} \right)^2, 
\end{multline}
where $M$ ranges from $5\cdot10^2$ to $M_{\mathrm{max}}$ and
\begin{equation}
 \estimate{P}{X}(\lambda;k) = \frac{\exp\left[-\estimate{A}{X}(\lambda;k)\right]}
 {\sum_{\varlambda\in\varLambda}\exp\left[ \pmf(\varlambda)- \estimate{A}{X}(\varlambda;k)\right]}. 
\end{equation} 
This quantity arises from the generic expression of the asymptotic variances of $\estimate{A}{X}$ and corresponds to the $k$th estimate of the quantity $\exp\left[-\pmf(\lambda)\right] \times\mathrm{p}^\varLambda_{\pmf}(\lambda)$. 
Note that when the bin of $\1_\lambda$ remains unvisited during the entire $k$th simulation, $\estimate{A}{TO}(\lambda;k)$ is infinite and $\estimate{P}{TO}(\lambda;k)$ is zero. However, the sums $\sum_{h=1}^K\estimate{P}{TO}(\lambda;h)$ over the $2\cdot 10^3$ simulations never canceled, so that the variance estimator~\eqref{variances} was always defined. 

We display in Fig.~\ref{fig:asymp} the estimated variances multiplied by the sample size $M$ as a function of $M$ in order to observe the convergence towards the asymptotic limit for the given number $K$ of independent simulations. 
We observe that a considerable variance reduction is achieved in practice owing to the conditioning scheme. Furthermore, the asymptotic regime is obtained faster and the estimated variances are also less fluctuating with conditioning. 

While fluctuations decrease with increasing $K$, it is extremely costly to obtain accurate estimates of the asymptotic variance for large sample sizes $M$ for the FEP method when $\omega$ is large. 
The observed inefficiency of the FEP method results from the fact that the simulation samples from non-overlapping distributions. In real FEP computations, the reaction path is usually divided into a number of small pieces and the standard reweighting estimator is applied for each of them. Staged simulations are also routinely performed within the umbrella sampling approach.~\cite{torrie:1977}

\begin{figure}
\begin{center}
 \includegraphics[width=0.9\textwidth]{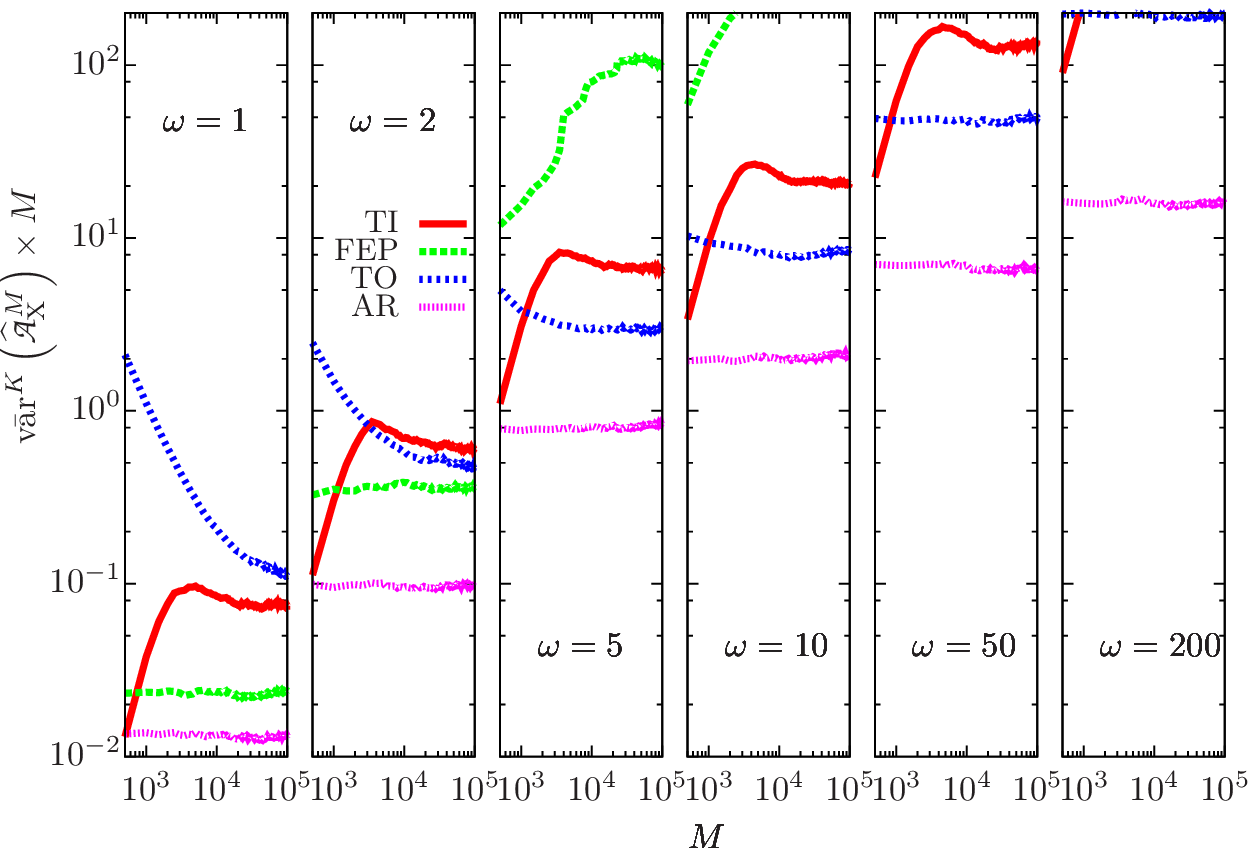}
\end{center}
 \caption{The estimated variances $\bar{\mathrm{var}}^K(\widehat{\pmf}^M_\mathrm{X})$ are evaluated using $K=2\cdot 10^3$ estimates of $\widehat{\pmf}^M_\mathrm{X}$ for methods X=TI, FEP, TO, AR. The scaled variances that are plotted as a function of $M$ reach a plateau value asymptotically.}
\label{fig:asymp}
\end{figure}

\label{subsec:confining}

\subsection{Estimation of rare-event probabilities}
\label{sec:rare_event_probability}
Here, we emphasize the importance of conditioning for estimating probabilities of rare-events.~\cite{terrier:2015} Let the conditional distribution given $\lambda$ equal to 0 in~Eq.~\eqref{eq:distrib} be the reference distribution. The probability to sample a point 
at a position larger than $1$ from this distribution is $\mathpzc{C}=\mathrm{erfc}(\sqrt{\omega})/2$  (the integral of the normal probability distribution $\rc \rightarrow \sqrt{\omega / \pi }\exp(-\omega  \rc^2)$  from 1 to $\infty$). Here, we show that it is possible to compute $\mathpzc{C}$ through biased sampling of $\mathrm{p}_\pmf(\varlambda,\brc)$, even for very narrow conditional distributions $\pi(\brc|\lambda)$ for which the value of $\omega$ is large and $\mathpzc{C}$ is very small. We thus bias the simulation by adding the soft restraint $\restrain(\lambda,\rc)= - 2 \lambda\omega \rc$ to the potential energy $\pot_0(\rc)=\omega \rc^2$, so as to gradually increase the fraction of points such that $\rc \ge 1$ with increasing $\lambda$ value. This way of proceeding corresponds to the so-called tilting protocol that is used in path sampling simulations to sample rare trajectories from a path distribution.~\cite{athenes:2002b,athenes:2004,adjanor:2005,terrier:2015} 

With the biasing potential set equal to the free energy, the probability $\mathpzc{C}$ that point $\rc$ is larger than one is $\e\left[\1_{\mathtt{B}}\big|\lambda=0\right]$, 
where $\mathtt{B} = [1,+\infty[$ and $\1_{\mathtt{B}}(\rc)$ denotes the indicator function  equal to $1$ if $\rc\in \mathtt{B}$ and $0$ otherwise. The conditional expectation is estimated using the adiabatic reweighting  estimator and standard reweighting estimator respectively given by
\begin{align}
\estimator{M}{\Pi} \left(\1_{\mathtt{B}}\big|0\right) &= \dfrac{\tfrac{1}{M} \sum_{m=1}^M \1_{\mathtt{B}}(\brc_m) \pi^0_\pmf(\brc_m)}{ \tfrac{1}{M} \sum_{m=1}^M \pi^0_\pmf(\brc_m)},  \label{eq:BF} \\[0.25cm]
 \estimator{M}{R} \left( \1_{\mathtt{B}}\big|0\right) &= \dfrac{\tfrac{1}{M} \sum_{m=1}^M \1_{\mathtt{B}}(\brc_m) \w^0_\pmf(\varlambda_m,\brc_m)}{ \tfrac{1}{M} \sum_{m=1}^M \w^0_\pmf(\varlambda_m,\brc_m)} \label{eq:SR},  
\end{align}
as obtained by replacing the observable $q \mapsto \obs(0,\rc)$ by the indicator function. Here, we did not consider the binning estimator as it is obviously not suited to the present rare event problem. The goal is to retrieve statistics from configurations exhibiting large $\brc_m$ values which are observed concomitantly with large $\varlambda_m$ sampled values. 

We observe in Fig.~\ref{fig:reductionfactor} that, with increasing $\omega$ parameter, only the AR estimator yields accurate estimates of $\mathpzc{C}$. The computational speed-up of convergence that is achieved by conditioning the standard reweighting estimator can be assessed from their respective standard errors. The inefficiency of the standard reweighting approach was explained in term of the fluctuation theorem in Ref.~\citenum{terrier:2015}. As soon as $\omega$ becomes larger than 20, conditioning reduces the variance by more than four orders of magnitude. 
For rare-event problems, available alternatives~\cite{athenes:2012} to conditioning consists of post-processing the harvested information by implementing a self-consistent reweighting estimator.~\cite{shirts:2008} 

\begin{figure}
\begin{center}
 \includegraphics[width=0.9\textwidth]{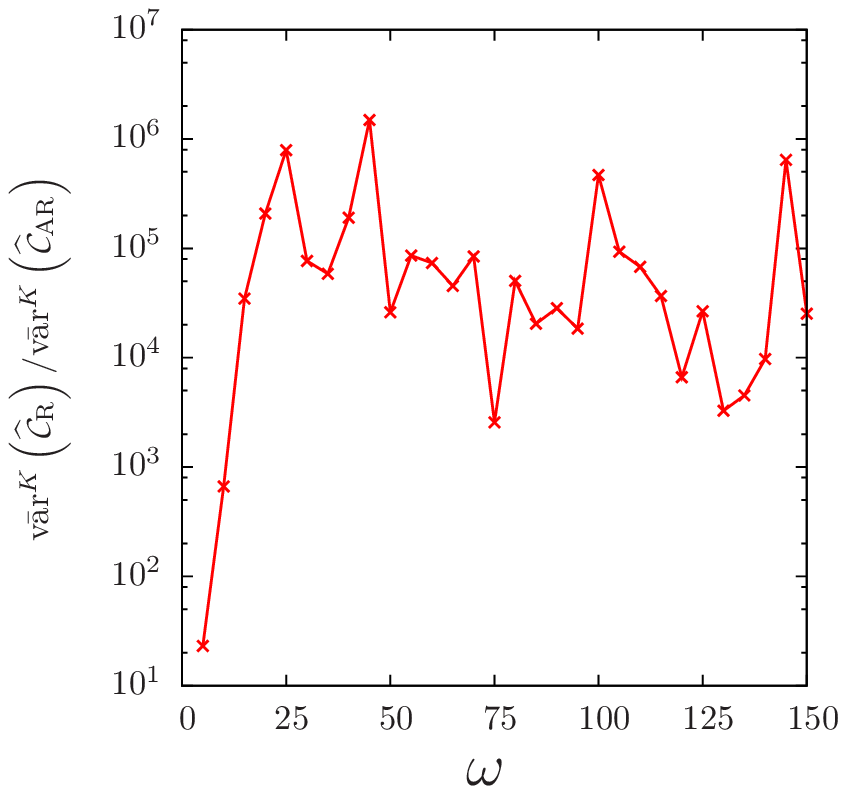}
\end{center}
 \caption{Variance reduction factor obtained through conditioning when estimating the probability $\mathpzc{C}$. The quantities $\widehat{\mathpzc{C}}_\mathrm{R}$ and $\widehat{\mathpzc{C}}_\mathrm{AR}$ denote the estimates obtained using $\estimator{M}{R} \left(h_{[1,+\infty[}\big|0\right)$ and $\estimator{M}{\Pi} \left( h_{[1,+\infty[}\big|0\right)$ estimators, respectively. We use $M=10^4$ points for each estimate and  $K=10^5$ estimates for both statistical variances $\bar{\mathrm{var}}^K$.}
 \label{fig:reductionfactor}
\end{figure}

\section{Estimation of the free energy along a reaction coordinate}

\label{application:lj38}

\subsection{$\xi$-AR estimator}

In many problems, the quantity of interest is the free energy along an internal reaction coordinate $\xi(q)$ rather than along an external parameter. When the reaction coordinate and the external parameter are coupled harmonically, the free energy along $\xi(q)$ similarly appears as the co-logarithm of a total expectation
\begin{equation}\label{def:fexi}
 \pmff(\xis) = - \ln \e \left[\2_{\xis}\circ\xi\right]
\end{equation}
where $\xis$ is a possible value of the reaction coordinate $\xi(q)$, recalling that $\2_{\xis}$ denotes the indicator function associated with the bin of $\xis$. The measured thermodynamic expectation is the probability to observe the reaction coordinate taking value $\xi^\star$. Substituting $\2_\xis\circ \xi(q)$ for the observable $\obs(\lambda,q)$ in Eq.~\eqref{estimator:total_general}, one obtains the AR estimator of total expectation $\e\left[ \2_{\xis}\right]$ 
\begin{eqnarray} \nonumber \estimator{M}{\Pi} (\2_{\xis} ) & = & \dfrac{\frac{1}{M} \sum_{m=1}^M  \2_\xis\circ\xi(\brc_m) \dspi_\aux (\brc_m)  }{\frac{1}{M} \sum_{m=1}^M \dspi_\aux (\brc_m)}. 
\end{eqnarray}
The corresponding asymptotic variance is given in Eq.~\ref{estimator:total_bobs} of Appendix~\ref{appendix:proof}. 
To write the function $\dspi_\aux$ explicitly, let introduce the following effective restraining potential 
\begin{equation}\label{def:raux}
 \raux(\xis) = \ln \sum_{\varlambda\in\varLambda} \exp\left[ \aux(\varlambda)- \restrain(\varlambda,\xis) \right]. 
\end{equation}
The marginal probability of $\rc$ therefore writes $\bpaux(\rc)=\exp\left[-\pot_0(\rc) +\raux \circ\xi(\rc) - \n_\aux  \right]$.  Thanks to the potential $\varepsilon$, the identity  $\ro \circ \xi (\rc) = 0$ holds whatever $\rc\in \mathcal{Q}$ and we have 
\[\dspi_\aux(\rc) = \frac{\p^{\mathcal{Q}}_0(\rc)}{\bpaux(\rc)} = \exp\left[ - \raux\circ \xi(\rc) \right] .\] 
Hence, the AR estimator further simplifies into
\begin{equation}\label{estimator:total_RC}
 \estimator{M}{\Pi}(\2_\xis) = \exp \left[-\raux (\xis)\right] \frac{\tfrac{1}{M} \sum_{m=1}^M \2_\xis(\xi_m)}{\tfrac{1}{M} \sum_{m=1}^M \exp\left[-\raux(\xi_m)\right] },  
\end{equation}
where $\xi_m = \xi(\brc_m)$. 
The estimator of the free energy is eventually obtained by taking the co-logarithm 
\begin{equation}\label{estimator:xi-ar}
\widehat{\pmff(\xis)}^M_{\mathrm{AR}} = \raux(\xis) - \ln  \iam\left( \1_\xis\circ\xi \right) + \ln \iam \left( \exp \left[-\raux\circ\xi \right]\right). 
\end{equation}
This specific estimator, denoted $\xi$-AR estimator in the following, has been applied to free energy calculations associated with vacancy migration in $\alpha$-iron~\cite{cao:2014} and molecular folding.~\cite{lindahl:2014} 
Recently, it has been proposed to evaluate the free energy $\pmff(\xis)$ using an alternative estimator~\cite{lesage:2016} referred to as \textit{corrected $z$-average restraint} (CZAR) wherein $\xis\equiv z$. The CZAR estimator only differs from the $\xi$-AR estimator in the way the effective restraining potential is evaluated in~\eqref{estimator:xi-ar}. In CZAR, $\raux(\xis)$ is evaluated through thermodynamic integration, i.e. by integrating estimates of its derivative $\raux^\prime(\xis)$. This is done by casting in the effective restraining gradient in the form of a conditional expectation 
\begin{eqnarray}
 \raux^\prime(\xis) &=& \partial_\xis \ln \sum\nolimits_{\lambda\in\varLambda}  \exp\left[ \aux(\lambda)-\restrain(\lambda,\xis) \right] \nonumber \\ &= & -\sum\nolimits_{\lambda\in\varLambda} \partial_\xis\restrain(\lambda,\xis) \pi_\aux^\lambda (\xis ) \nonumber \\
 &=& -\eaux \left[ \partial_\xis\restrain(\cdot,\xis)  |\xis\right], \label{rauprime}
\end{eqnarray}
where the conditional probability given $\xis$ reads
\begin{equation}\label{gibbs:sampler}
\pi^\lambda_\aux(\xis) = \exp\left[ \aux(\lambda)-\restrain(\lambda,\xis) - \raux(\xis) \right]. 
\end{equation}
The effective restraining gradient in CZAR is estimated using the $\estimator{M}{H}$ estimator in which binning is performed using indicator function  $\2_\xis$ instead of $\1_\lambda$: 
\begin{equation}\label{estimator:czar}
 \estimator{M}{H}\left(\partial_\xis \restrain |\xi=\xis  \right) = 
\beta\kappa \left\lbrace 
\frac{\sum_{m=1}^M\varlambda_m \2_{\xis}(\xi_m)}{\sum_{m=1}^M\2_{\xis}(\xi_m)}
-\xis  \right\rbrace,    
\end{equation}
where we resorted to $\partial_\xis \restrain(\varlambda,\xis)= \beta\kappa (\xis- \varlambda)$, omitting the correcting force $\varepsilon^\prime(\xis)$. 
Conditioning the binning estimator in Eq.~\eqref{estimator:czar} yields 
 \begin{equation} \label{estimator:cczar}
 \estimator{M}{\Pi}\left(\partial_\xis \restrain |\xis  \right) = 
\beta\kappa \left\lbrace 
\frac{\sum_{m=1}^M \eaux\left[ \varlambda |\xis\right] \2_{\xis}(\xi_m)}{\sum_{m=1}^M\2_{\xis}(\xi_m)} -\xis  \right\rbrace  = - \raux^\prime(\xis),    
\end{equation}
which, once integrated, leads to the formulation of the $\xi$-AR estimator. Because the asymptotic variance of $\estimator{M}{\Pi}\left(\partial_\xis \restrain |\xis  \right)$ in Eq.~\eqref{estimator:cczar} is zero, one concludes, owing to the delta method, that the CZAR estimator exhibits a larger asymptotic variance than the $\xi$-AR estimator~\eqref{estimator:xi-ar}. 

To quantify the cumulated error in the estimation of $\Delta \raux = \raux(\xi_\mathrm{max}) - \raux(\xi_\mathrm{min})$ with CZAR method, let first write down the asymptotic variance of estimator $\estimator{M}{H}\left(\partial_\xis \restrain |\xis  \right)$: 
\begin{eqnarray}\label{variance:czar}
\sigma_\aux^2\left(\estimator{\infty}{H}(\partial_\xis \restrain |\xis) \right) = \frac{ \vara \left[ \partial_\xis \restrain |\xis\right]}{\eaux\left[\2_\xis \right]}, 
\end{eqnarray}
where the variance is conditional on $\xis$. Next, we assume that $\raux$ is integrated over $\Xi$, a discrete set of evenly spaced values along $\Delta \xi =\xi_\mathrm{max}-\xi_\mathrm{min}$ interval. Let $\|\Xi\|$ be the cardinal of $\Xi$. Then, 
$\Delta \xi/\|\Xi \|$ is the constant spacing between consecutive values and the cumulated asymptotic error writes 
\begin{equation}\label{error:czar}
 \sigma(\Delta \raux) = \frac{\Delta \xi \beta\kappa}{\|\Xi\| } \sqrt{\sum\nolimits_{\xis \in\Xi} \frac{\vara\left[\varlambda|\xis\right]}{\eaux\left[\2_\xis \right]}}. 
\end{equation}
The $\varlambda$ variable in each conditional variance is assumed to be i.i.d. This assumption is verified when a Gibbs sampler directly generates the $\varlambda$ values from the conditional probability~\eqref{gibbs:sampler}. 
In contrast, residual metastability usually persists along $\xi$ and the sampled values $\xi(q)$ will be highly correlated. The dominating contribution to the asymptotic variance of the CZAR estimator is expected to arise from the $\iam\left( \1_\xis\circ\xi \right)$ term. Thus, the cumulated error~\eqref{error:czar} should not be significant in practical applications. We illustrate this point in next subsection. 

\subsection{The 38-atom Lennard-Jones cluster and $Q_4$ bond-orientational order parameter}
\label{subsec:arhc}

We consider the problem of calculating the free energy along the $Q_4$ bond-orientational order parameter~\cite{steinhardt:1983} in the 38-atom Lennard-Jones (LJ) cluster. 
The LJ potential reads 
\begin{equation}
 \mathcal{V}_\mathrm{LJ}(\rc) = 4\epsilon\sum_{i>j} \left[ r_{ij}^{-12} - r_{ij}^{-6} \right], 
\end{equation}
where $r_{ij}= \| \rcm_j -\rcm_i \| /\sigma$ is the reduced distance separating atoms  $\rcm_i$ and $\rcm_j$.  
LJ reduced units of length, energy and mass ($\sigma$ = 1, $\epsilon$ = 1, $m$ = 1) will be used in the following. LJ$_{38}$ undergoes a two-stage phase change with increasing
temperature.~\cite{neirotti:2000a,neirotti:2000b} A solid-solid transition between the
truncated octahedral funnel and the icosahedral funnel
occurs near T$_\mathrm{ss}$=0.12, melting follows near T$_\mathrm{sl}$= 0.17. As in other finite size systems,~\cite{wales:2003} the transitions are not sharp but gradual. The most stable octahedral and icosahedral structures of the 38-atom cluster are a truncated octahedron with energy $E_0 = -173.9284$ (global minimum) and an incomplete icosahedron with energy $E_1 = -173.2524$, respectively.
The order parameter $Q_4$ is a convenient collective variable to distinguish between the cubic structure favored at low temperatures and the icosahedral isomers above $T_{\mathrm{ss}}$. The values of $Q_4$ typically range from $0.002~\--~0.06 $ (icosahedral structures)  
to $0.19$ (octahedral structures). 
The ABF method is implemented in the harmonic expanded ensemble resorting to the AR estimator in order to compute the free energy along the external parameter $\lambda$ coupled harmonically to $Q_4$. The standard ABF method~\cite{darve:2001,darve:2008,comer:2014} cannot be performed directly along $Q_4$ because the second derivatives of the $Q_4$ are not available. The restraining potential is set to $\frac{\kappa}{2T_{\mathrm{ref}}} \| Q_4(\rc)-\lambda\|^2$ and the extended potential exhibits the form given in~Eq.~\eqref{eq:quadratic}, so that the free energy derivative $\pmf^\prime(\lambda)$ is constructed from the following identity 
\begin{equation} \label{biased_ratio}
\e\left[ \partial_\lambda \pot|\lambda\right] = \frac{\kappa}{T_\mathrm{ref}}\left[ \lambda - \dfrac{\beaux\left[  \pi_\aux^\lambda Q_4\right]}{\beaux\left[\pi^\lambda_\aux \right]} \right]. 
\end{equation}
We set $\kappa=10^{4}$ and $T_\mathrm{ref}=0.15$. At this temperature, spontaneous structural transitions cannot be observed on the simulation timescale when $\aux(\lambda)$ is set to 0 (no bias). The implemented ABF method, described in Algorithm~1 of the \textcolor{blue}{Supplementary Material}, adapt the bias $\aux(\lambda)$ on the free energy $\pmf(\lambda)$ via the running estimate of its gradient. One collateral advantage of conditioning is to facilitate the sampling process, as it needs not propagating the external parameter. Twelve replicas of the system are propagated using Langevin dynamics~\cite{cao:2014} (a Metropolis algorithm is also implementable) to sample the marginal probability of $q$ directly. 
The force acting on a replica is the gradient of the involved marginal probability (see Eq.~\ref{eq:quadratic})
\begin{eqnarray}
 \nabla_q \ln \bpaux (q) & = & - \nabla_q \pot_0 (q) + \nabla_q Q_4(q) \raux^\prime \circ  Q_4(q) , 
\end{eqnarray}
where $\pot_0 = \mathpzc{V}_{\mathrm{LJ}}/T_\mathrm{ref}$ 	and $\raux^\prime$ is detailed in Eq.~\eqref{rauprime}. 
The time-step $\tau = 5\cdot 10^{-5}\; \rm{(lju)}$ is chosen very small so as to keep discretization errors negligible. 
The simulation length $M=10^8$ and replica number $K=12$ enable  the auxiliary biasing force to converge. Replicas are handled in parallel using a multi-core processor. They share the same biasing force during both the learning and subsequent production runs.
We next freeze the biasing potential to the converged free energy ${A}(\lambda)=\aux_M(\lambda)$ and implement Algorithm~2 of the \textcolor{blue}{Supplementary Material} to check that the  
the sampling is homogeneous along $Q_4$ parameter despite the persistence of strong correlations, as shown in Fig.~\ref{fig:homogeneous}. 
\begin{figure}
\centering
\includegraphics[width=0.98\textwidth]{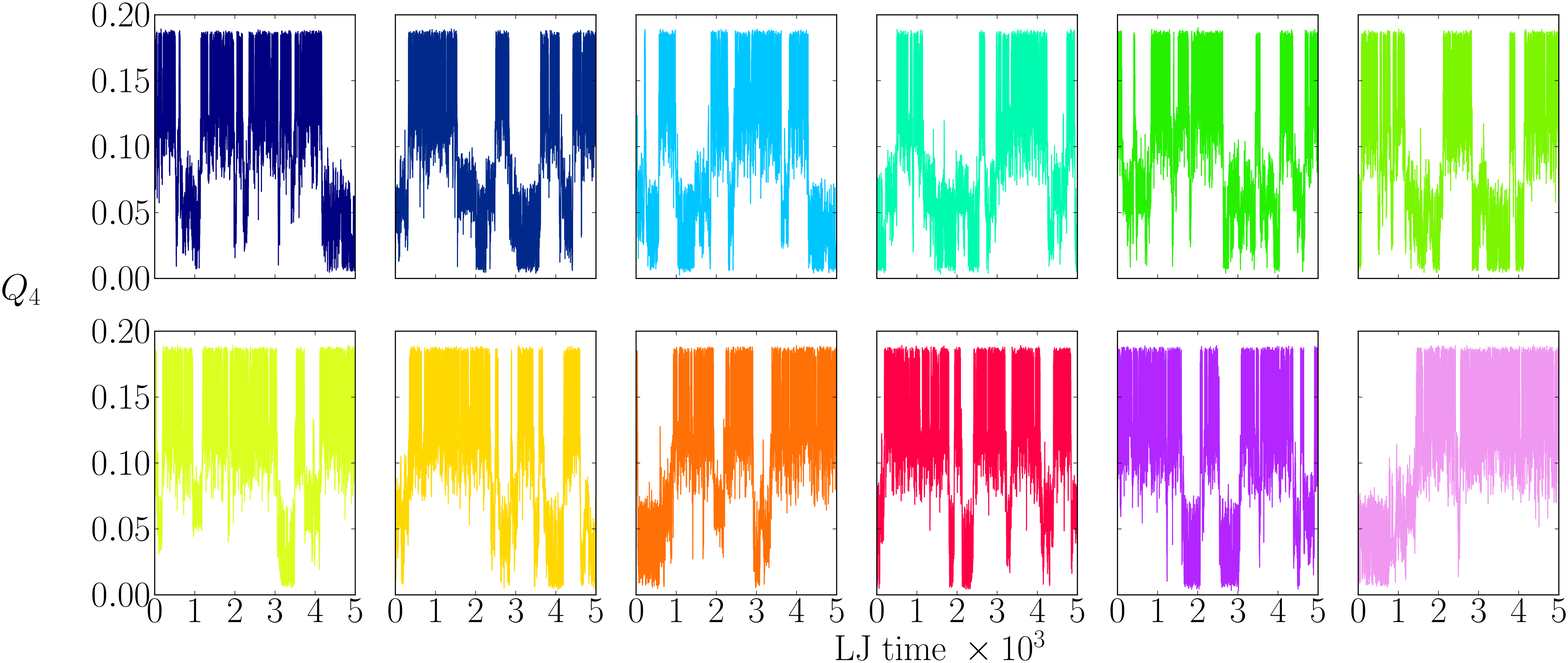}
\caption{Evolution of the $Q_4$ order parameter during the production run for each replica. Each panel corresponds to a distinct replica.}
\label{fig:homogeneous}
\end{figure}
With Algorithm~2, we also construct 40 (biased) histograms of $Q_4^\star$, using estimates $\mathtt{I}^{KM}(\2_{Q_4^\star})$ of $\e_{A}\left[\1_{Q_4^\star}\right]$ with $K=12$, $M=10^{8}$, a bin size of $2\cdot 10^{-4}$ and $Q_4^\star$ ranging from 0 to 0.2.   
The free energy along $Q_4$ order parameter is then estimated using the $\xi$-AR estimator and the 40 harvested histograms denoted by $\mathrm{P}_{A}(Q_4^\star)$. The obtained results are plotted in the three panels of Fig.~\ref{fig:profile}. The brackets $\langle \rangle$ indicate averaging over the series of 40 independent simulations. 
Thus, the displayed quantity $\langle B_A(Q_4) \rangle$ in Fig.~\ref{fig:profile}.a represents the (averaged and scaled) effective biasing potential, $\beta^{-1} \raux\left(\mathcal{Q}_4\right)$. The displayed histogram of $Q_4$ in Fig.~\ref{fig:profile}.b represents the averaged of the  $\mathrm{P}_{A}(Q_4^\star)$ estimates of $\e_{A}\left[ \2_{Q_4^\star} \right]$. 
The free energies displayed in Fig.~\ref{fig:profile}.c are estimated from $F(Q_4^\star)=B_{\widehat{A}}(Q_4^\star) - \beta^{-1} \ln \mathrm{P}_{A}(Q_4^\star)$. They correspond to the unreduced free energies $\beta^{-1}\pmff(Q_4^\star)$. 
\begin{figure}
\centering
\includegraphics[width=0.7\textwidth]{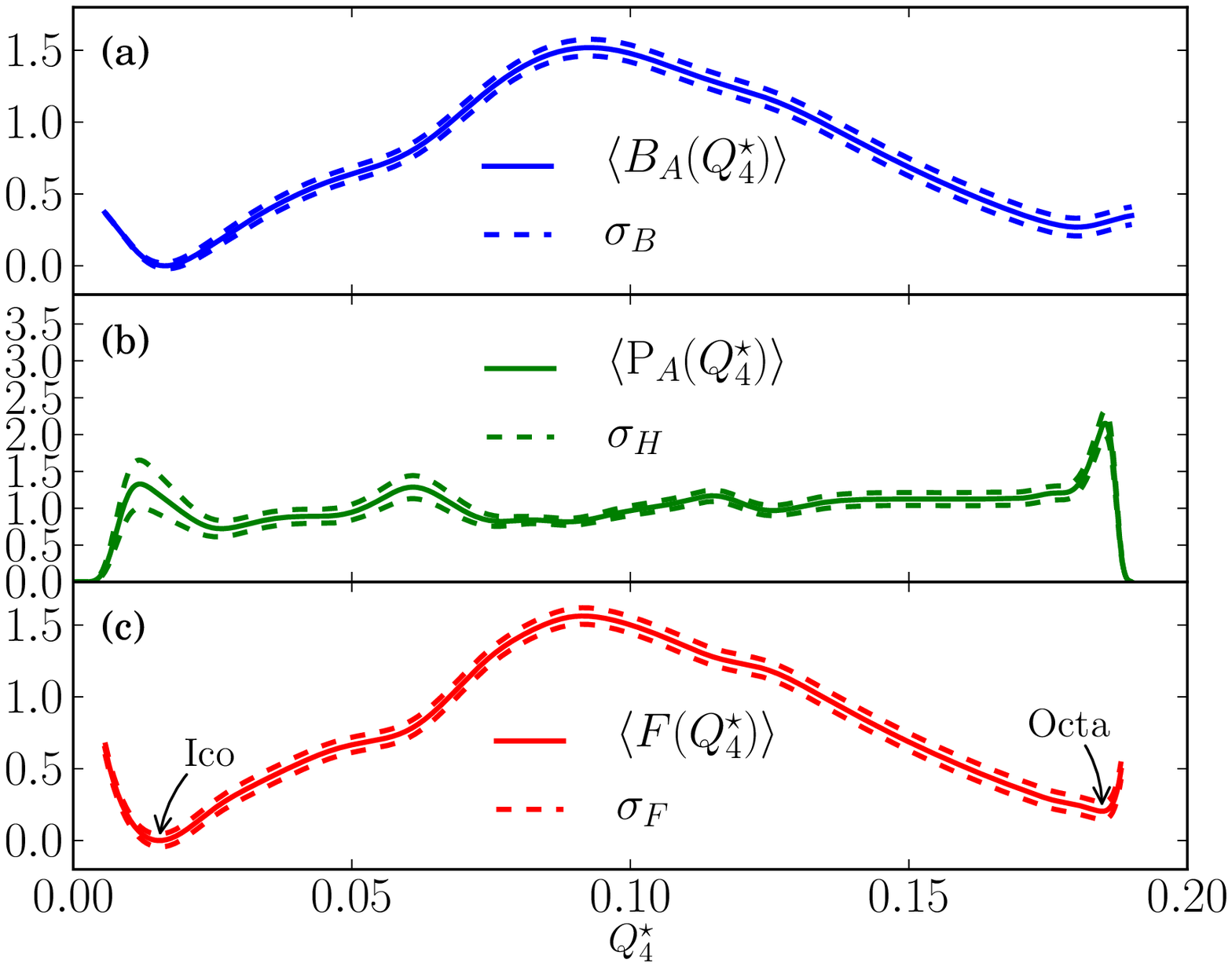} 
\caption{(a) mean effective biasing potential $\langle B_{A}(Q^\star_4) \rangle$ shifted by its minimum value; (b) $\langle \mathrm{P}_{A}(Q_4^\star) \rangle$ is the mean estimate of $\e_{A}\left[\1_{Q^\star_4}\right]$ histogram; (c) mean estimate of free energy $\langle F(Q^\star_4)= B_{A}(Q_4^\star) - \beta^{-1} \ln \mathrm{P}_{A}(Q_4^\star) \rangle$ shifted by its minimum value. Means and displayed standard deviations $\sigma_B$, $\sigma_H$ and $\sigma_\mathcal{F}$ (for 68\%-confidence intervals) are evaluated from 40 simulations.}
\label{fig:profile}
\end{figure}

The variance reduction compared to CZAR method is too small to be measurable from 40 simulations. Instead, we estimate the order of magnitude of the error that is made in the estimation of $\Delta\raux$ with the CZAR method.   
Neglecting the variations of the biasing forces along $\lambda$, which are much smaller than the value of $\kappa/T_\mathrm{ref}$, the conditional variance in Eq.~\eqref{error:czar} simplifies into $\mathbbm{V}_0(\varlambda|Q_4^\star)$ which is equal to $T_{\mathrm{ref}}/\kappa$. If we additionally neglect the variations of the histogram $\mathrm{P}_{A}(Q_4^\star)$ and 
assume that $\e_{A}\left(\2_{Q_4^\star}\right)= \|\Xi\|^{-1}$ between the icosahedral and octahedral structures in Eq.~\eqref{error:czar}, then the integrated error between the two structures simplifies into $\Delta Q_4 /\sqrt{T_\mathrm{ref}\kappa}$. 
The expected standard deviation on $\Delta B_{A} = T_\mathrm{ref} \Delta \mathpzc{B}_A$ is therefore $\Delta Q_4 \sqrt{T_\mathrm{ref}/(\kappa MK)}=2.0\cdot 10^{-4}$ wherein  $M K = 1.2\cdot10^{9}$ is the total number of sampled points. This is less than two orders of magnitude than $6.1\cdot 10^{-2}$, the standard deviation that is obtained from the 40 estimates of the free energy difference $\Delta F$.

\section{Characterization of structural transition}\label{sec:characterization}

From the free energy profile along $Q_4$, the occurrence probabilities of the icosahedral and octahedral structures can be evaluated. For these two structures, the Landau free energies $\mathtt{A}({\mathrm{ico}}|T_\mathrm{ref}^{-1})$ and $\mathtt{A}({\mathrm{octa}}|T_\mathrm{ref}^{-1}))$, defined as minus the logarithm of their respective occurrence probabilities, can be directly evaluated. We now wish to compute the two Landau free energies at other temperatures so as to characterize $T_\mathrm{ss}$, the \textit{solid-solid} structural transition temperature for which 
$\mathtt{A}({\mathrm{ico}}|T_\mathrm{ss}^{-1}) = \mathtt{A}({\mathrm{octa}}|T_\mathrm{ss}^{-1})$. 
Unfortunately, ABF simulations along $Q_4$ do not converge at temperature lower than $0.13$, meaning that $Q_4$ becomes a bad reaction coordinate and that $T_\mathrm{ss}$ can not be determined this way. 

The problem is solved by evaluating $\mathtt{A}(\mathtt{x}|T^{-1})$ as a function of inverse temperature separately for the two structures $\mathtt{x}$ of set $\mathtt{X} \triangleq \left\{\mathtt{ico},\mathtt{octa}\right\}$ from the knowledge of $\pmf(T^{-1}|\mathtt{ico})$ and $\pmf(T^{-1}|\mathtt{octa})$, the free energies along the inverse temperature given the structure. The identity connecting the two kinds of free energies corresponds to  Bayes formula expressed in log-space
\begin{equation}\label{Lambdaxgt}
 \mathtt{A}(\mathtt{x} | T^{-1}) = \pmf (T^{-1} | \mathtt{x} ) + \mathtt{A}(\mathtt{x}) -\pmf(T^{-1}),
\end{equation}
where $\mathtt{A}(\mathtt{x})$ is the marginal colog-probability of structure $\mathtt{x}$ in the expanded ensemble. 
Subtracting Bayes formula at $T^{-1}_\mathrm{ref}$ in log-space from Eq.~\eqref{Lambdaxgt} cancels the marginal colog-probability of $\mathtt{x}$, which leads to the desired Landau free energy
\begin{equation}\nonumber
 \mathtt{A}\left(\mathtt{x} | T^{-1}\right) = \pmf \left(T^{-1} | \mathtt{x} \right) - \pmf\left(T^{-1}_\mathrm{ref}|\mathtt{{x}}\right) + \mathtt{A}\left(\mathtt{{x}} | T^{-1}_\mathrm{ref}\right) + \pmf\left(T^{-1}_\mathrm{ref}\right)-\pmf\left(T^{-1}\right).  
\end{equation} 
The free energy difference $\pmf(T^{-1}_\mathrm{ref})-\pmf(T^{-1})$ is a common contribution to all structures. It does not affect phase equilibrium properties and merely serves as a normalizing constant. It may be calculated from the relation
\begin{equation}
 \pmf\left(T^{-1}_\mathrm{ref}\right)-\pmf\left(T^{-1}\right) = \ln \sum_{\tilde{\mathtt{x}}\in \mathtt{X}} \exp \left[-\pmf \left(T^{-1} | \tilde{\mathtt{x}} \right)+\pmf\left(T^{-1}_\mathrm{ref}|\tilde{\mathtt{x}}\right)-\mathtt{A}\left(\tilde{\mathtt{x}}|T^{-1}_\mathrm{ref}\right) \right]. 
\end{equation}
In our case study, phase equilibrium is reached when the two structures are equi-probable: 
\begin{equation}\nonumber
\mathtt{A}\left(\mathtt{ico}|T^{-1}_\mathrm{ss}\right) = \mathtt{A}\left(\mathtt{octa}|T^{-1}_\mathrm{ss}\right). 
\end{equation}

We proceed as follows: after setting the external parameter to the inverse temperature $\lambda=T^{-1}$, we perform two supplementary simulations  to compute $\pmf(T^{-1}|\mathtt{ico})$ and $\pmf(T^{-1}|\mathtt{octa})$ with the AR estimator and the external parameter ranging from $\lambda_\mathrm{min}=6.25$ to $\lambda_\mathrm{max}=14$, taking advantage of the fact that structural transitions do not occur spontaneously in the range of involved temperatures. 
The transition temperature is characterized by the intersection point of our two Landau free energies as a function of inverse temperature, as shown in Fig.~\ref{fig:structural_transition}. The error that is made in the estimation of the transition temperature essentially arises from the uncertainty in the evaluation of $\Delta F$ at $T_{\mathrm{ref}}$. 

\begin{figure}
\centering
\includegraphics[width=0.75\textwidth]{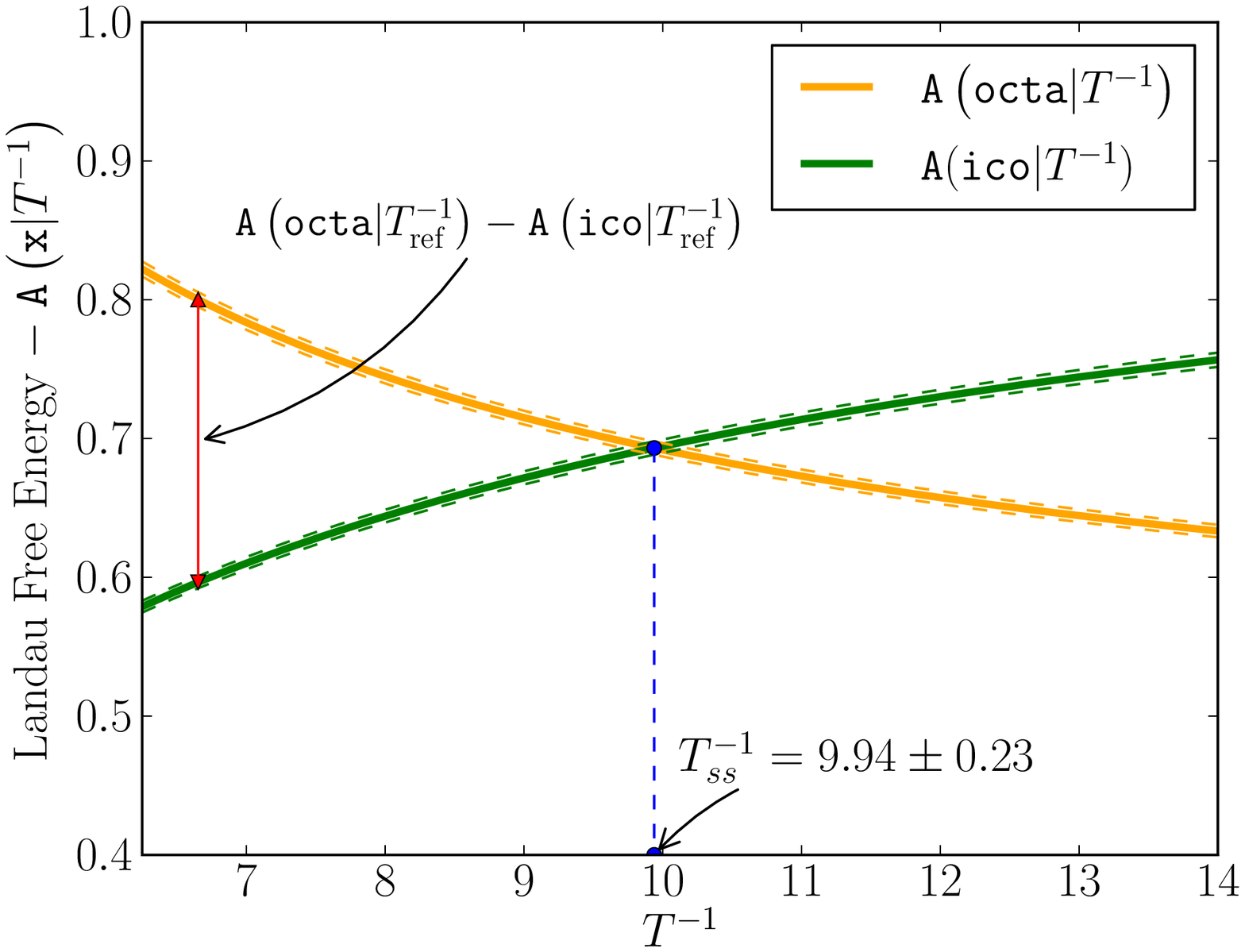}
\caption{Estimation of the free energies of the icosahedral and octahedral structures as a function of inverse temperature.}
\label{fig:structural_transition}
\end{figure}

\section{Summary and Conclusion}

Molecular simulations methods usually requires to vary the value of an external parameter $\lambda$ within a specified range $\varLambda$. In such problems, enhanced ergodicity is achieved if simulations are performed in an expanded ensemble wherein the external parameter behaves like an additional coordinate subject to an auxiliary biasing force. The various methods of computing thermodynamic expectations and free energies which can be used in the expanded ensemble framework are reviewed below
\begin{itemize}
\item[(a)] \emph{Thermodynamic integration} to estimate the derivative of the free energy along $\lambda$ and then to obtain the free energy through integration by numerical quadrature. This approach resorts to a binning estimator and is used in combination with adaptive biasing force method to construct an auxiliary biasing potential that converges towards the free energy. 

\item[(b)] \emph{Thermodynamic occupation} to directly extract the free energy from occupation probabilities. This approach is based on a binning estimator and may also be employed to adapt the biasing potential directly on the free energy along $\lambda$. 

\item[(c)] \emph{Free energy perturbation} to extract the free energies from partition function ratios using a simple standard reweighting estimator. This approach requires substantial overlaps between the various $\lambda$-samples and has been rarely used. 

\item[(d)] Post-treatment procedures such as the \emph{self-consistent reweighting} estimator based on Bennett acceptance ratio method (STWHAM) or the \emph{histogram reweighting} estimator based on thermodynamic integration along the reaction coordinate (CZAR). The goal is to estimate the desired thermodynamic property more accurately than with estimators (a-c). The latter procedure is employed to compute the free energy along reaction coordinates using mechanical restraints. 
\end{itemize}
Remarkably, conditioning over the external parameter in the expanded ensemble provides us with a unifying methodology: once conditioned, the various estimators associated with methods (a-d) become equivalent and exhibit systematically reduced statistical variances. In practice, simulations performed with the conditioned estimators are facilitated by the fact that the external parameter need not being sampled. Characterizing the structural transition temperature in LJ$_{38}$ has been shown to be a simple task. The advocated conditioning approach is well positioned to further extend the range of applicability of Monte Carlo and molecular dynamics techniques to the calculation of free energies and thermodynamic properties in condensed matter systems. 

The avenue for future research on the conditioned estimators will involve formulating and constructing biasing potentials that allow optimal variance reduction for estimating thermodynamic expectations within a given set of external parameter values. 

\section*{Supplementary Material}

See \textcolor{blue}{Supplementary Material} for a description of the algorithms used in  Section~\ref{application:lj38}. 

\begin{acknowledgements}
The authors thank Benjamin Jourdain for drawing our attention on the connection between conditioning and adiabatic reweighting and for indicating us the delta method (Theorem A.2 given in Appendix A). We are also grateful to John Chodera, Lingling Cao, Gabriel Stoltz and a reviewer for many helpful comments and discussions. This work was supported by the Energy oriented Centre of Excellence
(EoCoE), grant agreement number 676629, funded within the Horizon2020
framework of the European Union.
\end{acknowledgements}
\appendix
\section{Asymptotic variances of estimators}
\label{appendix:proof}
\subsection*{Delta method}

Since expectations are considered with respect to the expanded ensemble, we write the dependence on the biasing potential $\aux(\lambda)$ explicitly when needed and use the notation $\vara$ adopted throughout the article for the variance. The lemma below establishes a general and useful property of covariance matrices 
\begin{lemma}\label{lemma}
Let $\Gamma$ be the covariance matrix of a random vector $Y$ taking its values in $\mathbbm{R}^d$  and $u \in \mathbbm{R}^d$ be a constant vector. Then, we have  
\begin{equation}\label{epcov}
 u^T \Gamma u = \vara \left[ u^T Y  \right]. 
\end{equation}
\end{lemma}

\begin{proof}
Let write $Y=(Y^i)_{1 \leq i \leq d}$ and $u=(u^i)_{1 \leq i \leq d}$. By definition, element $\Gamma_{ij}$ of covariance matrix $\Gamma$ is equal to $\cov \left(Y^i,Y^j\right)$
the covariance of the one-dimensional random variables $Y^i$ and $Y^j$, defined by 
\begin{equation}\nonumber
 \cov [Y^i,Y^j] = \eaux\left[ Y^iY^j \right] - \eaux\left[ Y^i \right] \eaux\left[ Y^j \right]. 
\end{equation}
Since the covariance has scalar product properties, we have the following sequence of equalities
\begin{multline*}
 u^T \Gamma u = \sum_{i=1}^d\sum_{j=1}^d u^i \cov(Y^i,Y^j)u^j = \\=  \cov (u^T Y,u^T Y) = \vara (u^TY).
\end{multline*}
\end{proof}

The delta method will allow us to characterize the asymptotic variance of all aforementioned estimators. It consists of applying the generalized central limit theorem that follows:  
\begin{theorem}
\label{thetheorem}
Let $\left\{ Y_m \right\}_{m\geq 1}$ be a sequence of independent, identically distributed and square integrable random vectors taking their values in $\mathbbm{R}^d$. Let $\mu$ and $\Gamma$ respectively denote the expected vector and the covariance matrix of the $Y_m$ and $\overline{Y}^M=\tfrac{1}{M} \sum_{m=1}^M Y_m $.   Let $g: \mathbbm{R}^d \mapsto \mathbbm{R}$ be a function that is differentiable at $\mu$. Then, we have the following convergence in law \begin{equation}\nonumber
 \sqrt{M} \left( g(\overline{Y}^M) - g(\mu)\right) \convlaw \mathcal{N} \big(0, \nabla g(\mu)^T \Gamma \nabla g(\mu) \big). 
 \end{equation}
\end{theorem}
We refer the reader to Ref.~\citenum{jourdain:2013} for a proof of this classical result. The variance of the Gaussian variable above is called the asymptotic variance of random variable $g(Y)$. 

Theorem~\ref{thetheorem}  is applied together with Eq.~\eqref{epcov} to express the asymptotic variance of the estimators below as a variance of random variable $\nabla g(\mu)^T Y$: 
\begin{equation}\label{eq:asympvar}
 \nabla g(\mu)^T \Gamma \nabla g (\mu) = \vara \left( \nabla g(\mu)^T Y  \right). 
\end{equation}    
\subsection*{Estimation of conditional expectations}

For the $\estimator{M}{G} (\obs|\lambda)$ estimator given in Eq.~\ref{estimator:generic} and conditioned in subsection~\ref{estimators:conditional_expectation}, we set $Y_m = \left(\g^\lambda_\aux(\varlambda_m,\brc_m) \obs(\lambda,\brc_m) , \g^\lambda_\aux(\varlambda_m,\brc_m) \right)^T \in \mathbbm{R}^2$, $g(r,s)=r/s$. We thus have
$\mu = \left( \eaux \left[\g^\lambda_\aux(\cdot,\cdot) \obs(\lambda,\cdot) \right] , \eaux \left[\g^\lambda_\aux \right] \right)^T$ and 
\begin{multline}\nonumber
\nabla g(\mu) = \left(\begin{array}{cc} 
1/ \eaux\left[ \g_\aux^\lambda \right]  \\
 -  \eaux \left[ \g_\aux^\lambda(\cdot,\cdot)  \obs(\lambda,\cdot) \right] \Big/\eaux\left[ \g_\aux^\lambda \right]^2
 \end{array}\right) \\
 =\dfrac{1}{\eaux\left[\g^\lambda_\aux\right]} \left(\begin{array}{cc} 
1  \\
 - \e \left[ \obs |\lambda \right]  
 \end{array}\right) 
 .\end{multline} 
Resorting to Eq.~\eqref{eq:asympvar}, the asymptotic variance $\nabla g(\mu)^T \Gamma \nabla g (\mu)$ is therefore 
\begin{multline}
\sigma_\aux^2 \left[ \estimator{M}{G} (\obs|\lambda) \right] =  \dfrac{1}{\eaux\left[\g^\lambda_\aux\right]^2} \vara \left[ \left(\begin{array}{c}
1 \\   - \e\left[\obs | \lambda \right]
\end{array}\right)^T 
\left(\begin{array}{c}
\g^\lambda_\aux(\cdot,\cdot) \obs(\lambda,\cdot) \\ \g^\lambda_\aux(\cdot,\cdot)  
\end{array}\right) \right] \nonumber 
\\ = \dfrac{\vara \left[\g_\aux^\lambda \obs^\lambda  \right]}{\lpaux(\lambda)^2}, \nonumber
\end{multline}
recalling that $\obs^\lambda (\rc) = \obs(\lambda,\rc) - \e\left[\obs|\lambda \right]$ and that $\eaux\left[\g^\lambda_\aux\right] = \lpaux(\lambda)$. The square root of the asymptotic variance  of $\estimator{M}{G} (\obs|\lambda)$ estimator is given in Eq.~\ref{asymptotic_error_phi}. Along the same line of reasoning, the asymptotic variance of adiabatic reweighting estimator $\estimator{M}{\Pi}(\obs|\lambda)$ can be deduced after substituting $\pi_\aux^\lambda$ for $\g_\aux^\lambda$. 

\subsection*{Estimation of total expectations}

For the generic estimator $\estimator{M}{G} (\obs)$ of total expectation $\e\left[ \obs \right]$ that is given in Eq.~\eqref{estimator:generic_total}, we set $g(r,s)=r/s$ together with 
\begin{eqnarray}\nonumber
 &Y & = 
 \left( 
\begin{array}{cc}
\sum_{\lambda\in\varLambda} \obs(\lambda,\brc)  \dsg^{\lambda}_\aux (\varlambda,\brc)  \\
\dsg_\aux (\varlambda,\brc)
\end{array} 
 \right), \\[0.5cm] \nonumber
 & \eaux[Y] &= \left( \begin{array}{cc} \eaux \left[\sum_{\lambda\in\varLambda}\obs(\lambda,\cdot) \dsg_\aux^\lambda(\cdot,\cdot) \right] \\ \eaux \left[\dsg_\aux \right] \end{array} \right)  , \\[0.5cm] 
& \nabla g \left(\eaux [Y] \right) & = \dfrac{1}{\eaux [\dsg_\aux ]} \left( 
\begin{array}{cc}
1\\
-\e [\obs]
\end{array} 
 \right), \nonumber
\end{eqnarray}
where $\dsg^\lambda_\aux(\varlambda,\brc) = \exp[-\aux(\lambda)] \g^{\lambda}_\aux (\varlambda,\brc)$ and $\dsg_\aux = \sum_{\lambda\in\varLambda} \dsg^\lambda_\aux $. The scalar product of the first and third vectors above is 
 \begin{equation} 
 Y^T \nabla g\left(\eaux [Y] \right) =  \sum_{\lambda\in\varLambda} \obs(\lambda,\brc) \frac{ \dsg^{\lambda}_\aux (\varlambda,\brc)} {\eaux [\dsg_\aux]} - \e[\obs] \frac{\dsg_\aux(\varlambda,\brc) } {\eaux [\dsg_\aux]}. 
 \end{equation}

This expression enables one to obtain the asymptotic variance of the estimator by application of Lemma~\ref{lemma} and Theorem~\ref{thetheorem}
\begin{equation}\nonumber
 \sigmaux^2 \left[\estimator{\infty}{G} (\obs) \right] = \vara  \left[\sum_{\lambda\in\varLambda}  \obs(\lambda,\cdot) \frac{\dsg^{\lambda}_\aux (\cdot,\cdot)} {\eaux [\dsg_\aux] }   - \e[\obs] \frac{\dsg_\aux(\cdot,\cdot)} {\eaux [\dsg_\aux] }  \right]. 
\end{equation}
The asymptotic variance of adiabatic reweighting estimator $\estimator{M}{\Pi} (\obs)$ can be deduced by substituting, $\dspi^\lambda_\aux = \eaux[\dsg^\lambda_\aux |\rc ]$ for $\dsg^\lambda_\aux$, $\dspi_\aux$ for $\dsg_\aux$ and following the same line of reasoning: 
\begin{equation}
 \sigmaux^2 \left[\estimator{\infty}{\Pi} (\obs) \right] = \vara  \left[ \sum_{\lambda\in\varLambda}  \obs(\lambda,\cdot) \frac{\dspi^{\lambda}_\aux (\cdot)} {\eaux [\dspi_\aux] } - \e[\obs] \frac{\dspi_\aux(\cdot) } {\eaux [\dspi_\aux] } \right]. 
\end{equation}
When the observable is an indicator function $\2_{\xis}\circ\xi(q)$ associated with a reaction coordinate $\xi(q)$ that is independent of the external parameter, the asymptotic variance simplifies into
\begin{eqnarray}\label{estimator:total_bobs} \sigmaux^2 \left[\estimator{\infty}{\Pi} (\2_{\xis}\circ\xi) \right] &= &\vara  \left[{\2}^\varLambda_{\xis}\circ\xi \frac{ \dspi_\aux } {\eaux [\dspi_\aux] } \right], 
\end{eqnarray}
where $\2^\varLambda_{\xis}\circ \xi (\rc) = \2_{\xis} \circ\xi(\rc) - \e\left[\2_{\xis}\circ\xi \right]$. 
The formulation of estimator $\estimator{M}{G}$ for observables that are dependent on the external parameter is useful in the free energy context below. 

\subsection*{Estimation of free energies along $\lambda$}

As for the generic estimator of the free energy $\pmf(\lambda)$  considered in subsection~\ref{estimation_free_energy}, the asymptotic variance can be obtained by noticing that the free energy corresponds to the co-logarithm of the total expectation of observable $\1_\lambda$. Here, one applies theorem~\ref{thetheorem} with function $g$ set to the co-logarithm function and $\overline{Y}^M$ set to $\estimator{M}{G}(\1_\lambda)$, where $\mathrm{G}$ corresponds to $\mathrm{R}$ (standard reweighting) and $\mathrm{H}$ (binning) for the FEP and TO methods, respectively. This yields the following expression for the asymptotic variance of the corresponding free energy method
\begin{multline}\label{variance_log}
 \sigmaux^2\left[\widehat{\pmf(\lambda)}^\infty_{\mathrm{FEP/TO}} \right] =   \frac{\sigmaux^2 \left[\estimator{\infty}{G} (\1_\lambda) \right]}{\lpo(\lambda)^2}\\= \vara  \left[ \frac{\dsg^{\lambda}_\aux}{\lpo(\lambda) \eaux \left[\dsg_\aux \right] }  - \frac{\e[\1_\lambda]}{\lpo(\lambda)}\frac{ \dsg_\aux }{\eaux [\dsg_\aux] } \right]. \nonumber
\end{multline}
Noticing that $\e[\1_\lambda]$ is $\lpo(\lambda) $ and that $\lpo(\lambda) \eaux \left[\dsg_\aux \right]$ is $\eaux \left[ \dsg^\lambda_\aux \right]$, one eventually obtain the desired asymptotic variance in a more compact form
\begin{equation} \label{variance_fepto}
\sigmaux^2\left[\widehat{\pmf(\lambda)}^\infty_{\mathrm{FEP/TO}} \right] =  \vara \left[ \frac{\dsg^\lambda_\aux}{\eaux\left[ \dsg^\lambda_\aux \right]} - \frac{\dsg_\aux}{\eaux\left[\dsg_\aux \right]}\right]. 
\end{equation}
The asymptotic variance of the AR method similarly writes
\begin{multline}\label{variance_fear}
 \sigmaux^2\left[\widehat{\pmf(\lambda)}^\infty_{\mathrm{AR}} \right] =  \vara^\mathcal{Q} \left[ \frac{\eaux \left[\dsg^\lambda_\aux(\cdot,q) \big|\rc\right]}{\eaux\left[ \dsg^\lambda_\aux \right]} - \frac{\eaux \left[ \dsg_\aux(\cdot,q) |\rc\right]}{\eaux\left[\dsg_\aux \right]}\right] \\=  \vara \left[ \frac{\dspi_\aux^\lambda}{\eaux\left[ \dspi^\lambda_\aux \right]} - \frac{{\dspi}_\aux}{\eaux\left[{\dspi}_\aux \right]}\right]. 
\end{multline}

\section{Variance reduction through conditioning of self-consistent reweighting}
\label{appendix:STWHAM}

We prove here variance reduction when conditioning is done within self-consistent reweighting. 
Let first assume the state space $\mathcal{Q}$ be enumerable and introduce the indicator functions $H_q(\brc)$ and $Y_{\lambda,q}(\varlambda,\brc) = \1_\lambda(\varlambda) H_q(\brc)$. We define the corresponding multi-dimensional vectors as $H=\left\{ H_q \right\}_{q\in\mathcal{Q}}$ and $Y=\left\{ Y_{\lambda,q} \right\}_{\lambda\in\varLambda,q\in\mathcal{Q}}$. The averaged vector $\overline{Y}^M = \tfrac{1}{M}\sum_{m=1}^M Y_m$ constructed from the Markov chain $\{ \varlambda_m,\brc_m \rbrace_{1\leq m \leq M}$ encodes all the information that is necessary to evaluate any observable estimate through self-consistent reweighting~\eqref{estimator:mbar}. It contains in particular the information that is required to solve the associated set of self-consistent equations~\eqref{system:SC} wherein $M_\lambda = \bar{\1}^M_\lambda$. Since the Markov chain is generated according the $\paux(\varlambda,\brc)$ distribution, we have for all $\lambda \in \varLambda$
\begin{equation}
\eaux \left[ \overline{\1}^\infty_\lambda \big| \overline{H}^\infty \right] = \overline{\1}^\infty _\lambda  = \eaux\left[ \1_\lambda\right] =\lpaux(\lambda)
\end{equation}
and, more generally
\begin{equation}\label{eq:lemu}
 \eaux \left[ \overline{Y}^\infty \big| \overline{H}^\infty \right] = \overline{Y}^\infty = \eaux \left[ Y \right]. 
\end{equation}
The law of total covariance allows us to decompose the covariance matrix of $Y$ into the following sum
\begin{equation}\label{law:totalcovariance}
\Gamma \left[ Y \right] = \eaux \left[ \Gamma [ Y| H ]\right] + \Gamma \left[ \eaux[Y|H] \right]
\end{equation}
where $\Gamma [ Y| H ]$ is the conditional covariance of $Y$ given $H$ and $\Gamma \left[ \eaux[Y|H] \right]$ is the covariance of the conditionally expected value of $Y$ given $H$. Covariance matrices are semidefinite. Herein, we consider that the three covariance matrices are definite positive, otherwise the problem would be degenerate and sampling might not be necessary. The following strict inequality thus holds 
\begin{equation} \label{vitue}
 \nabla g (\overline{Y}^\infty)^T \eaux \left[ \Gamma [ Y| H ]\right] \nabla g (\overline{Y}^\infty) > 0, 
\end{equation}
where $g$ denote the smooth function $\overline{Y}^M \mapsto \estimator{M}{SC}(\obs | \lambda)$, i.e. the (unknown) function returning the estimate from the given data $\overline{Y}^M$. The estimate being unique, $\nabla g (\overline{Y}^\infty)$ is a non zero vector. The asymptotic variance of the self-consistent reweighting estimator writes
\begin{equation}
 \sigmaux^2 \left[ \estimator{\infty}{SC} (\obs|\lambda) \right] =  \nabla g (\overline{Y}^\infty)^T \Gamma [ Y ]\nabla g (\overline{Y}^\infty). 
\end{equation}
As a result of~\eqref{eq:lemu},~\eqref{law:totalcovariance} and~\eqref{vitue}, the asymptotic variance of the conditioned self-consistent reweighting estimator is systematically lowered 
\begin{equation}
 \nabla g (\overline{Y}^\infty)^T \Gamma [ \eaux \left[ Y| H\right] ]\nabla g (\overline{Y}^\infty) = \sigmaux^2 \left[ \estimator{\infty}{\Pi} (\obs|\lambda) \right]  <  \sigmaux^2 \left[ \estimator{\infty}{SC} (\obs|\lambda) \right]. 
\end{equation}

When the state space $\mathcal{Q}$ is not enumerable, we may first approximate the observable expectation using appropriate histograms of the energies $H_{\lambda,\pot}$ and of the observable values $H_{\lambda,\pot,\obs}$ (see Section 2.6 in Ref.~\citenum{chodera:2007}), perform the conditioning, compare the asymptotic limit and eventually consider the small bin-width limit to conclude that variance reduction still holds.

\section*{Algorithms}
\label{sec:abf}
The ABF method in expanded ensemble is described in Algorithm~\ref{algo:abf_ee}.  
The description that is given is general and encompasses both the harmonic and linear couplings. The external parameter couples either harmonically to the $Q_4$ coordinate or linearly to the potential energy $\mathcal{V}_\mathrm{LJ}(\rc)$ so as to scale the inverse temperature. $\mathcal{N}(\mu,\sigma)$ denotes the normal law of mean $\mu$ and standard deviation $\sigma$. We describe both Metropolis sampling and overdamped Langevin dynamics. In practice, we used the latter dynamics with a small time step of $10^{-4}$ LJ units. 

\begin{algorithm}[H]
\caption{ABF algorithm in expanded ensemble~\cite{cao:2014} with conditioning for $\lambda \in \varLambda$. The $K$ replicas of the system share the same biasing force. Replicas are propagated using either a Metropolis sampling or Langevin dynamics. The integration of the biasing force is performed by numerical quadrature based on trapezoidal rule. }
\label{algo:abf_ee}
\begin{algorithmic}
\vspace*{0.25cm}
\For {$\lambda \in \varLambda$}
\State $C(\lambda) \gets 0$, $D(\lambda) \gets 0$
\EndFor
\For {$m=1$ to $M$}          
\For {$k=1$ to $K$}
\State $\rc^k = \brc^k_{m-1} + \sqrt{2\tau} G^k_m$ with $G^k_m \sim \mathcal{N}(0,1)$ 
\If {Metropolis algorithm} 
\State $r \sim \mathbb{U}_{[0,1[}$ \Comment $r$ is drawn uniformly in $[0,1[$
\If {$r\leq \dfrac{\sum_{\lambda\in\varLambda} \exp\left[\aux_{m-1}(\lambda) - \pot(\lambda,\trc^k)\right] }{\sum_{\lambda\in\varLambda} \exp \left[\aux_{m-1}(\lambda) -\pot (\lambda,\brc^k_{m-1}) \right] } $} 
    \State $\brc^k_m\gets \trc^k$
\Else
   \State $\brc^k_m\gets \brc^k_{m-1}$
\EndIf
\ElsIf {Overdamped Langevin dynamics}
   \State $\brc^k_m = \trc^k + \tau \nabla_\brc \ln \sum\limits_{\lambda\in\varLambda} \exp\left[\aux_{m-1}(\lambda) - \pot(\lambda,\rc^k)\right]$ 
\EndIf
\EndFor
\For {$\lambda\in \varLambda$} 
\State $\pi_{\aux_{m-1}}^\lambda \big( \brc_m^k \big) =\dfrac{ \exp\left[\aux_{m-1}(\lambda) - \pot(\lambda,\trc^k)\right]}{ \sum_{\varlambda\in\varLambda} \exp\left[\aux_{m-1}(\varlambda) - \pot(\varlambda,\trc^k)\right]}$
\State
\State $C(\lambda) \gets C(\lambda) + \sum_{k=1}^K \partial_\lambda \pot\big(\lambda,\brc^k_m\big) \bpi_{\aux_{m-1}}^\lambda \big( \brc_m^k \big)$
\State $D(\lambda) \gets D(\lambda) + \sum_{k=1}^K  \bpi_{\aux_{m-1}}^\lambda \big( \brc_m^k \big)$
\State $\aux^\prime_{m} (\lambda) = \dfrac{C(\lambda)}{D(\lambda)}$ \Comment biasing force update
\State $\aux_{m} (\lambda) \leftarrow \int^{\lambda}_{\lambda^0} \aux^\prime_m(\varlambda)d\varlambda$ \Comment numerical integration 
\State $\aux_m(\lambda) \gets \aux_m(\lambda) + \ln \sum{\varlambda\in\varLambda} \exp\big[-\aux_m(\varlambda)\big]$
\EndFor
\EndFor
\end{algorithmic}
\end{algorithm}

Once the biasing force has converged to the corresponding free energy using Algorithm~\ref{algo:abf_ee}, the expected value of any observable can be estimated using the adiabatic reweighting estimator described in Algorithm~\ref{algo:expectation}. \\

\vspace*{0.0cm}

\begin{algorithm}[H]
\caption{Algorithm for sampling the marginal probability in the expanded ensemble with biasing potential $\wpmf(\lambda)$ homogeneous in time and for estimating the conditional and total expectations of observable $\obs(\varlambda,\brc)$ based on the conditioning procedure (adiabatic reweighting estimator).}
\label{algo:expectation}
\begin{algorithmic}
\vspace*{0.25cm}
\State $\tilde{C} \gets 0$, $\tilde{D} \gets 0$
\For {$\lambda \in \varLambda$}
\State $C(\lambda) \gets 0$, $D(\lambda) \gets 0$
\EndFor
\For {$m=1$ to $M$}          
\For {$k=1$ to $K$}
\State $\trc^k = \brc^k_{m-1} + \sqrt{2\tau} G^k_m$   with $G^k_m \sim \mathcal{N}(0,1)$
\If {Metropolis algorithm} 
\State $r \sim \mathrm{U}_{[0,1[}$ 
\If {$r\leq \dfrac{\sum_{\lambda\in\varLambda} \exp\left[A(\lambda) - \pot(\lambda,\trc^k)\right]}{\sum_{\lambda\in\varLambda}\varLambda \exp \left[A(\lambda) -\pot (\lambda,\brc^k_{m-1}) \right] }$} 
   \State $\brc^k_m\gets \trc^k$
\Else 
   \State $\brc^k_m\gets \brc^k_{m-1}$
\EndIf
\ElsIf {Overdamped Langevin dynamics}
   \State $\brc^k_m = \trc^k + \tau \nabla_\brc \ln \sum_{\lambda\in\varLambda} \exp\left[A(\lambda) - \pot(\lambda,\trc^k)\right] $ 
\EndIf
\EndFor
\For {$\lambda\in \varLambda$}
\State $\bpi_{A}^\lambda \big( \brc_m^k \big) = \dfrac{\exp\left[A(\lambda) - \pot(\lambda,\trc^k)\right]}{\sum_{\varlambda\in\varLambda} \exp\left[A_{m-1}(\varlambda) - \pot(\varlambda,\trc^k)\right]}$
\State
\State $C(\lambda) \gets C(\lambda) + \sum_{k=1}^K \obs\big(\brc^k_m\big) \bpi_{\wpmf}^\lambda \big( \brc_m^k \big)$
\State $D(\lambda) \gets D(\lambda) + \sum_{k=1}^K  \bpi_{\wpmf}^\lambda \big( \brc_m^k \big)$
\State $\dspi^\lambda_{A}(\brc_m^k) = \exp\left[-A(\lambda) \right] \bpi_{A}^\lambda \big( \brc_m^k \big) $
\EndFor
\State $\tilde{C} \gets \tilde{C} + \sum_{k=1}^K \sum_{\lambda\in\varLambda} \dspi^\lambda_{A}(\brc_m^k) \obs\big(\lambda,\brc^k_m\big)$
\State $\dspi_{A}(\brc_m^k) = \sum_{\lambda\in\varLambda}   \dspi^\lambda_{A}(\brc_m^k) \obs\big(\lambda,\brc^k_m\big)$
\State $\tilde{D} \gets  \tilde{D} + \sum_{k=1}^K \dspi_{A}(\brc_m^k)$
\EndFor
\For {$\lambda\in \varLambda$} 
\State $\estimator{M}{\Pi}(\obs |\lambda ) = {C(\lambda) }/{D(\lambda)}$ \Comment Estimate of $\e\left[\obs|\lambda \right]$
\EndFor
\State $\estimator{M}{\Pi}(\obs) = {\tilde{C}}\slash{\tilde{D}}$ \Comment Estimate of $\e\left[\obs\right]$
\end{algorithmic}
\end{algorithm}
For both algorithms, we used $K=12$ replicas of the system. The adaptive or time-homogeneous biasing potential and its gradient are shared by all the replicas which are propagated using $K$ distinct cores on a parallel computer architecture. \\


\begin{thebibliography}{10}%
\makeatletter
\providecommand \@ifxundefined [1]{%
 \ifx #1\undefined \expandafter \@firstoftwo
 \else \expandafter \@secondoftwo
\fi
}%
\providecommand \@ifnum [1]{%
 \ifnum #1\expandafter \@firstoftwo
 \else \expandafter \@secondoftwo
\fi
}%
\providecommand \enquote [1]{``#1''}%
\providecommand \bibnamefont  [1]{#1}%
\providecommand \bibfnamefont [1]{#1}%
\providecommand \citenamefont [1]{#1}%
\providecommand\href[0]{\@sanitize\@href}%
\providecommand\@href[1]{\endgroup\@@startlink{#1}\endgroup\@@href}%
\providecommand\@@href[1]{#1\@@endlink}%
\providecommand \@sanitize [0]{\begingroup\catcode`\&12\catcode`\#12\relax}%
\@ifxundefined \pdfoutput {\@firstoftwo}{%
 \@ifnum{\z@=\pdfoutput}{\@firstoftwo}{\@secondoftwo}%
}{%
 \providecommand\@@startlink[1]{\leavevmode}%
 \providecommand\@@endlink[0]{}%
}{%
 \providecommand\@@startlink[1]{%
  \leavevmode
  \pdfstartlink
   attr{/Border[0 0 1 ]/H/I/C[0 1 1]}%
   user{/Subtype/Link/A<</Type/Action/S/URI/URI(#1)>>}%
  \relax
 }%
 \providecommand\@@endlink[0]{\pdfendlink}%
}%
\providecommand \url  [0]{\begingroup\@sanitize \@url }%
\providecommand \@url [1]{\endgroup\@href {#1}{\urlprefix}}%
\providecommand \urlprefix [0]{URL }%
\providecommand \Eprint[0]{\href }%
\@ifxundefined \urlstyle {%
  \providecommand \doi [1]{doi:\discretionary{}{}{}#1}%
}{%
  \providecommand \doi [0]{doi:\discretionary{}{}{}\begingroup
  \urlstyle{rm}\Url }%
}%
\providecommand \doibase [0]{http://dx.doi.org/}%
\providecommand \Doi[1]{\href{\doibase#1}}%
\providecommand \selectlanguage [0]{\@gobble}%
\providecommand \bibinfo [0]{\@secondoftwo}%
\providecommand \bibfield [0]{\@secondoftwo}%
\providecommand \translation [1]{[#1]}%
\providecommand \BibitemOpen[0]{}%
\providecommand \bibitemStop [0]{}%
\providecommand \bibitemNoStop [0]{.\EOS\space}%
\providecommand \EOS [0]{\spacefactor3000\relax}%
\providecommand \BibitemShut [1]{\csname bibitem#1\endcsname}%
\bibitem{chandler:1987}%
  \BibitemOpen
  \bibfield{author}{%
  \bibinfo {author} {\bibfnamefont{D.}~\bibnamefont{Chandler}},\ }%
  \emph{\bibinfo {title} {Introduction to modern statistical mechanics}}\
  (\bibinfo {publisher} {Oxford University Press},\ \bibinfo {year}
  {1987})\BibitemShut{NoStop}%
\bibitem{frenkel:2001}%
  \BibitemOpen
  \bibfield{author}{%
  \bibinfo {author} {\bibfnamefont{D.}~\bibnamefont{Frenkel}}\ and\ \bibinfo
  {author} {\bibfnamefont{B.}~\bibnamefont{Smit}},\ }%
  \emph{\bibinfo {title} {Understanding molecular simulation: from algorithms
  to applications}}\ (\bibinfo {publisher} {Academic Press},\ \bibinfo {year}
  {2001})\BibitemShut{NoStop}%
\bibitem{wales:2003}%
  \BibitemOpen
  \bibfield{author}{%
  \bibinfo {author} {\bibfnamefont{D.}~\bibnamefont{Wales}},\ }%
  \emph{\bibinfo {title} {Energy Landscapes}},\ Cambridge Molecular Science\
  (\bibinfo {publisher} {Cambridge University Press},\ \bibinfo {address}
  {Cambridge, UK},\ \bibinfo {year} {2003})\BibitemShut{NoStop}%
\bibitem{lelievre:2010}%
  \BibitemOpen
  \bibfield{author}{%
  \bibinfo {author} {\bibfnamefont{T.}~\bibnamefont{Lelièvre}}, \bibinfo
  {author} {\bibfnamefont{M.}~\bibnamefont{Rousset}},\ and\ \bibinfo {author}
  {\bibfnamefont{G.}~\bibnamefont{Stoltz}},\ }%
  \emph{\bibinfo {title} {Free energy computations: a mathematical
  perspective}}\ (\bibinfo {publisher} {Imperial College Press},\ \bibinfo
  {address} {London},\ \bibinfo {year} {2010})\BibitemShut{NoStop}%
\bibitem{athenes:2014}%
  \BibitemOpen
  \bibfield{author}{%
  \bibinfo {author} {\bibfnamefont{M.}~\bibnamefont{Ath\`enes}}\ and\ \bibinfo
  {author} {\bibfnamefont{V.~V.}\ \bibnamefont{Bulatov}},\ }%
  \bibfield{journal}{%
  \bibinfo {journal} {Physical Review Letters}\ }%
  \textbf{\bibinfo {volume} {113}},\ \bibinfo {pages} {230601} (\bibinfo {year}
  {2014})\BibitemShut{NoStop}%
\bibitem{lyubartsev:1992}%
  \BibitemOpen
  \bibfield{author}{%
  \bibinfo {author} {\bibfnamefont{A.~P.}\ \bibnamefont{Lyubartsev}}, \bibinfo
  {author} {\bibfnamefont{A.~A.}\ \bibnamefont{Martsinovski}}, \bibinfo
  {author} {\bibfnamefont{S.~V.}\ \bibnamefont{Shevkunov}},\ and\ \bibinfo
  {author} {\bibfnamefont{P.~N.}\ \bibnamefont{Vorontsov-Velyaminov}},\ }%
  \bibfield{journal}{%
  \bibinfo {journal} {The Journal of Chemical Physics}\ }%
  \textbf{\bibinfo {volume} {96}},\ \bibinfo {pages} {1776} (\bibinfo {year}
  {1992})\BibitemShut{NoStop}%
\bibitem{burov:2006}%
  \BibitemOpen
  \bibfield{author}{%
  \bibinfo {author} {\bibfnamefont{S.~V.}\ \bibnamefont{Burov}}, \bibinfo
  {author} {\bibfnamefont{P.~N.}\ \bibnamefont{Vorontsov-Velyaminov}},\ and\
  \bibinfo {author} {\bibfnamefont{E.~M.}\ \bibnamefont{Piotrovskaya}},\ }%
  \bibfield{journal}{%
  \bibinfo {journal} {Molecular Simulation}\ }%
  \textbf{\bibinfo {volume} {32}},\ \bibinfo {pages} {437} (\bibinfo {year}
  {2006})\BibitemShut{NoStop}%
\bibitem{marinari:1992}%
  \BibitemOpen
  \bibfield{author}{%
  \bibinfo {author} {\bibfnamefont{E.}~\bibnamefont{Marinari}}\ and\ \bibinfo
  {author} {\bibfnamefont{G.}~\bibnamefont{Parisi}},\ }%
  \bibfield{journal}{%
  \bibinfo {journal} {EPL (Europhysics Letters)}\ }%
  \textbf{\bibinfo {volume} {19}},\ \bibinfo {pages} {451} (\bibinfo {year}
  {1992})\BibitemShut{NoStop}%
\bibitem{geyer:1995}%
  \BibitemOpen
  \bibfield{author}{%
  \bibinfo {author} {\bibfnamefont{C.~J.}\ \bibnamefont{Geyer}}\ and\ \bibinfo
  {author} {\bibfnamefont{E.~A.}\ \bibnamefont{Thompson}},\ }%
  \bibfield{journal}{%
  \bibinfo {journal} {Journal of the American Statistical Association}\ }%
  \textbf{\bibinfo {volume} {90}},\ \bibinfo {pages} {909} (\bibinfo {year}
  {1995})\BibitemShut{NoStop}%
\bibitem{park:2007}%
  \BibitemOpen
  \bibfield{author}{%
  \bibinfo {author} {\bibfnamefont{S.}~\bibnamefont{Park}}\ and\ \bibinfo
  {author} {\bibfnamefont{V.~S.}\ \bibnamefont{Pande}},\ }%
  \bibfield{journal}{%
  \bibinfo {journal} {Phys. Rev. E}\ }%
  \textbf{\bibinfo {volume} {76}},\ \bibinfo {pages} {016703} (\bibinfo {year}
  {2007})\BibitemShut{NoStop}%
\bibitem{li:2007}%
  \BibitemOpen
  \bibfield{author}{%
  \bibinfo {author} {\bibfnamefont{H.}~\bibnamefont{Li}}, \bibinfo {author}
  {\bibfnamefont{M.}~\bibnamefont{Fajer}},\ and\ \bibinfo {author}
  {\bibfnamefont{W.}~\bibnamefont{Yang}},\ }%
  \bibfield{journal}{%
  \bibinfo {journal} {The Journal of Chemical Physics}\ }%
  \textbf{\bibinfo {volume} {126}},\ \bibinfo {eid} {024106} (\bibinfo {year}
  {2007})\BibitemShut{NoStop}%
\bibitem{cuendet:2014}%
  \BibitemOpen
  \bibfield{author}{%
  \bibinfo {author} {\bibfnamefont{M.~A.}\ \bibnamefont{Cuendet}}\ and\
  \bibinfo {author} {\bibfnamefont{M.~E.}\ \bibnamefont{Tuckerman}},\ }%
  \bibfield{journal}{%
  \bibinfo {journal} {Journal of Chemical Theory and Computation}\ }%
  \textbf{\bibinfo {volume} {10}},\ \bibinfo {pages} {2975} (\bibinfo {year}
  {2014})\BibitemShut{NoStop}%
\bibitem{iannuzzi:2003}%
  \BibitemOpen
  \bibfield{author}{%
  \bibinfo {author} {\bibfnamefont{M.}~\bibnamefont{Iannuzzi}}, \bibinfo
  {author} {\bibfnamefont{A.}~\bibnamefont{Laio}},\ and\ \bibinfo {author}
  {\bibfnamefont{M.}~\bibnamefont{Parrinello}},\ }%
  \bibfield{journal}{%
  \bibinfo {journal} {Phys. Rev. Lett.}\ }%
  \textbf{\bibinfo {volume} {90}},\ \bibinfo {pages} {238302} (\bibinfo {year}
  {2003})\BibitemShut{NoStop}%
\bibitem{laio:2008}%
  \BibitemOpen
  \bibfield{author}{%
  \bibinfo {author} {\bibfnamefont{A.}~\bibnamefont{Laio}}\ and\ \bibinfo
  {author} {\bibfnamefont{F.~L.}\ \bibnamefont{Gervasio}},\ }%
  \bibfield{journal}{%
  \bibinfo {journal} {Reports on Progress in Physics}\ }%
  \textbf{\bibinfo {volume} {71}},\ \bibinfo {pages} {126601} (\bibinfo {year}
  {2008})\BibitemShut{NoStop}%
\bibitem{iba:2001}%
  \BibitemOpen
  \bibfield{author}{%
  \bibinfo {author} {\bibfnamefont{Y.}~\bibnamefont{Iba}},\ }%
  \bibfield{journal}{%
  \bibinfo {journal} {International Journal of Modern Physics C}\ }%
  \textbf{\bibinfo {volume} {12}},\ \bibinfo {pages} {623} (\bibinfo {year}
  {2001})\BibitemShut{NoStop}%
\bibitem{mitsutake:2001}%
  \BibitemOpen
  \bibfield{author}{%
  \bibinfo {author} {\bibfnamefont{A.}~\bibnamefont{Mitsutake}}, \bibinfo
  {author} {\bibfnamefont{Y.}~\bibnamefont{Sugita}},\ and\ \bibinfo {author}
  {\bibfnamefont{Y.}~\bibnamefont{Okamoto}},\ }%
  \bibfield{journal}{%
  \bibinfo {journal} {Peptide Science}\ }%
  \textbf{\bibinfo {volume} {60}},\ \bibinfo {pages} {96} (\bibinfo {year}
  {2001})\BibitemShut{NoStop}%
\bibitem{chodera:2007}%
  \BibitemOpen
  \bibfield{author}{%
  \bibinfo {author} {\bibfnamefont{J.~D.}\ \bibnamefont{Chodera}}, \bibinfo
  {author} {\bibfnamefont{W.~C.}\ \bibnamefont{Swope}}, \bibinfo {author}
  {\bibfnamefont{J.~W.}\ \bibnamefont{Pitera}}, \bibinfo {author}
  {\bibfnamefont{C.}~\bibnamefont{Seok}},\ and\ \bibinfo {author}
  {\bibfnamefont{K.~A.}\ \bibnamefont{Dill}},\ }%
  \bibfield{journal}{%
  \bibinfo {journal} {Journal of Chemical Theory and Computation}\ }%
  \textbf{\bibinfo {volume} {3}},\ \bibinfo {pages} {26} (\bibinfo {year}
  {2007})\BibitemShut{NoStop}%
\bibitem{park:2008}%
  \BibitemOpen
  \bibfield{author}{%
  \bibinfo {author} {\bibfnamefont{S.}~\bibnamefont{Park}},\ }%
  \bibfield{journal}{%
  \bibinfo {journal} {Phys. Rev. E}\ }%
  \textbf{\bibinfo {volume} {77}},\ \bibinfo {pages} {016709} (\bibinfo {year}
  {2008})\BibitemShut{NoStop}%
\bibitem{chodera:2011}%
  \BibitemOpen
  \bibfield{author}{%
  \bibinfo {author} {\bibfnamefont{J.~D.}\ \bibnamefont{Chodera}}\ and\
  \bibinfo {author} {\bibfnamefont{M.~R.}\ \bibnamefont{Shirts}},\ }%
  \bibfield{journal}{%
  \bibinfo {journal} {The Journal of Chemical Physics}\ }%
  \textbf{\bibinfo {volume} {135}},\ \bibinfo {eid} {194110} (\bibinfo {year}
  {2011})\BibitemShut{NoStop}%
\bibitem{cao:2014}%
  \BibitemOpen
  \bibfield{author}{%
  \bibinfo {author} {\bibfnamefont{L.}~\bibnamefont{Cao}}, \bibinfo {author}
  {\bibfnamefont{G.}~\bibnamefont{Stoltz}}, \bibinfo {author}
  {\bibfnamefont{T.}~\bibnamefont{Lelièvre}}, \bibinfo {author}
  {\bibfnamefont{M.-C.}\ \bibnamefont{Marinica}},\ and\ \bibinfo {author}
  {\bibfnamefont{M.}~\bibnamefont{Athènes}},\ }%
  \bibfield{journal}{%
  \bibinfo {journal} {The Journal of Chemical Physics}\ }%
  \textbf{\bibinfo {volume} {140}},\ \bibinfo {eid} {104108} (\bibinfo {year}
  {2014})\BibitemShut{NoStop}%
\bibitem{terrier:2015}%
  \BibitemOpen
  \bibfield{author}{%
  \bibinfo {author} {\bibfnamefont{P.}~\bibnamefont{Terrier}}, \bibinfo
  {author} {\bibfnamefont{M.-C.}\ \bibnamefont{Marinica}},\ and\ \bibinfo
  {author} {\bibfnamefont{M.}~\bibnamefont{Athènes}},\ }%
  \bibfield{journal}{%
  \bibinfo {journal} {The Journal of Chemical Physics}\ }%
  \textbf{\bibinfo {volume} {143}},\ \bibinfo {eid} {134121} (\bibinfo {year}
  {2015})\BibitemShut{NoStop}%
\bibitem{lesage:2016}%
  \BibitemOpen
  \bibfield{author}{%
  \bibinfo {author} {\bibfnamefont{A.}~\bibnamefont{Lesage}}, \bibinfo {author}
  {\bibfnamefont{T.}~\bibnamefont{Lelièvre}}, \bibinfo {author}
  {\bibfnamefont{G.}~\bibnamefont{Stoltz}},\ and\ \bibinfo {author}
  {\bibfnamefont{J.}~\bibnamefont{Hénin}},\ }%
  \bibfield{journal}{%
  \bibinfo {journal} {The Journal of Physical Chemistry B}\ }%
  \textbf{\bibinfo {volume} {121}},\ \bibinfo {pages} {3676} (\bibinfo {year}
  {2017})\BibitemShut{NoStop}%
\bibitem{bennett:1976}%
  \BibitemOpen
  \bibfield{author}{%
  \bibinfo {author} {\bibfnamefont{C.~H.}\ \bibnamefont{Bennett}},\ }%
  \bibfield{journal}{%
  \bibinfo {journal} {Journal of Computational Physics}\ }%
  \textbf{\bibinfo {volume} {22}},\ \bibinfo {pages} {245 } (\bibinfo {year}
  {1976})\BibitemShut{NoStop}%
\bibitem{swendsen:1986}%
  \BibitemOpen
  \bibfield{author}{%
  \bibinfo {author} {\bibfnamefont{R.~H.}\ \bibnamefont{Swendsen}}\ and\
  \bibinfo {author} {\bibfnamefont{J.-S.}\ \bibnamefont{Wang}},\ }%
  \bibfield{journal}{%
  \bibinfo {journal} {Phys. Rev. Lett.}\ }%
  \textbf{\bibinfo {volume} {57}},\ \bibinfo {pages} {2607} (\bibinfo {year}
  {1986})\BibitemShut{NoStop}%
\bibitem{ferrenberg:1989}%
  \BibitemOpen
  \bibfield{author}{%
  \bibinfo {author} {\bibfnamefont{A.~M.}\ \bibnamefont{Ferrenberg}}\ and\
  \bibinfo {author} {\bibfnamefont{R.~H.}\ \bibnamefont{Swendsen}},\ }%
  \bibfield{journal}{%
  \bibinfo {journal} {Phys. Rev. Lett.}\ }%
  \textbf{\bibinfo {volume} {63}},\ \bibinfo {pages} {1195} (\bibinfo {year}
  {1989})\BibitemShut{NoStop}%
\bibitem{geyer:1994}%
  \BibitemOpen
  \bibfield{author}{%
  \bibinfo {author} {\bibfnamefont{C.}~\bibnamefont{Geyer}},\ }%
  \enquote{\bibinfo {title} {Estimating normalizing constants and reweighting
  mixtures},}\ \bibinfo {type} {Technical Report}\ \bibinfo {number} {568}\
  (\bibinfo {institution} {School of Statistics University of Minnesota},\
  \bibinfo {year} {1994})\BibitemShut{NoStop}%
\bibitem{gelman:1998}%
  \BibitemOpen
  \bibfield{author}{%
  \bibinfo {author} {\bibfnamefont{A.}~\bibnamefont{Gelman}}\ and\ \bibinfo
  {author} {\bibfnamefont{X.-L.}\ \bibnamefont{Meng}},\ }%
  \bibfield{journal}{%
  \Doi{10.1214/ss/1028905934}{\bibinfo {journal} {Statist. Sci.}}\ }%
  \textbf{\bibinfo {volume} {13}},\ \bibinfo {pages} {163} (\bibinfo {month}
  {05}\ \bibinfo {year} {1998})\BibitemShut{NoStop}%
\bibitem{shirts:2008}%
  \BibitemOpen
  \bibfield{author}{%
  \bibinfo {author} {\bibfnamefont{M.~R.}\ \bibnamefont{Shirts}}\ and\ \bibinfo
  {author} {\bibfnamefont{J.~D.}\ \bibnamefont{Chodera}},\ }%
  \bibfield{journal}{%
  \bibinfo {journal} {The Journal of Chemical Physics}\ }%
  \textbf{\bibinfo {volume} {129}},\ \bibinfo {eid} {124105} (\bibinfo {year}
  {2008})\BibitemShut{NoStop}%
\bibitem{tan:2012}%
  \BibitemOpen
  \bibfield{author}{%
  \bibinfo {author} {\bibfnamefont{Z.}~\bibnamefont{Tan}}, \bibinfo {author}
  {\bibfnamefont{E.}~\bibnamefont{Gallicchio}}, \bibinfo {author}
  {\bibfnamefont{M.}~\bibnamefont{Lapelosa}},\ and\ \bibinfo {author}
  {\bibfnamefont{R.~M.}\ \bibnamefont{Levy}},\ }%
  \bibfield{journal}{%
  \bibinfo {journal} {The Journal of Chemical Physics}\ }%
  \textbf{\bibinfo {volume} {136}},\ \bibinfo {pages} {144102} (\bibinfo {year}
  {2012})\BibitemShut{NoStop}%
\bibitem{kong:2003}%
  \BibitemOpen
  \bibfield{author}{%
  \bibinfo {author} {\bibfnamefont{A.}~\bibnamefont{Kong}}, \bibinfo {author}
  {\bibfnamefont{P.}~\bibnamefont{McCullagh}}, \bibinfo {author}
  {\bibfnamefont{X.-L.}\ \bibnamefont{Meng}}, \bibinfo {author}
  {\bibfnamefont{D.}~\bibnamefont{Nicolae}},\ and\ \bibinfo {author}
  {\bibfnamefont{Z.}~\bibnamefont{Tan}},\ }%
  \bibfield{journal}{%
  \bibinfo {journal} {Journal of the Royal Statistical Society: Series B
  (Statistical Methodology)}\ }%
  \textbf{\bibinfo {volume} {65}},\ \bibinfo {pages} {585} (\bibinfo {year}
  {2003})\BibitemShut{NoStop}%
\bibitem{tan:2004}%
  \BibitemOpen
  \bibfield{author}{%
  \bibinfo {author} {\bibfnamefont{Z.}~\bibnamefont{Tan}},\ }%
  \bibfield{journal}{%
  \bibinfo {journal} {Journal of the American Statistical Association}\ }%
  \textbf{\bibinfo {volume} {99}},\ \bibinfo {pages} {1027} (\bibinfo {year}
  {2004})\BibitemShut{NoStop}%
\bibitem{berg:1992}%
  \BibitemOpen
  \bibfield{author}{%
  \bibinfo {author} {\bibfnamefont{B.~A.}\ \bibnamefont{Berg}}\ and\ \bibinfo
  {author} {\bibfnamefont{T.}~\bibnamefont{Neuhaus}},\ }%
  \bibfield{journal}{%
  \bibinfo {journal} {Phys. Rev. Lett.}\ }%
  \textbf{\bibinfo {volume} {68}},\ \bibinfo {pages} {9} (\bibinfo {year}
  {1992})\BibitemShut{NoStop}%
\bibitem{wang:2001}%
  \BibitemOpen
  \bibfield{author}{%
  \bibinfo {author} {\bibfnamefont{F.}~\bibnamefont{Wang}}\ and\ \bibinfo
  {author} {\bibfnamefont{D.~P.}\ \bibnamefont{Landau}},\ }%
  \bibfield{journal}{%
  \bibinfo {journal} {Phys. Rev. Lett.}\ }%
  \textbf{\bibinfo {volume} {86}},\ \bibinfo {pages} {2050} (\bibinfo {year}
  {2001})\BibitemShut{NoStop}%
\bibitem{laio:2002}%
  \BibitemOpen
  \bibfield{author}{%
  \bibinfo {author} {\bibfnamefont{A.}~\bibnamefont{Laio}}\ and\ \bibinfo
  {author} {\bibfnamefont{M.}~\bibnamefont{Parrinello}},\ }%
  \bibfield{journal}{%
  \bibinfo {journal} {Proceedings of the National Academy of Sciences}\ }%
  \textbf{\bibinfo {volume} {99}},\ \bibinfo {pages} {12562} (\bibinfo {year}
  {2002})\BibitemShut{NoStop}%
\bibitem{junghans:2014}%
  \BibitemOpen
  \bibfield{author}{%
  \bibinfo {author} {\bibfnamefont{C.}~\bibnamefont{Junghans}}, \bibinfo
  {author} {\bibfnamefont{D.}~\bibnamefont{Perez}},\ and\ \bibinfo {author}
  {\bibfnamefont{T.}~\bibnamefont{Vogel}},\ }%
  \bibfield{journal}{%
  \bibinfo {journal} {Journal of Chemical Theory and Computation}\ }%
  \textbf{\bibinfo {volume} {10}},\ \bibinfo {pages} {1843} (\bibinfo {year}
  {2014})\BibitemShut{NoStop}%
\bibitem{darve:2001}%
  \BibitemOpen
  \bibfield{author}{%
  \bibinfo {author} {\bibfnamefont{E.}~\bibnamefont{Darve}}\ and\ \bibinfo
  {author} {\bibfnamefont{A.}~\bibnamefont{Pohorille}},\ }%
  \bibfield{journal}{%
  \bibinfo {journal} {The Journal of Chemical Physics}\ }%
  \textbf{\bibinfo {volume} {115}},\ \bibinfo {pages} {9169} (\bibinfo {year}
  {2001})\BibitemShut{NoStop}%
\bibitem{marsili:2006}%
  \BibitemOpen
  \bibfield{author}{%
  \bibinfo {author} {\bibfnamefont{S.}~\bibnamefont{Marsili}}, \bibinfo
  {author} {\bibfnamefont{A.}~\bibnamefont{Barducci}}, \bibinfo {author}
  {\bibfnamefont{R.}~\bibnamefont{Chelli}}, \bibinfo {author}
  {\bibfnamefont{P.}~\bibnamefont{Procacci}},\ and\ \bibinfo {author}
  {\bibfnamefont{V.}~\bibnamefont{Schettino}},\ }%
  \bibfield{journal}{%
  \bibinfo {journal} {The Journal of Physical Chemistry B}\ }%
  \textbf{\bibinfo {volume} {110}},\ \bibinfo {pages} {14011} (\bibinfo {year}
  {2006})\BibitemShut{NoStop}%
\bibitem{fort:2014}%
  \BibitemOpen
  \bibfield{author}{%
  \bibinfo {author} {\bibfnamefont{G.}~\bibnamefont{Fort}}, \bibinfo {author}
  {\bibfnamefont{B.}~\bibnamefont{Jourdain}}, \bibinfo {author}
  {\bibfnamefont{E.}~\bibnamefont{Kuhn}}, \bibinfo {author}
  {\bibfnamefont{T.}~\bibnamefont{Leli{\`e}vre}},\ and\ \bibinfo {author}
  {\bibfnamefont{G.}~\bibnamefont{Stoltz}},\ }%
  \bibfield{journal}{%
  \bibinfo {journal} {{AMRX Appl.Math.Res.Express}}\ }%
  \textbf{\bibinfo {volume} {2014}},\ \bibinfo {pages} {275} (\bibinfo {year}
  {2014})\BibitemShut{NoStop}%
\bibitem{brukhno:1996}%
  \BibitemOpen
  \bibfield{author}{%
  \bibinfo {author} {\bibfnamefont{A.~V.}\ \bibnamefont{Brukhno}}, \bibinfo
  {author} {\bibfnamefont{T.~V.}\ \bibnamefont{Kuznetsova}}, \bibinfo {author}
  {\bibfnamefont{A.~P.}\ \bibnamefont{Lyubartsev}},\ and\ \bibinfo {author}
  {\bibfnamefont{P.}~\bibnamefont{Vorontsov-Vel’yaminov}},\ }%
  \bibfield{journal}{%
  \bibinfo {journal} {Vysokomolekuliarnye soedineniia. Seriia A i Seriia B}\ }%
  \textbf{\bibinfo {volume} {38}},\ \bibinfo {pages} {77} (\bibinfo {year}
  {1996})\BibitemShut{NoStop}%
\bibitem{brukhno:1996b}%
  \BibitemOpen
  \bibfield{author}{%
  \bibinfo {author} {\bibfnamefont{A.}~\bibnamefont{Brukhno}}, \bibinfo
  {author} {\bibfnamefont{T.}~\bibnamefont{Kuznetsova}}, \bibinfo {author}
  {\bibfnamefont{A.}~\bibnamefont{Lyubartsev}},\ and\ \bibinfo {author}
  {\bibfnamefont{P.}~\bibnamefont{Vorontsov-Velyaminov}},\ }%
  \bibfield{journal}{%
  \bibinfo {journal} {Polym. Sci., Ser. A}\ }%
  \textbf{\bibinfo {volume} {38}},\ \bibinfo {pages} {64} (\bibinfo {year}
  {1996})\BibitemShut{NoStop}%
\bibitem{darve:2008}%
  \BibitemOpen
  \bibfield{author}{%
  \bibinfo {author} {\bibfnamefont{E.}~\bibnamefont{Darve}}, \bibinfo {author}
  {\bibfnamefont{D.}~\bibnamefont{Rodríguez-Gómez}},\ and\ \bibinfo {author}
  {\bibfnamefont{A.}~\bibnamefont{Pohorille}},\ }%
  \bibfield{journal}{%
  \bibinfo {journal} {The Journal of Chemical Physics}\ }%
  \textbf{\bibinfo {volume} {128}},\ \bibinfo {eid} {144120} (\bibinfo {year}
  {2008})\BibitemShut{NoStop}%
\bibitem{henin:2010}%
  \BibitemOpen
  \bibfield{author}{%
  \bibinfo {author} {\bibfnamefont{J.}~\bibnamefont{Hénin}}, \bibinfo {author}
  {\bibfnamefont{G.}~\bibnamefont{Fiorin}}, \bibinfo {author}
  {\bibfnamefont{C.}~\bibnamefont{Chipot}},\ and\ \bibinfo {author}
  {\bibfnamefont{M.~L.}\ \bibnamefont{Klein}},\ }%
  \bibfield{journal}{%
  \bibinfo {journal} {Journal of Chemical Theory and Computation}\ }%
  \textbf{\bibinfo {volume} {6}},\ \bibinfo {pages} {35} (\bibinfo {year}
  {2010})\BibitemShut{NoStop}%
\bibitem{lelievre:2008}%
  \BibitemOpen
  \bibfield{author}{%
  \bibinfo {author} {\bibfnamefont{T.}~\bibnamefont{Lelièvre}}, \bibinfo
  {author} {\bibfnamefont{M.}~\bibnamefont{Rousset}},\ and\ \bibinfo {author}
  {\bibfnamefont{G.}~\bibnamefont{Stoltz}},\ }%
  \bibfield{journal}{%
  \bibinfo {journal} {Nonlinearity}\ }%
  \textbf{\bibinfo {volume} {21}},\ \bibinfo {pages} {1155} (\bibinfo {year}
  {2008})\BibitemShut{NoStop}%
\bibitem{comer:2014}%
  \BibitemOpen
  \bibfield{author}{%
  \bibinfo {author} {\bibfnamefont{J.}~\bibnamefont{Comer}}, \bibinfo {author}
  {\bibfnamefont{J.~C.}\ \bibnamefont{Gumbart}}, \bibinfo {author}
  {\bibfnamefont{J.}~\bibnamefont{Hénin}}, \bibinfo {author}
  {\bibfnamefont{T.}~\bibnamefont{Lelièvre}}, \bibinfo {author}
  {\bibfnamefont{A.}~\bibnamefont{Pohorille}},\ and\ \bibinfo {author}
  {\bibfnamefont{C.}~\bibnamefont{Chipot}},\ }%
  \bibfield{journal}{%
  \bibinfo {journal} {The Journal of Physical Chemistry B}\ }%
  \textbf{\bibinfo {volume} {119}},\ \bibinfo {pages} {1129} (\bibinfo {year}
  {2015})\BibitemShut{NoStop}%
\bibitem{fort:2016}%
  \BibitemOpen
  \bibfield{author}{%
  \bibinfo {author} {\bibfnamefont{G.}~\bibnamefont{{Fort}}}, \bibinfo {author}
  {\bibfnamefont{B.}~\bibnamefont{{Jourdain}}}, \bibinfo {author}
  {\bibfnamefont{T.}~\bibnamefont{{Leli{\`e}vre}}},\ and\ \bibinfo {author}
  {\bibfnamefont{G.}~\bibnamefont{{Stoltz}}},\ }%
  \bibfield{journal}{%
  \bibinfo {journal} {ArXiv e-prints}}%
   (\bibinfo {month} {Oct.}\ \bibinfo {year} {2016}),\
  \Eprint{http://arxiv.org/abs/1610.09194}{arXiv:1610.09194}\BibitemShut{NoSto%
p}%
\bibitem{tan:2015}%
  \BibitemOpen
  \bibfield{author}{%
  \bibinfo {author} {\bibfnamefont{Z.}~\bibnamefont{Tan}},\ }%
  \bibfield{journal}{%
  \Doi{10.1080/10618600.2015.1113975}{\bibinfo {journal} {Journal of
  Computational and Graphical Statistics}}\ }%
  \textbf{\bibinfo {volume} {26}},\ \bibinfo {pages} {54} (\bibinfo {year}
  {2017}),\
  \Eprint{http://arxiv.org/abs/http://dx.doi.org/10.1080/10618600.2015.1113975%
}{http://dx.doi.org/10.1080/10618600.2015.1113975},\
  \url{http://dx.doi.org/10.1080/10618600.2015.1113975}\BibitemShut{NoStop}%
\bibitem{footnote1}%
  \BibitemOpen
  \bibinfo {note} {The ABP scheme in Ref.~\citenum{tan:2015} is constructed
  through Rao-Blackwellisation of the histogram binning estimator. The
  Rao-Blackwell-Kolmogorov theorem in statistics is related to conditioning in
  probability theory.}\BibitemShut{Stop}%
\bibitem{jourdain:2013}%
  \BibitemOpen
  \bibfield{author}{%
  \bibinfo {author} {\bibfnamefont{B.}~\bibnamefont{Jourdain}},\ }%
  \enquote{\bibinfo {title} {Probabilités et statistique},}\ \bibinfo
  {howpublished} {http://cermics.enpc.fr/$\sim$jourdain/probastat/poly.pdf}
  (\bibinfo {year} {2013})\BibitemShut{NoStop}%
\bibitem{mcbook}%
  \BibitemOpen
  \bibfield{author}{%
  \bibinfo {author} {\bibfnamefont{A.~B.}\ \bibnamefont{Owen}},\ }%
  \enquote{\bibinfo {title} {Monte carlo theory, methods and examples},}\
  \bibinfo {howpublished} {http://statweb.stanford.edu/$\sim$owen/mc} (\bibinfo
  {year} {2013})\BibitemShut{NoStop}%
\bibitem{pham:2012}%
  \BibitemOpen
  \bibfield{author}{%
  \bibinfo {author} {\bibfnamefont{T.~T.}\ \bibnamefont{Pham}}\ and\ \bibinfo
  {author} {\bibfnamefont{M.~R.}\ \bibnamefont{Shirts}},\ }%
  \bibfield{journal}{%
  \bibinfo {journal} {The Journal of Chemical Physics}\ }%
  \textbf{\bibinfo {volume} {136}},\ \bibinfo {pages} {124120} (\bibinfo {year}
  {2012})\BibitemShut{NoStop}%
\bibitem{zwanzig:1954}%
  \BibitemOpen
  \bibfield{author}{%
  \bibinfo {author} {\bibfnamefont{R.~W.}\ \bibnamefont{Zwanzig}},\ }%
  \bibfield{journal}{%
  \bibinfo {journal} {The Journal of Chemical Physics}\ }%
  \textbf{\bibinfo {volume} {22}},\ \bibinfo {pages} {1420} (\bibinfo {year}
  {1954})\BibitemShut{NoStop}%
\bibitem{torrie:1977}%
  \BibitemOpen
  \bibfield{author}{%
  \bibinfo {author} {\bibfnamefont{G.}~\bibnamefont{Torrie}}\ and\ \bibinfo
  {author} {\bibfnamefont{J.}~\bibnamefont{Valleau}},\ }%
  \bibfield{journal}{%
  \bibinfo {journal} {Journal of Computational Physics}\ }%
  \textbf{\bibinfo {volume} {23}},\ \bibinfo {pages} {187 } (\bibinfo {year}
  {1977})\BibitemShut{NoStop}%
\bibitem{steinhardt:1983}%
  \BibitemOpen
  \bibfield{author}{%
  \bibinfo {author} {\bibfnamefont{P.~J.}\ \bibnamefont{Steinhardt}}, \bibinfo
  {author} {\bibfnamefont{D.~R.}\ \bibnamefont{Nelson}},\ and\ \bibinfo
  {author} {\bibfnamefont{M.}~\bibnamefont{Ronchetti}},\ }%
  \bibfield{journal}{%
  \bibinfo {journal} {Phys. Rev. B}\ }%
  \textbf{\bibinfo {volume} {28}},\ \bibinfo {pages} {784} (\bibinfo {year}
  {1983})\BibitemShut{NoStop}%
\bibitem{frenkel:2004}%
  \BibitemOpen
  \bibfield{author}{%
  \bibinfo {author} {\bibfnamefont{D.}~\bibnamefont{Frenkel}},\ }%
  \bibfield{journal}{%
  \bibinfo {journal} {Proceedings of the National Academy of Sciences of the
  United States of America}\ }%
  \textbf{\bibinfo {volume} {101}},\ \bibinfo {pages} {17571} (\bibinfo {year}
  {2004})\BibitemShut{NoStop}%
\bibitem{delmas:2009}%
  \BibitemOpen
  \bibfield{author}{%
  \bibinfo {author} {\bibfnamefont{J.-F.}\ \bibnamefont{Delmas}}\ and\ \bibinfo
  {author} {\bibfnamefont{B.}~\bibnamefont{Jourdain}},\ }%
  \bibfield{journal}{%
  \bibinfo {journal} {J. Appl. Probab.}\ }%
  \textbf{\bibinfo {volume} {46}},\ \bibinfo {pages} {938} (\bibinfo {month}
  {12}\ \bibinfo {year} {2009})\BibitemShut{NoStop}%
\bibitem{lidmar:2012}%
  \BibitemOpen
  \bibfield{author}{%
  \bibinfo {author} {\bibfnamefont{J.}~\bibnamefont{Lidmar}},\ }%
  \bibfield{journal}{%
  \bibinfo {journal} {Phys. Rev. E}\ }%
  \textbf{\bibinfo {volume} {85}},\ \bibinfo {pages} {056708} (\bibinfo {month}
  {May}\ \bibinfo {year} {2012})\BibitemShut{NoStop}%
\bibitem{frenkel:2006}%
  \BibitemOpen
  \bibfield{author}{%
  \bibinfo {author} {\bibfnamefont{D.}~\bibnamefont{Frenkel}},\ }%
  \enquote{\bibinfo {title} {Waste-recycling monte carlo},}\ in\ \emph{\bibinfo
  {booktitle} {Computer Simulations in Condensed Matter Systems: From Materials
  to Chemical Biology Volume 1}},\ \bibinfo {editor} {edited by\ \bibinfo
  {editor} {\bibfnamefont{M.}~\bibnamefont{Ferrario}}, \bibinfo {editor}
  {\bibfnamefont{G.}~\bibnamefont{Ciccotti}},\ and\ \bibinfo {editor}
  {\bibfnamefont{K.}~\bibnamefont{Binder}}}\ (\bibinfo {publisher} {Springer
  Berlin Heidelberg},\ \bibinfo {address} {Berlin, Heidelberg},\ \bibinfo
  {year} {2006})\ pp.\ \bibinfo {pages} {127--137}\BibitemShut{NoStop}%
\bibitem{kirkwood:1935}%
  \BibitemOpen
  \bibfield{author}{%
  \bibinfo {author} {\bibfnamefont{J.~G.}\ \bibnamefont{Kirkwood}},\ }%
  \bibfield{journal}{%
  \bibinfo {journal} {The Journal of Chemical Physics}\ }%
  \textbf{\bibinfo {volume} {3}},\ \bibinfo {pages} {300} (\bibinfo {year}
  {1935})\BibitemShut{NoStop}%
\bibitem{athenes:2002b}%
  \BibitemOpen
  \bibfield{author}{%
  \bibinfo {author} {\bibfnamefont{M.}~\bibnamefont{Ath\`enes}},\ }%
  \bibfield{journal}{%
  \bibinfo {journal} {Phys. Rev. E}\ }%
  \textbf{\bibinfo {volume} {66}},\ \bibinfo {pages} {046705} (\bibinfo {year}
  {2002})\BibitemShut{NoStop}%
\bibitem{athenes:2004}%
  \BibitemOpen
  \bibfield{author}{%
  \bibinfo {author} {\bibfnamefont{M.}~\bibnamefont{Athènes}},\ }%
  \bibfield{journal}{%
  \bibinfo {journal} {The European Physical Journal B - Condensed Matter and
  Complex Systems}\ }%
  \textbf{\bibinfo {volume} {38}},\ \bibinfo {pages} {651} (\bibinfo {year}
  {2004})\BibitemShut{NoStop}%
\bibitem{adjanor:2005}%
  \BibitemOpen
  \bibfield{author}{%
  \bibinfo {author} {\bibfnamefont{G.}~\bibnamefont{Adjanor}}\ and\ \bibinfo
  {author} {\bibfnamefont{M.}~\bibnamefont{Athènes}},\ }%
  \bibfield{journal}{%
  \bibinfo {journal} {The Journal of Chemical Physics}\ }%
  \textbf{\bibinfo {volume} {123}},\ \bibinfo {eid} {234104} (\bibinfo {year}
  {2005})\BibitemShut{NoStop}%
\bibitem{athenes:2012}%
  \BibitemOpen
  \bibfield{author}{%
  \bibinfo {author} {\bibfnamefont{M.}~\bibnamefont{Athènes}}, \bibinfo
  {author} {\bibfnamefont{M.-C.}\ \bibnamefont{Marinica}},\ and\ \bibinfo
  {author} {\bibfnamefont{T.}~\bibnamefont{Jourdan}},\ }%
  \bibfield{journal}{%
  \bibinfo {journal} {The Journal of Chemical Physics}\ }%
  \textbf{\bibinfo {volume} {137}},\ \bibinfo {eid} {194107} (\bibinfo {year}
  {2012})\BibitemShut{NoStop}%
\bibitem{lindahl:2014}%
  \BibitemOpen
  \bibfield{author}{%
  \bibinfo {author} {\bibfnamefont{V.}~\bibnamefont{Lindahl}}, \bibinfo
  {author} {\bibfnamefont{J.}~\bibnamefont{Lidmar}},\ and\ \bibinfo {author}
  {\bibfnamefont{B.}~\bibnamefont{Hess}},\ }%
  \bibfield{journal}{%
  \bibinfo {journal} {The Journal of Chemical Physics}\ }%
  \textbf{\bibinfo {volume} {141}},\ \bibinfo {pages} {044110} (\bibinfo {year}
  {2014})\BibitemShut{NoStop}%
\bibitem{neirotti:2000a}%
  \BibitemOpen
  \bibfield{author}{%
  \bibinfo {author} {\bibfnamefont{J.~P.}\ \bibnamefont{Neirotti}}, \bibinfo
  {author} {\bibfnamefont{F.}~\bibnamefont{Calvo}}, \bibinfo {author}
  {\bibfnamefont{D.~L.}\ \bibnamefont{Freeman}},\ and\ \bibinfo {author}
  {\bibfnamefont{J.~D.}\ \bibnamefont{Doll}},\ }%
  \bibfield{journal}{%
  \bibinfo {journal} {The Journal of Chemical Physics}\ }%
  \textbf{\bibinfo {volume} {112}},\ \bibinfo {pages} {10340} (\bibinfo {year}
  {2000})\BibitemShut{NoStop}%
\bibitem{neirotti:2000b}%
  \BibitemOpen
  \bibfield{author}{%
  \bibinfo {author} {\bibfnamefont{F.}~\bibnamefont{Calvo}}, \bibinfo {author}
  {\bibfnamefont{J.~P.}\ \bibnamefont{Neirotti}}, \bibinfo {author}
  {\bibfnamefont{D.~L.}\ \bibnamefont{Freeman}},\ and\ \bibinfo {author}
  {\bibfnamefont{J.~D.}\ \bibnamefont{Doll}},\ }%
  \bibfield{journal}{%
  \bibinfo {journal} {The Journal of Chemical Physics}\ }%
  \textbf{\bibinfo {volume} {112}},\ \bibinfo {pages} {10350} (\bibinfo {year}
  {2000})\BibitemShut{NoStop}%
\end{thebibliography}
%

\end{document}